\def\mode{1}

\if 0\mode
    \documentclass{vldb}

	\makeatletter
	\let\eaddfnt\@undefined
	\makeatother
	\newfont{\eaddfnt}{phvr at 9pt}	%
	\def\email#1{{{\eaddfnt{\vskip 0pt#1}}}}
	\title{Optimal Algorithms for Ranked Enumeration of \\Answers to Full Conjunctive Queries}

    \vldbTitle{Optimal Algorithms for Ranked Enumeration of Answers to Full Conjunctive Queries}
    \vldbAuthors{Nikolaos Tziavelis, Deepak Ajwani, Wolfgang Gatterbauer, Mirek Riedewald and Xiaofeng Yang}
    \vldbDOI{https://doi.org/10.14778/3397230.3397250}
    \vldbVolume{13}
    \vldbNumber{9}
    \vldbYear{2020}

    \numberofauthors{5}
    \author{
    \alignauthor
    Nikolaos Tziavelis
	\orcidicon{0000-0001-8342-2177}\\
           \affaddr{Northeastern University}\\
           \affaddr{Boston, MA, USA}\\
           \email{ntziavelis@ccs.neu.edu}
    \alignauthor
    Deepak Ajwani
    \orcidicon{0000-0001-7269-4150}\\
           \affaddr{University College Dublin}\\
           \affaddr{Dublin, Ireland}\\
            \email{deepak.ajwani@ucd.ie}
    \alignauthor 
	Wolfgang Gatterbauer 
	\orcidicon{0000-0002-9614-0504}\\
           \affaddr{Northeastern University}\\
           \affaddr{Boston, MA, USA}\\
            \email{w.gatterbauer@northeastern.edu}
    \and  %
    \alignauthor Mirek Riedewald
	\orcidicon{0000-0002-6102-7472}\\
           \affaddr{Northeastern University}\\
           \affaddr{Boston, MA, USA}\\
            \email{m.riedewald@northeastern.edu}
    \alignauthor Xiaofeng Yang\titlenote{Work performed while PhD student at Northeastern University.}\hspace{-1.5mm}\orcidicon{0000-0002-4364-3458}\\
           \affaddr{VMware}\\
           \affaddr{Palo Alto, CA, USA}\\
            \email{alice.xiaofeng.yang@gmail.com}
    }

\else
    \documentclass[acmsmall,anonymous=false, natbib=false, screen=true]{acmart}

	\title[Optimal Algorithms for Ranked Enumeration]{Optimal Algorithms for Ranked Enumeration of \\Answers to Full Conjunctive Queries}
	
	\author{Nikolaos Tziavelis}
	\email{ntziavelis@ccs.neu.edu}
	\affiliation{%
	    \institution{Northeastern University}
	    \city{Boston}
	    \state{Massachusetts}
	    \country{USA}
	}

	\author{Deepak Ajwani}
	\email{deepak.ajwani@ucd.ie}
	\affiliation{%
	    \institution{University College Dublin}
	    \city{Dublin}
	    \country{Ireland}
	}

	\author{Wolfgang Gatterbauer}
	\email{w.gatterbauer@northeastern.edu}
	\affiliation{%
	    \institution{Northeastern University}
	    \city{Boston}
	    \state{Massachusetts}
	    \country{USA}
	}

	\author{Mirek Riedewald}
	\email{m.riedewald@northeastern.edu}
	\affiliation{%
	    \institution{Northeastern University}
	    \city{Boston}
	    \state{Massachusetts}
	    \country{USA}
	}

	\author{Xiaofeng Yang}
	\email{alice.xiaofeng.yang@gmail.com}
	\affiliation{%
	    \institution{Northeastern University}
	    \city{Boston}
	    \state{Massachusetts}
	    \country{USA}
	}
	
\fi

\usepackage{amsmath}
\usepackage{mathtools}
\usepackage{graphicx} %
\usepackage{url}
\usepackage{enumitem} %

\usepackage{stmaryrd} %

\usepackage{upgreek} %
\usepackage{bbding} %

\usepackage{listings} %

\usepackage{tabularx} %
\usepackage{booktabs} %
\usepackage{multirow}
\usepackage{hhline} %
\usepackage{diagbox}
\usepackage{pgfplotstable} %

\usepackage{xcolor,colortbl}    %

\usepackage{tikz} %
\usetikzlibrary{matrix}
\usetikzlibrary{calc}
\usetikzlibrary{math}
\usetikzlibrary{arrows.meta,positioning}
\usetikzlibrary{intersections, pgfplots.fillbetween}

\usepackage{easybmat} %

\usepackage{adjustbox}
\def\eox{\unskip\kern 10pt{\unitlength1pt\linethickness{.4pt}$\diamondsuit${}}} %

\usepackage{rotating}		%

\usepackage{xspace} %

\usepackage[bottom]{footmisc}   %

\usepackage{subcaption} %
\newcommand{\hide}[1]{}

\usepackage[ruled,noend,linesnumbered]{algorithm2e} %

\usepackage{setspace} %

\DontPrintSemicolon     %
\SetNlSty{}{}{}                %
\SetAlgoInsideSkip{smallskip}   %

\SetAlFnt{\small}			%
\SetAlCapFnt{\small}		%
\SetAlCapNameFnt{\small}

\hypersetup{                                                            %
	pdfpagemode=UseNone,               %
    pdfstartview=Fit,
    pdfpagelayout=SinglePage,       %
    bookmarks=false,
    bookmarksnumbered=true,
    bookmarksopen=false,             %
    pdftex=true,
    unicode,
    colorlinks=true,       %
    linkcolor=blue,          %
    citecolor=cyan,  %
    filecolor=magenta,      %
     urlcolor=blue           %
}

\usepackage[capitalise,nameinlink]{cleveref} %

\usepackage{aliascnt}  	

\newtheorem{theorem}{Theorem} %

\newaliascnt{corollary}{theorem}
\newtheorem{corollary}[corollary]{Corollary}
\aliascntresetthe{corollary}

\newaliascnt{example}{theorem}
\newtheorem{example}[example]{Example}
\aliascntresetthe{example}

\newaliascnt{definition}{theorem}
\newtheorem{definition}[definition]{Definition}
\aliascntresetthe{definition}

\newaliascnt{proposition}{theorem}
\newtheorem{proposition}[proposition]{Proposition}
\aliascntresetthe{proposition}

\newaliascnt{lemma}{theorem}

\aliascntresetthe{lemma}

\newaliascnt{conjecture}{theorem}

\aliascntresetthe{conjecture}

\newtheorem{questionW}{Question}
\newtheorem{resultW}{Result}

{\end{itshape}
}
\let\footnotesize\scriptsize
\if 1\mode
	\AtEndPreamble{%
	    \hypersetup{colorlinks,
	      linkcolor=purple,
	      citecolor=blue,
	      urlcolor=ACMDarkBlue,
	      filecolor=ACMDarkBlue}}
\fi

\if 0\mode
	\definecolor[named]{ACMDarkBlue}{cmyk}{1,0.58,0,0.21}
	\hypersetup{colorlinks,
	  linkcolor=purple,
	  citecolor=blue,
	  filecolor=ACMDarkBlue}
\fi

\if 0\mode
	\makeatletter 

	\renewcommand{\@opargbegintheorem}[3]{%
	    \parskip 0pt %
	    \trivlist
	    \item[%
	    	\hskip 10\p@										%
	        \hskip \labelsep
	        {\sc #1\ #2\             %
	   \setbox\@tempboxa\hbox{(#3)}  %
	        \ifdim \wd\@tempboxa>\z@ %
	            \hskip 0\p@\relax    %
	            \box\@tempboxa       %
	        \fi.}%
	    ]
	    \it
	}

	\def\@begintheorem#1#2{%
	    \parskip 0pt %
	    \trivlist
	    \item[%
	    	\hskip 10\p@										%
	        \hskip \labelsep
	        {{\sc #1}\hskip 5\p@\relax#2.}%
	    ]
	    \it
	}

	\makeatother
\fi

\usepackage{tcolorbox}		%
\tcbuselibrary{breakable,skins}		%

\tcbset{examplestyle/.style={
		enhanced jigsaw,	%
		colback=blue!10,	%
		colframe=blue!10,	%
		arc=0mm,
		boxrule=0pt,		%
		left=1mm,
		right=1mm,
		left skip= 0mm,  %
		right skip= 0mm, %
		top=1pt,		%
		bottom=1pt,		%
		breakable,		%
		parbox = false,		%
		before={\par\pagebreak[0]\vspace{1mm}\parindent=0pt},		
		after={\par\pagebreak[0]\vspace{1mm}\parindent=0pt},				
		bottomrule = 0mm,
		boxsep = 0mm,					%
		topsep at break=0pt,			%
		bottomsep at break=0pt,			%
		pad at break=0mm,
		pad before break=1mm,		
		pad after break=1mm,		
		bottomrule at break=0mm,
		toprule at break=0mm,		
		}}

\usepackage{footnote}

\tcolorboxenvironment{example}{examplestyle}
\newcommand{\resultbox}[1]{
\begin{tcolorbox}[
	enhanced jigsaw,		%
	colback=red!5,
	colframe=red!75!black,	
	arc=0mm,
	left skip=-1mm,
	right skip=-1mm,	
	left=0mm,
	topsep at break=1mm,			%
	right=0mm,
	top=0mm,
	bottom=0mm,		%
	breakable,		%
	parbox = false		%
]
\emph{#1}
\end{tcolorbox}
}

\makeatletter
\DeclareRobustCommand*\uell{\mathpalette\@uell\relax}
\newcommand*\@uell[2]{
  \setbox0=\hbox{$#1\ell$}
  \setbox1=\hbox{\rotatebox{10}{$#1\ell$}}
  \dimen0=\wd0 \advance\dimen0 by -\wd1 \divide\dimen0 by 2
  \mathord{\lower 0.1ex \hbox{\kern\dimen0\unhbox1\kern\dimen0}}
}
\makeatother

\usepackage{marginnote}

\newcommand{\smallsection}[1]{\vspace{2mm}\noindent\textbf{#1.}} %
\newcommand{\introparagraph}[1]{\textbf{#1.}} %

\renewcommand{\epsilon}{\varepsilon} %

\newcommand{\define}{:=} %

\newcommand{\datarule}{{\,:\!\!-\,}} %
\renewcommand{\vec}[1]{\boldsymbol{\mathbf{#1}}}

\newcommand{\N}{\mathbb{N}} %
\newcommand{\R}{{\mathbb{R}}} %
\newcommand{\bigO}{{\mathcal{O}}} %
\renewcommand{\O}{{\mathcal{O}}} %
\renewcommand{\P}{{\mathbb{P}}} %
\newcommand{\E}{\mathbb{E}}

\newcommand{\val}{\textup{\textrm{val}}}

\usepackage{scalerel}
\usepackage{tikz}
\usetikzlibrary{svg.path}

\definecolor{orcidlogocol}{HTML}{A6CE39}
\tikzset{
  orcidlogo/.pic={
    \fill[orcidlogocol] svg{M256,128c0,70.7-57.3,128-128,128C57.3,256,0,198.7,0,128C0,57.3,57.3,0,128,0C198.7,0,256,57.3,256,128z};
    \fill[white] svg{M86.3,186.2H70.9V79.1h15.4v48.4V186.2z}
                 svg{M108.9,79.1h41.6c39.6,0,57,28.3,57,53.6c0,27.5-21.5,53.6-56.8,53.6h-41.8V79.1z M124.3,172.4h24.5c34.9,0,42.9-26.5,42.9-39.7c0-21.5-13.7-39.7-43.7-39.7h-23.7V172.4z}
                 svg{M88.7,56.8c0,5.5-4.5,10.1-10.1,10.1c-5.6,0-10.1-4.6-10.1-10.1c0-5.6,4.5-10.1,10.1-10.1C84.2,46.7,88.7,51.3,88.7,56.8z};
  }
}

\RequirePackage{etoolbox}
\DeclareRobustCommand\orcidicon[1]{\href{https://orcid.org/#1}{\mbox{\scalerel*{
\begin{tikzpicture}[yscale=-1,transform shape]
\pic{orcidlogo};
\end{tikzpicture}
}{|}}}}

\newcommand{\new}[1]{#1}

\newcommand{\adom}{{\mathtt{ADom}}}

\newcommand{\aggr}{+}
\newcommand{\aggrsum}{\sum}

\newcommand{\Choices}{\mathtt{Choices}}

\newcommand{\Suc}{\mathtt{Succ}}
\newcommand{\stages}{\ell}

\newcommand{\sgiter}{i}

\newcommand{\sglim}{r}

\newcommand{\Dec}{E}
\newcommand{\DecR}{\mathbb{E}}
\newcommand{\Sset}{S}
\newcommand{\SsetR}{\mathbb{S}}
\newcommand{\prank}{j}
\newcommand{\best}{\mathtt{best}} %

\newcommand{\concat}{\circ}
\newcommand{\Ch}{C}
\newcommand{\parent}{pr}
\newcommand{\chnum}{\lambda}
\newcommand{\solW}{\pi}
\newcommand{\sol}{\Pi}
\newcommand{\weight}{w}	%
\newcommand{\tree}{\concat}	

\newcommand{\PSQL}{\textsc{PSQL}\xspace}
\newcommand{\NAIVE}{\textsc{Batch}\xspace}
\newcommand{\LAZY}{\textsc{Lazy}\xspace}
\newcommand{\EAGER}{\textsc{Eager}\xspace}

\newcommand{\MIN}{\textsc{All}\xspace}
\newcommand{\HEAP}{\textsc{Take2}\xspace}
\newcommand{\REDP}{\textsc{anyK-rec}\xspace}
\newcommand{\RPDP}{\textsc{anyK-part}\xspace}
\newcommand{\RECURSIVE}{\textsc{Recursive}\xspace}

\newcommand{\NPRR}{\textsc{NPRR}\xspace}
\newcommand{\PANDA}{\textsc{Panda}\xspace}
\newcommand{\GenJoin}{\textsc{Generic-Join}\xspace}

\usepackage{booktabs} %
\graphicspath{{./figs/}}
\usepackage{numprint}
\usepackage{multirow}
\usepackage{makecell}

\usepackage{balance}
\usepackage{tikz}
\usepackage{tikz,tikz-qtree,tikz-qtree-compat,tikz-3dplot}
\usetikzlibrary{shapes,decorations.markings,arrows.meta,fit}

\definecolor{dkgreen}{rgb}{0,0.6,0}
\newcommand{\cc}{olive}
\newcommand{\algocomment}[1]{\textcolor{\cc}{{//#1}}}

\newcommand{\TTF}{\ensuremath{\mathrm{TTF}}}

\newcommand{\TTL}{\ensuremath{\mathrm{TTL}}\xspace}
\newcommand{\TT}{\ensuremath{\mathrm{TT}}}
\newcommand{\Del}{\ensuremath{\mathrm{Delay}}}
\newcommand{\MEM}{\ensuremath{\mathrm{MEM}}}

\renewcommand{\H}{\mathcal{H}}
\renewcommand{\G}{\mathcal{G}}
\newcommand{\HD}{\textit{HD}}

\newcommand{\Cand}{\mathtt{Cand}} %

\newcommand{\Union}{\mathtt{Union}} %

\newcommand{\minweight}{min-weight-projection\xspace} 
\newcommand{\Minweight}{Min-weight-projection\xspace} 
\newcommand{\allweights}{all-weight-projection\xspace}
\newcommand{\Allweights}{All-weight-projection\xspace}

\newcommand{\varset}{\textup{\texttt{var}}}

\newcommand{\0}{\bar{0}}
\newcommand{\1}{\bar{1}}

\newcommand{\ghtw}{\textsf{ghw}\xspace}
\newcommand{\fhtw}{\textsf{fhw}\xspace}
\newcommand{\subw}{\textsf{subw}\xspace}

\newcommand{\lex}{\textrm{L}}

\newcommand{\BMM}{\textsc{BMM}\xspace}
\newcommand{\SBMM}{\textsc{sparseBMM}\xspace}
\newcommand{\TRIANGLE}{\textsc{Triangle}\xspace}
\newcommand{\HYPERCLQ}{\textsc{Hyperclique}\xspace}

\makeatletter
\newcommand{\markZwicky}[1][]{\pgfutil@ifnextchar({\mark@Zwicky{#1}}{\mark@Zwicky{#1}()}}
\def\mark@Zwicky#1(#2)#3{%
   \tikz[every Zwicky picture,#1]{%
     \node[every Zwicky node,draw=none,inner sep=+\z@,outer sep=+\z@] {#3};
     \def\tikz@Mark@name{#2}%
     \ifx\tikz@Mark@name\pgfutil@empty\else
       \tikzset{every Zwicky node/.append style={name={#2}}}%
     \fi
     \node[every Zwicky node,overlay] {\phantom{#3}};
   }%
}
\newcommand{\tikzZwicky}[1][]{%
  \def\tikz@Zwicky@args{#1}%
  \let\tikz@Zwicky@list\pgfutil@gobble
  \let\tikz@Zwicky@first\pgfutil@empty
  \pgfutil@ifnextchar(\tikz@Zwicky@collect\tikz@Zwicky@finish
}
\def\tikz@Zwicky@collect(#1){%
  \ifx\tikz@Zwicky@first\pgfutil@empty
    \edef\tikz@Zwicky@first{#1}%
  \else
    \edef\tikz@Zwicky@list{\tikz@Zwicky@list,#1}%
  \fi
  \pgfutil@ifnextchar(\tikz@Zwicky@collect\tikz@Zwicky@finish
}
\def\tikz@Zwicky@finish{%
  \tikz[remember picture,overlay]
    \draw[every Zwicky connector,/expanded=\tikz@Zwicky@args]
      (\tikz@Zwicky@first) [/expanded={@Zwicky@list/.list={\tikz@Zwicky@list}}] [every Zwicky connect finish/.try];
}
\pgfkeys{/expanded/.code={\edef\pgfkeys@temp{{#1}}\expandafter\pgfkeysalso\pgfkeys@temp}}
\makeatother
\tikzset{
  @Zwicky@list/.style={insert path={to[every Zwicky connector how/.try] (#1)}},
  every Zwicky picture/.style={
    baseline,
    remember picture,
  },
  every Zwicky node/.style={
    remember picture,
    anchor=base,
    inner sep=+2pt
  },
  every Zwicky connector/.style={
    ultra thick,
    red!80!black,
    draw opacity=.5,
    line cap=round,
    line join=round
  }
}

\usepackage[style=trad-abbrv,backend=bibtex,doi=true,isbn=false,url=true,eprint=false, maxbibnames=99]{biblatex} 
\begin{document}

\if 0\mode
    \maketitle
\fi

\begin{abstract} 
We study ranked enumeration of join-query results according to very general
orders defined by selective dioids.
Our main contribution is a framework for ranked enumeration over a class of
dynamic programming problems that generalizes seemingly
different problems that had been studied in isolation.
To this end, we extend classic algorithms that find the $k$-shortest paths
in a weighted graph.
For full conjunctive queries, including cyclic ones,
our approach is optimal in terms of the time to return the top result
and the delay between results.
These optimality properties are derived for the widely used notion of
data complexity, which treats query size as a constant.
By performing a careful cost analysis, we are able to
uncover a previously unknown trade-off between two incomparable 
enumeration approaches: one has lower complexity when the number of
returned results is small, the other when the number is very large.
We theoretically and empirically demonstrate the superiority of our
techniques over batch algorithms, which produce the full result and then sort it.
Our technique is not only faster for returning the first few
results, but on some inputs beats the batch algorithm even when all
results are produced.
\end{abstract}

\if 1\mode
    \maketitle
\fi

\section{Introduction}

Joins are an essential building block of queries in relational and graph databases,
and recent work on worst-case optimal joins for cyclic queries renewed interest in their
efficient evaluation~\cite{ngo2018worst}.
Part of the excitement stems from the fact that
conjunctive query (CQ) evaluation is equivalent to other key problems
such as constraint satisfaction~\cite{KOLAITIS2000302}
and hypergraph homomorphism~\cite{Friedgut1998}.
Similar to \cite{ngo2018worst}, we consider \emph{full conjunctive queries}, 
yet we are interested in 
\emph{ranked enumeration}, recently identified as an important open problem \cite{dagstuhl19enumeration}:
return output tuples in the order determined by a given ranking function.
Here success is measured not only in the time for total result computation,
but the main challenge lies in
\emph{returning the top-ranked result(s) as quickly as possible}.

We share this motivation with top-$k$ query evaluation~\cite{ilyas08survey},
which defines the importance of an output tuple based on the \emph{weights}
of its participating input tuples.
However, many top-$k$ approaches, including the famous
Threshold Algorithm~\cite{fagin03},
were developed for a {middleware-centric} cost model that charges
an algorithm only for accesses to external data sources,
but does not take into account the cost associated with
potentially huge intermediate results. We want optimality guarantees
in the standard RAM-model of computation
for (1) the time until the first result is returned
and (2) the delay between results.

\begin{example}[4-cycle query]\label{ex:4cycle}
Let $w$ be a function that returns the real-valued weight of a tuple and
consider 4-cycle query $Q_{C4}$
over $R_1(A_1, A_2)$, $R_2(A_2, A_3)$,
$R_3(A_3, A_4)$, and $R_4(A_4, A_1)$ with at most $n$ tuples each:
\begin{verbatim}
    SELECT   R1.A1, R2.A2, R3.A3, R4.A4, R1.W + R2.W + R3.W + R4.W as Weight
    FROM     R1, R2, R3, R4
    WHERE    R1.A2=R2.A2 AND R2.A3=R3.A3 AND 
             R3.A4=R4.A4 AND R4.A1=R1.A1
    ORDER BY Weight ASC
    LIMIT    k
\end{verbatim}
One can compute the full output with a worst-case optimal join algorithm
such as \NPRR~\cite{ngo2018worst} or \GenJoin~\cite{Ngo:2014:SSB:2590989.2590991} and then sort it.
Since the fractional edge cover number $\rho^{*}$ 
of $Q_{C4}$ is 2,
it takes $\O(n^2)$ just to produce the full output~\cite{AGM}.

On the other hand, the Boolean version of this query
(``Is there any 4-cycle?'') can be answered in $\O(n^{1.5})$~\cite{Marx:2013:THP:2555516.2535926}.
Our approach returns the top-ranked 4-cycle in $\O(n^{1.5})$ as well.
This is remarkable, given that determining the existence of a 4-cycle
appears easier than finding the top-ranked 4-cycle
(we can use the latter to answer the former).
After the top-ranked 4-cycle is found,
our approach continues to return the remaining results in ranked order with
``minimal'' delay.
\end{example}

We develop a theory of \emph{optimal ranked enumeration over full CQs}.
It reveals deeper relationships between recent work that
only partially addresses the problem we are considering: 
Putting aside the focus on twig patterns~\cite{chang15enumeration}
and subgraph isomorphism~\cite{yang2018any}, graph-pattern ranking techniques
can in principle be applied to conjunctive queries.
An unpublished paper \cite{deep19} that was developed concurrently with our work
offers a recursive solution for ranked enumeration. 
All this prior work raises
the question of how the approaches are related and whether they can be improved: 
Can time complexity of the top-$k$ algorithm by Chang et al.~\cite{chang15enumeration}
be improved for large $k$ and is it possible to extend it to give
{optimality guarantees for cyclic queries}?
For \cite{KimelfeldS2006,yang2018any}, how can the worst-case delay be reduced?
Is it possible to reduce the complexity of \cite{deep19} for returning the
first few results and can one close the asymptotic gap between the time
complexity for returning the top-ranked result and the complexity of the corresponding
Boolean query for simple cycles?

It is non-trivial to answer the above questions, because those approaches blend
various elements into monolithic solutions, sometimes reinventing the wheel in the
process.

\smallsection{Key contributions}
We identify and formally model \emph{the underlying structure of the ranked
enumeration problem} for conjunctive queries
and then solve it in a principled way:

(1) For CQs that are paths, we
\emph{identify and formalize the deeper common foundations
of problems that had been studied in isolation}: $k$-shortest path,
top-$k$ graph-pattern retrieval, and ranked enumeration over joins.
While interesting in its own right, uncovering those relationships enables us to propose
the first algorithms with \emph{optimal time complexity for ranked enumeration}
of the results of both cyclic and acyclic full CQs.
In particular, the top-ranked output tuple of an acyclic join query is returned in time
linear in input size.
For cyclic queries this complexity increases with the submodular
width (\subw) of the query~\cite{Marx:2013:THP:2555516.2535926}, 
which is currently the best known for Boolean queries. 
Delay between consecutive output tuples is logarithmic
in $k$.

(2) To achieve optimality, we make several technical
contributions. 
First, for path CQs we propose a new algorithm 
\HEAP 
\new{with lower delay given linear-time pre-processing}
than
all previous work but Eppstein's algorithm~\cite{eppstein1998finding},
whose practical performance is unknown.
\HEAP matches the latter
and has the added benefit that it can be generalized to arbitrary acyclic queries.\footnote{This generalization is unknown for Eppstein and it would be challenging
due to the complex nature of that algorithm.}
Second, to generalize $k$-shortest path algorithms to
arbitrary acyclic CQs, we introduce
\emph{ranked enumeration over Tree-based Dynamic Programming} (T-DP), a
variant of Non-Serial Dynamic Programming (NSDP)~\cite{bertele72nsdp}.
Third, we propose \emph{Union of T-DP problems} (UT-DP), 
a framework
for \emph{optimally incorporating in our approach all existing decompositions}
of a cyclic CQ into a union of trees. 
Thereby, any decomposition of a full CQ
that achieves optimality for the Boolean version of the query will
result in an optimal algorithm for ranked enumeration over
full CQs in our framework.

(3) Ranked enumeration over path CQs forms the backbone of our approach, therefore
we analyze all techniques for this problem in terms
of both \emph{data and query complexity}. 
This is complemented by the
\emph{first empirical study} that directly compares landmark results on
ranked enumeration from diverse domains such as
$k$-shortest paths, graph-pattern search, and CQs.
Our analysis reveals several interesting insights: 
($i$) In terms of time complexity the best \emph{Lawler-type}~\cite{lawler72}
approaches are asymptotically optimal for general inputs and dominate the
\emph{Recursive Enumeration Algorithm (REA)}~\cite{deep19,jimenez99shortest}.
($ii$) Since REA smartly reuses comparisons, there exist inputs for which it
produces the \emph{full ordered output} with lower time complexity than Lawler;
it is even faster than sorting! Our experiments verify this behavior
and suggest that Lawler-type approaches should be preferred for small $k$,
but REA for large $k$. 
Thus we are the first to not only propose different approaches, but also reveal that
\emph{neither dominates all others}, both in terms of asymptotic complexity and measured running time.
\new{
($iii$) Even though our new \HEAP algorithm needs asymptotically the same pre-processing and has lower delay than \LAZY~\cite{chang15enumeration}, its overall \emph{time-to-the-$k$'th result is the same} and we do not find it winning in our experiments.
}

\new{
This is the extended version of a paper appearing
in VLDB'20 \cite{tziavelis20vldb}.
}
Further information is available on the project web page at \url{https://northeastern-datalab.github.io/anyk/}.

\vspace{-1mm}
\section{Formal Setup}
\label{sec:def}

We use $\N_i^j$ to denote the set of natural numbers $\{ i, \ldots, j \}$.

\subsection{Conjunctive Queries (CQs)}
\label{sec:cq_def}

Our approach can be applied to any join query, including those
with theta-join conditions and projections, but we provide optimality results
only for \emph{full conjunctive queries (CQs)} with equi-joins~\cite{ngo2018worst}
and hence focus on them. 
A full CQ is a  first-order formula 
$Q(\vec x)= (g_1 \wedge \cdots \wedge g_\stages)$, written
$Q(\vec x) \datarule g_1(\vec x_1),\ldots, g_\stages(\vec x_\stages)$ 
in Datalog notation, where each atom $g_i$ represents a relation $R_i(\vec x_i)$ 
with different atoms possibly referring to the same physical relation, and 
$\vec x = \bigcup_i \vec x_i$ is a set of $m$ attributes.
\new{
An \emph{answer} or 
\emph{query result} or 
\emph{output tuple} is an assignment of the variables 
$\vec x$ to values from the domain of the database such that the formula is satisfied. 
}
The \emph{size} of the query $|Q|$ is the size of the formula.
We use $n$ to refer to the maximal cardinality
of any input relation referenced in $Q$.
Occurrence of the same variable in different atoms encodes an \emph{equi-join} condition. 
A CQ can be represented by a hypergraph with the variables as the nodes and the atoms
as the hyperedges; acyclicity of the query is defined in terms of the acyclicity
of the associated hypergraph~\cite{goodman1983gyo}.
In particular, we say that a query is acyclic when its hypergraph is alpha-acyclic\cite{baron16acyclic}.
This property can be verified efficiently in $\O(|Q|)$ by the
well-known GYO reduction \cite{graham80gyo,yu79gyo} which constructs a join tree.
A \emph{Boolean} CQ just asks for the satisfiability of the formula.
We use $Q^B$ to denote the Boolean version of $Q$.
A query with \emph{self-joins} has at least one relation appearing in more than one atom and a
\emph{self-join-free} query has no self-joins.
To avoid notational clutter and without loss of generality,
we assume that there are no selection conditions on individual relations (like $R(x,1)$ or $R(x,x)$):
Tables can be copied, and selection conditions can always be applied directly to the tables in a preprocessing step that takes $\O(n)$.

\begin{example}[$\ell$-path and $\ell$-cycle queries]
\label{ex:path_cycle}
Let $R_i(A, B), i \in \N_1^\stages$, be tables 
containing directed graph edges from $A$ to $B$.
A length-$\stages$ path and
a length-$\stages$ cycle can respectively be expressed as:
\begin{align*}
	Q_{P\stages }(\vec x) &\datarule R_1(x_1, x_2), R_2(x_2, x_3), \ldots, R_\stages(x_\stages, x_{\stages+1}) 
	&& \textrm{($\ell$-path)}\\
	Q_{C\stages }(\vec x) &\datarule R_1(x_1, x_2), R_2(x_2, x_3), \ldots, R_\stages(x_\stages, x_1)
	&& \textrm{($\ell$-cycle)}.
\end{align*}
\end{example}

We often represent an output tuple as a vector of those input tuples that joined
to produce it, e.g., $(r_1, r_2, r_3, r_4) \in R_1 \times R_2 \times R_3 \times R_4$
for 4-path query $Q_{P4}$. 
We refer to this vector as the \emph{witness} of a result.

\subsection{Ranked Enumeration Problem}
\label{sec:ranking}

We want to order the results of a full CQ based on the weights of their corresponding
witnesses. 
For maximal generality, we define 
the order of query results
based on
algebraic structures called
\emph{selective dioids}~\cite{GondranMinoux:2008:Semirings},
which are semirings with an ordering property.

A \emph{monoid} is a 3-tuple $(W, \oplus, \0)$
where $W$ is a 
set
and
$\oplus: W \times W \to W$ is a closed 
binary 
operation such that:
\begin{enumerate}
    \item $(x \oplus y) \oplus z = x \oplus (y \oplus z)$ (associativity),
    \item $\0 \in W$ satisfies $x \oplus \0 = \0 \oplus x = x, \forall x \in W$ (neutral element).
\end{enumerate}
If additionally it holds that
\begin{enumerate}
\setcounter{enumi}{2}
    \item $x \oplus y = y \oplus x, \forall x, y \in W$ (commutativity),
\end{enumerate}
then the monoid is called a commutative monoid. 

A \emph{semiring} is a 5-tuple $(W, \oplus, \otimes, \0, \1)$, where
\begin{enumerate}
    \item $(W, \oplus, \0)$ is a commutative monoid,
    \item $(W, \otimes, \1)$ is a monoid,
    \item $\forall x, y, z \in W: 
            (x \oplus y) \otimes z =
            (x \otimes z) \oplus  (y \otimes z)$
            (distributivity of $\otimes$ over $\oplus$),
    \item $\forall a \in W: a \otimes \0 = \0 \otimes a = \0$
            ($\0$ is absorbing for $\otimes$).
\end{enumerate}

\begin{definition}[Selective dioid]\label{def:dioid}
A \emph{selective dioid} is a semiring for which
$\oplus$ is \emph{selective}, i.e., it always returns one of the operands:
$\forall x,y \in W: (x \oplus y = x) \lor (x \oplus y = y)$.
\end{definition}

Note that $\oplus$ being selective induces a total order
on $W$ by setting $x \leq y$ iff $x \oplus y = x$. 
We
define result weight as an aggregate of input-tuple weights using the binary operator $\otimes$ repeatedly: 

\begin{definition}[Result Weights]
Let $w$ be a weight function that maps each input tuple to some value
in $W$ and let
$Q(\vec x) \datarule R_1(\vec x_1),\ldots, R_\stages(\vec x_\stages)$
be a full CQ. The \emph{weight} of a result tuple $r$ is the weight
of its witness $(r_1,\ldots, r_\stages)$, $r_i \in R_i$, $i \in \N_1^\stages$,
defined as 
$\bigotimes_{i=1}^{\ell} w(r_i)$.
\end{definition}

Recall \cref{ex:4cycle} where we rank output tuples
by the sum of the weights of the corresponding input tuples,
i.e., the weight of $(r_1,\ldots, r_\stages)$ is $\sum_{i=1}^{\stages} w(r_i)$.
We achieve this by using the selective dioid $(\mathbb{R}^{\infty}, \min, +, \infty, 0)$
with $\mathbb{R}^{\infty} = \mathbb{R} \cup \{ \infty \}$
that is also called the \emph{tropical semiring}.
Notice the correspondence of $\otimes$ to $+$ and $\oplus$ to $\min$.
\new{
In general, we use the term \emph{ranking function} to refer to a
function that maps the query results to a domain equipped 
with a total order $\leq$.
In this paper, the ranking function is captured by a selective dioid:
we use the $\otimes$ operator to aggregate the input weights
into a result weight and then we use 
the $\oplus$ operator on the result weights to compare (or \emph{rank}) them.
}

The central problem in this paper is the following:

\begin{definition}[Ranked enumeration]
Given a query $Q$ over an input database $D$, selective dioid $(W, \oplus, \otimes, \0, \1)$,
and weight function $w$ as defined above, a \emph{ranked enumeration} algorithm returns the
output of $Q$ on $D$ according to the total order induced by $\oplus$.
\end{definition}

We refer to algorithms for ranked enumeration
over the results of a CQ as \emph{any-k} algorithms.
This conforms to our previous work~\cite{yang2018any} and reflects the fact
that the number of returned results need not be set apriori.
Thus, any-k algorithms can be seen as a fusion of
top-k and
anytime algorithms \cite{Zaimag96}
that gradually improve their result over time.

\introparagraph{Generality}
Our approach supports any selective dioid, 
including less obvious cases such as
\emph{lexicographic ordering}
where two output tuples are first compared on their
$R_1$ component, and if equal then on their $R_2$ component, and so on.
For this to be well-defined, there must be a total order on the tuples
within each relation. Without loss of generality, assume this total order is represented
by the natural numbers, such that input tuple $r$ has weight $w'(r) \in \N$.
For the selective dioid, we set $W = \N^\stages$, i.e., each
tuple weight is an $\stages$-dimensional vector of integers. Input tuple $r_j \in R_j$ has
weight $w(r_j) = (0,\ldots, 0, w'(r_j), 0,\ldots, 0)$ 
with zeros except for position $j$ that stores the ``local'' weight value in $R_j$.
Operator $\otimes$ is standard element-wise vector addition, therefore the weight
of a result tuple with witness $(r_1,\ldots, r_\stages)$ is $(w'(r_1),\ldots, w'(r_\stages))$.
To order two such vectors, 
the selective dioid addition
$\oplus$ 
\new{
returns the operand that comes first according to the lexicographic order
}
e.g.,
for $\stages=2$, $(a,b) \oplus (c,d) = (a,b)$ 
if $w'(a) < w'(c)$, 
or $w'(a) = w'(c)$ and $w'(b) < w'(d)$,
and $(c,d)$ otherwise.
The $\0$ and $\1$ elements of the dioid are $\stages$-dimensional
vectors $(\infty,\ldots, \infty)$ and $(0,\ldots, 0)$, respectively.

We will present our approach for the tropical semiring
$(\mathbb{R}^{\infty}, \min, +, \infty, 0)$. 
Generalization to other selective dioids
follows immediately from the fact that the only algebraic properties
that are used in our derivations and proofs are 
those that imply the algebraic structure of a selective dioid \cref{def:dioid}.

Notice that addition over real numbers has an \emph{inverse}, hence
$(\mathbb{R}^{\infty}, +, 0)$ is a group, not just a monoid.
This simplifies the algorithms to a certain extend. 
Our main result (\Cref{TH:MAIN}) still holds even without the inverse with some minor subtleties
that we explain in \Cref{sec:inverse}. 

\new{
\subsection{Complexity Measures}
}
\label{sec:complexity_measures}

For complexity results we use the standard \emph{RAM-model of computation}
that charges $\O(1)$ per data-element access. 
Reading or storing a vector of $i$ elements
therefore costs $\O(i)$.
In line with previous work~\cite{berkholz19submodular,GottlobGLS:2016,ngo2018worst},
we also assume the existence of a data structure that can be built in linear
time to support tuple lookups in constant time. 
In practice, this is virtually guaranteed by hashing,
though formally speaking, only in an expected, amortized sense.

\new{
We measure success with respect to three measures:
($i$) the pre-processing time or \emph{time-to-first} denoted by \TTF, 
($ii$) the delay between the $k-1$'th and $k$'th results for any value of $k$ denoted by $\Del(k)$
and ($iii$) the space requirement until the $k$'th result denoted by $\MEM(k)$.
We will also look at $\TT(k)$ which is the overall time to get the $k$'th result and 
the special case of the \emph{time-to-last} ($\TTL = \TT(|\mathrm{out}|)$), where $\mathrm{out}$
denotes the output of the query.
Notice that $\TTF = \TT(1)$.
}

In line with most previous work on worst-case optimal join algorithms
and decompositions of cyclic queries, we measure asymptotic cost in terms of
{data complexity}~\cite{DBLP:conf/stoc/Vardi82}, i.e., treat query size $|Q|$ as a constant.
The exception is the in-depth analysis of ranked enumeration algorithms
for path CQs (\Cref{sec:complexity}), where including query complexity
reveals interesting differences.

\subsection{Determining Optimality}
\label{sec:optimalityDef}

Consider full CQ $Q$ over input relations with at most $n$ tuples. It takes
$\Omega(n)$ just to look at each input tuple and $\Omega(k)$ to output $k$ result tuples.
Since we also require the output to be sorted and sorting $k$ items has
complexity $\Omega(k \log k)$, we consider a ranked enumeration algorithm to be
\emph{optimal} if it satisfies 
\new{
$\TTF = \O(n)$
and $\Del(k) = \O(\log k)$.
}
\footnote{To be precise, sorting may add less than $k \log k$ overhead if one can replace generic
comparison-based sorting with an algorithm that exploits structural relationships
between weights of input and output tuples. 
However, this is not possible for all inputs and $k$ values.}
For \emph{acyclic} CQs, this 
optimality target is realistic, because
the well-known Yannakakis algorithm~\cite{DBLP:conf/vldb/Yannakakis81}
computes the full (unsorted) output in time $\O(n + |\mathrm{out}|)$.

For \emph{cyclic} CQs, Ngo et al.~\cite{ngo2018worst} argue that the join result
cannot be computed in
$\O(n + |\mathrm{out}|)$ 
and propose the notion of \emph{worst-case optimal} (WCO) join algorithms,
whose computation time is $\O(n + |\textrm{out}_\textrm{WC}|)$.
Here, $|\textrm{out}_\textrm{WC}|$ is the maximum output size of query $Q$ over
\emph{any possible} database instance, which is determined by the AGM bound~\cite{AGM}. 
WCO join algorithms are thus \emph{not sensitive}
to the actual output size of the query on a given database instance.
Abo Khamis et al.~\cite{khamis17panda} argue for a stronger,
output-sensitive notion of optimality based on the \emph{width} $\omega$ of a decomposition
of a cyclic CQ $Q$ into a set $\boldsymbol{\mathcal{Q}}$ of acyclic CQs covering
$Q$.\footnote{The union of their output equals the output of $Q$.}
The input relations of the acyclic CQs in $\boldsymbol{\mathcal{Q}}$ are derived
from the original input and have cardinality $\O(n^\omega)$ for
$\omega \ge 1$ ideally as small as possible.
Let $\mathcal{A}$ be such a decomposition-based
algorithm and let $T(\mathcal{A})$ denote its time complexity for creating decomposition
$\boldsymbol{\mathcal{Q}}$. By applying the Yannakakis algorithm to the acyclic queries in
$\boldsymbol{\mathcal{Q}}$, cyclic query $Q$ can
be evaluated in time $\O(T(\mathcal{A}) + |\mathrm{out}|)$
and its Boolean version $Q^B$
in $\O(T(\mathcal{A}))$. The current frontier are decompositions based on the
\emph{submodular} width $\omega = \subw(Q)$~\cite{Marx:2013:THP:2555516.2535926},
which is considered a yardstick of optimality for full and Boolean CQs~\cite{khamis17panda}.

We adopt this notion of optimality and, arguing similar to the acyclic case, we
say that ranked enumeration over a full CQ is optimal if
\new{
$\TTF = \O(T(\mathcal{A}))$ and
$\Del(k) = \O(\log k)$.
}
Intuitively,
this ensures that ranked enumeration adds ``almost no overhead''
compared to unranked enumeration, because outputting $k$ results would
take at least $\Omega(k)$.

\section{Path query and its connection to Dynamic Programming (DP)}
\label{sec:DPtoAnyK}

We formulate optimal ranked enumeration for path queries as a
Dynamic Programming (DP) problem, then generalize to trees and
cyclic queries. Following common terminology,
we use DP to denote what would more precisely be called
\emph{deterministic serial DP} with a finite fixed number of decisions
\cite{bertsekas05dp,Cormen:2009dp,dpv08book}.
These problems have a {\em unique minimum} of the cost function and
DP constructs a {\em single} solution that realizes it.
Formally, a DP problem has a set of \emph{states} $\Sset$, which contain
local information for decision-making~\cite{bertsekas05dp}.
We focus on what we will refer to as \emph{multi-stage DP}. Here each state belongs
to exactly one of $\stages > 0$ stages, where $\Sset_\sgiter$ denotes the set of
states in stage $\sgiter$, $\sgiter \in \N_0^\stages$. 
The \emph{start} stage has a single state $\Sset_0 = \{s_0\}$ and there is
a \emph{terminal} state $s_{\stages+1}$
which we also denote by $t$ for convenience.
At each state $s$ of stage $\sgiter$, we have to make a {\em decision}
that leads to a state $s' \in \Sset_{\sgiter+1}$. 
We use $\Dec \subseteq \bigcup_{\sgiter=0}^\stages (\Sset_\sgiter \times \Sset_{\sgiter+1})$ 
for the set of possible decisions.

DP is equivalent to a shortest-path problem on a corresponding weighted graph, in our case a $(\stages+2)$-partite
directed acyclic graph (DAG)~\cite{bertsekas05dp, dpv08book}, where states correspond to nodes and decisions define
the corresponding edges. Each decision $(s, s')$ is associated with a \emph{cost} $\weight(s, s')$, 
which defines the weight of the corresponding edge in the shortest-path problem.\footnote{We use \emph{cost} and
\emph{weight} interchangeably. Cost is more common in optimization problems, weight in shortest-path problems.
We sometimes use ``lightest path'' in order to emphasize that all paths have the same number of nodes, but differ
in their weights.}
By convention, an edge exists iff its weight is less than $\infty$.

We now generalize the path definition from \Cref{ex:path_cycle} and show that
ranked enumeration over this query can be modeled as an instance of DP. Consider
\begin{equation*}
Q'_{P\stages}(\vec x, \vec y) \datarule
R_1(\vec y_1, \vec x_2), R_2(\vec x_2, \vec y_2, \vec x_3), \ldots, R_\stages(\vec x_\stages, \vec y_\stages, \vec x_{\stages+1}),
\end{equation*}
allowing multiple attributes in the equi-join conditions and additional attribute sets
$\vec y_i$ that do not participate in joins.
This query can be mapped to a DP instance as follows:
(1) atom $R_i$ corresponds to stage $\Sset_i$ and each
tuple in $R_i$ maps to a unique state in $\Sset_i$,
(2) there is an edge
between $s \in \Sset_i$ and $s' \in \Sset_{i+1}$ iff the corresponding
input tuples join and the edge's weight is the weight of the tuple corresponding to $s'$,
(3) there is an edge from $s_0$ to each state in $\Sset_1$ whose weight is
the weight of the corresponding $R_1$-tuple, and
(4) each state in $\Sset_\stages$ has an edge to $t$ of weight 0.
Clearly, there is a 1:1 correspondence between
paths from $s_0$ to $t$ and output tuples of $Q'_{P\stages}$, and path
``length'' (weight) equals output-tuple weight. Hence the $k$-th heaviest
output tuple corresponds to the $k$-shortest path in the DP instance.

\begin{figure}
\centering
\scalebox{0.85}{	
	\begin{tikzpicture}[
	        mycircle/.style={
	         circle,
	         draw=black,
	         fill=gray,
	         fill opacity = 0.3,
	         text opacity=1,
	         inner sep=0pt,
	         minimum size=18pt,
			 text width=22pt,
			 align=center,
	         font=\small},
	        myarrow/.style={-Stealth},
	        node distance=0.6cm and 1.1cm,
	        ]
	        \node[mycircle,fill=white] (c1) {$s_0$};
	        \node[mycircle,right=of c1] (c3) {``2''};
	        \node[mycircle,below=of c3] (c2) {``3''};
	        \node[mycircle,above=of c3] (c4) {``1''};
	        \node[mycircle,right=of c3] (c6) {``20''};
	        \node[mycircle,below=of c6] (c5) {``30"};
	        \node[mycircle,above=of c6] (c7) {``10"};
	        \node[mycircle,right=of c6] (c9) {``200"};
	        \node[mycircle,below=of c9] (c8) {``300''};
	        \node[mycircle,above=of c9] (c10){``100"};
	        \node[mycircle,right=of c9,fill=white] (c11){$t=s_{4}$};
      
	        \node[above=of c4] (t1) {$\Sset_1$};
	        \node[above=of c7] (t2) {$\Sset_2$};
	        \node[above=of c10] (t3) {$\Sset_3$};
	        \node (t4) at (c1 |- t1) {$\Sset_0$};		
	        \node (t5) at (c11 |- t1) {$\Sset_4$};

	    \foreach \i/\j/\txt/\p in {%
	      c1/c2/3/above,
	      c1/c3/2/above,
	      c1/c4/1/above,
	      c3/c5//above,
	      c3/c6//above,
	      c3/c7//above,
	      c4/c7/10/above,
	      c6/c8//above,
	      c6/c9//above,
	      c6/c10//above,
	      c7/c10/100/above,
	      c8/c11/0/above,
	      c9/c11/0/above,
	      c10/c11/0/above}
	       \draw [myarrow] (\i) -- node[sloped,font=\small,\p] {\txt} (\j);
    
	    \draw [myarrow] (c4) -- node[sloped,font=\small,above,pos=0.3] {20} (c6);
	    \draw [myarrow] (c4) -- node[sloped,font=\small,below,pos=0.15] {30} (c5);
	    \draw [myarrow] (c7) -- node[sloped,font=\small,above,pos=0.3] {200} (c9);
	    \draw [myarrow] (c7) -- node[sloped,font=\small,below,pos=0.15] {300} (c8);
    
	    \draw [myarrow] (c2) -- node[sloped,font=\small,below,pos=0.3] {20} (c6);
	    \draw [myarrow] (c2) -- node[sloped,font=\small,below] {30} (c5);
	    \draw [myarrow] (c2) -- node[sloped,font=\small,above,pos=0.15] {10} (c7);    
	    \draw [myarrow] (c5) -- node[sloped,font=\small,below,pos=0.3] {200} (c9);
	    \draw [myarrow] (c5) -- node[sloped,font=\small,below] {300} (c8);
	    \draw [myarrow] (c5) -- node[sloped,font=\small,above,pos=0.15] {100} (c10); 
	\end{tikzpicture}
}
\caption{DP instance for \Cref{ex:cartesian1}.
}
\label{fig:cartesian}
\end{figure}	

\begin{example}[Cartesian product]\label{ex:cartesian1}
We use the problem of finding the minimum-weight output of Cartesian product
$R_1 \times R_2 \times R_3$ as the running example. 
Let $R_1 = \{``1", ``2", ``3"\}$, 
$R_2 = \{``10", ``20", ``30"\}$ and 
$R_3 = \{``100", ``200", ``300"\}$ 
and set tuple weight equal to tuple label, e.g., tuple $``20"$ in $R_2$ has weight $\weight(``20")=20$. 
\cref{fig:cartesian} depicts how this problem translates into our framework.
\end{example}

A \emph{solution} to the DP problem is a sequence of $\stages$ states 
$\sol = \langle s_1 \ldots s_\stages \rangle$ 
that is {\em admissible}, i.e.\ $(s_{\sgiter}, s_{\sgiter+1}) \in \Dec$,
$\forall i \in \N_0^{\stages}$. 
The \emph{objective function} is the total cost of a solution,
\begin{align}
	\weight(\sol) = \aggrsum_{\sgiter=0}^\stages \weight(s_{\sgiter}, s_{\sgiter+1}), \label{eq:costDP}
\end{align}
and DP finds the minimal-cost solution $\sol_1$. The index denotes the rank,
i.e., $\sol_k$ is the $k$-th best solution.

\introparagraph{Principle of optimality}~\cite{bellman:1958:routing,bellman1954}
The core property of DP is that a solution can be efficiently derived from
solutions to subproblems.
In the shortest-path view of DP, the subproblem at \emph{any} state $s \in \Sset_\sgiter$
is the problem of finding the shortest path from $s$ to $t$. 
With $\sol_1(s)$ and $\solW_1(s)$ denoting the shortest path from $s$ and its weight respectively, 
DP is recursively defined for all states $s \in \Sset_\sgiter, \sgiter \in \N_0^{\stages+1}$ by
\begin{equation}
\begin{aligned}
    \!\!\!\!\!\!\!\solW_1(s) &= 0 \textrm{ for terminal } s \in S_{\stages+1}		 \\
    \!\!\!\!\!\!\!\solW_1(s) &= 
		\!\!\min_{(s, s') \in \Dec} 
		\{\weight(s, s') \aggr \solW_1(s') \},
		\textrm{ for }
		s \in \Sset_\sgiter, \sgiter \in \N_0^\stages .
		\!\!
    \label{eq:DP_recursion}
\end{aligned}
\end{equation}
The optimal DP solution is $\solW_1(s_0)$, i.e., the weight of the lightest path
from $s_0$ to $t$. For convenience we define the set of optimal paths reachable
from $s$ according to \cref{eq:DP_recursion} as
$\Choices_1(s) = \{s \concat \sol_1(s') \ |\ (s, s') \in \Dec \}$.
Here $\concat$ denotes concatenation, i.e., 
$s_i \concat \langle s_{i+1} \ldots s_\stages \rangle = \langle s_i\ s_{i+1} \ldots s_\stages \rangle$.

\begin{example}[continued]
Consider state ``2'' in \cref{fig:LawlerDP}.
It has three outgoing edges and
$\solW_1(``2")$ is computed as the minimum over these three choices. The winner is path
$``2" \concat \sol_1(``10")$ of weight 112. 
Similarly, $\sol_1(``10")$ is found as $``10" \concat \sol_1(``100")$, and so on.
\end{example}

\begin{figure}[tb]
\centering
\includegraphics[width=0.7\linewidth]{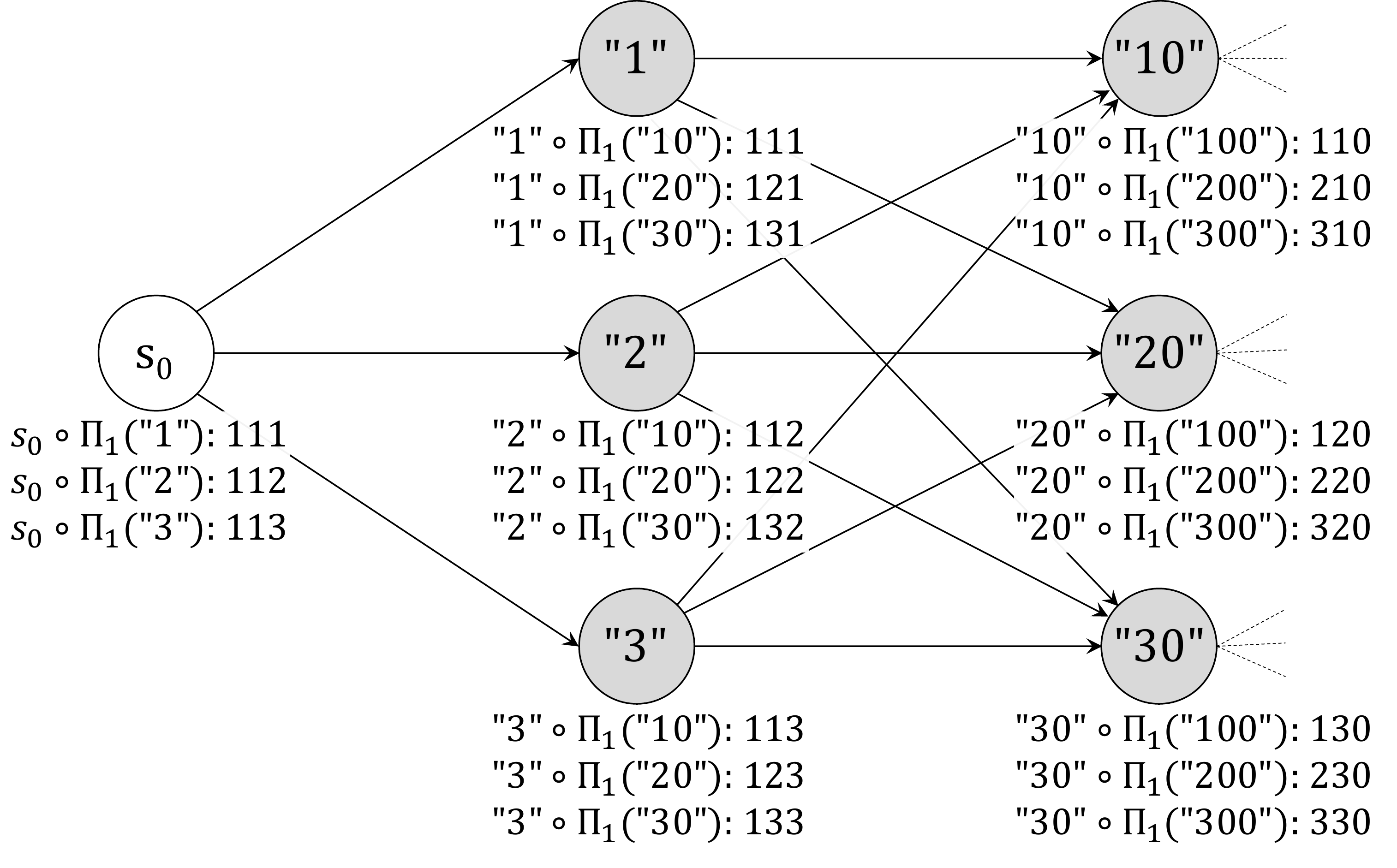}
\caption{Excerpt from \cref{fig:cartesian}, showing $\Choices_1(s)$ for some states $s$.
Term $s \concat \sol_1(s'): w$ denotes a choice, which is a path from $s$, and its weight $w = w(s, s') \aggr \solW_1(s')$.}
\label{fig:LawlerDP}
\end{figure}

\Cref{eq:DP_recursion} can be computed for all states in time $\O(|\Sset| + |\Dec|)$ bottom-up,
i.e., in decreasing stage order from $\stages+1$ to 0.
Consider stage $\Sset_\sgiter$:
To compute $\Choices_1(s)$ for state $s \in \Sset_\sgiter$,
the algorithm retrieves all edges $(s, s') \in \Dec$
from $s$ to any state $s' \in \Sset_{\sgiter+1}$, looks up $\solW_1(s')$,
and keeps track of the minimal total weight $\weight(s, s') \aggr \solW_1(s')$ on-the-fly.
(If no such edge is found, then the weight is set to $\infty$.)
When computing $\solW_1(s)$, the algorithm also adds pointers to keep track of optimal solutions.
E.g., in \cref{fig:LawlerDP} entry $``2" \concat \sol_1(``30")$ at state $``2"$
would point to the minimum-weight choice $``30" \concat \sol_1(``100")$ at state $``30"$. 
This way the 
corresponding paths can be reconstructed by tracing the pointers back ``top-down'' from 
$\solW_1(s_0)$~\cite{bertsekas05dp}. 
Notice that \emph{DP needs only the pointer from the top choice
at each state}, but adding the others is ``free'' complexity-wise,
which we later use for ranked enumeration.

Whenever the bottom-up phase determines $\solW_1(s)=\infty$ during the evaluation of \cref{eq:DP_recursion}, then that state $s$ and all its adjacent edges can be
removed without affecting the space of solutions. 
We use $\SsetR_i \subseteq \Sset_i$ and $\DecR  \subseteq \Dec$ to denote the
\emph{remaining sets of states and decisions}, respectively. 
This DP algorithm corresponds to variable elimination~\cite{DBLP:journals/ai/Dechter99}
on the \emph{tropical semiring}~\cite{Golan:1999di,pin_taylor_atiyah_1998}
and
is reminiscent of the \emph{semi-join reductions by Yannakakis}~\cite{DBLP:conf/vldb/Yannakakis81},
which corresponds to DP with variable elimination on the Boolean semiring~\cite{abo16faq}.

\introparagraph{Encoding equi-joins efficiently}
For an equi-join, the shortest-path problem has $\bigO(\stages n)$ states and
$\bigO(\stages n^2)$ edges,
therefore the DP algorithm has quadratic time complexity in the number of tuples.
We reduce this to $\O(\stages n)$ by an 
\emph{equi-join specific graph transformation} illustrated in 
\cref{fig:equiJoinGraph}. 
Consider the join between $R_1$ and $R_2$, representing stages $\Sset_1$ and $\Sset_2$, respectively.
For each join-attribute value, the corresponding 
states in $R_1$ and $R_2$ 
form a fully connected bipartite graph. For each state,
all \emph{incoming} edges have the 
same weight, as edge weight is determined by tuple weight. 
Hence we can represent the subgraph equivalently
with a single node ``in-between'' the matching states in $\Sset_1$ and $\Sset_2$,
assigning zero weight to the edges adjacent to 
states in $\Sset_1$ and the corresponding tuple weight to those adjacent to a state in $\Sset_2$.
The transformed representation has only $\O(\stages n)$ edges.
At its core, our encoding relies on the conditional independence of the non-joining
attributes given the join attribute value, a property also exploited in
factorized databases \cite{olteanu16record}.
Here we provide a different perspective on it as a graph transformation
that preserves all paths.

\section{Any-k Algorithms for DP}
\label{sec:DPalgorithms}

We defined a class of DP problems that can be described in terms of a multi-stage DAG,
where every solution is equivalent to a path from $s_0$ to $t$ in
graph $(\SsetR = \bigcup_{\sgiter=0}^{\stages+1} \SsetR_\sgiter, \DecR)$. 
Hence we use terminology from DP (solution, state, decision)
and graphs (path, node, edge) interchangeably.

\begin{figure}[tb]
\centering
\includegraphics[width=.7\linewidth]{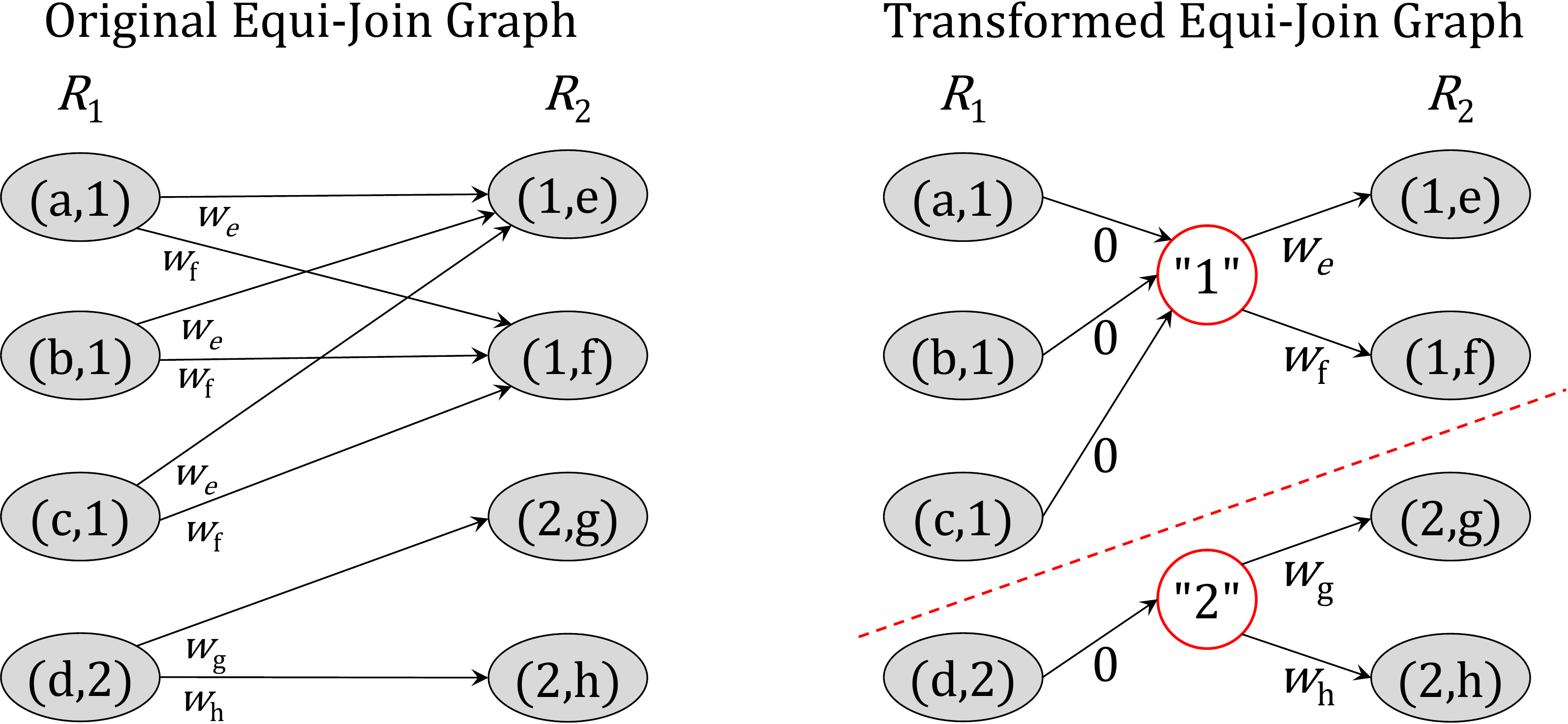}
\caption{\emph{Equi-join} from $\O(n^2)$ representation to $\O(n)$.}
\label{fig:equiJoinGraph}
\end{figure}

In addition to the minimum-cost path, ranked enumeration must retrieve \emph{all paths}
in cost order. Let $\sol_k(s)$ be the
$k^\textrm{th}$-shortest path from $s$ to $t$ and $\solW_k(s)$
its cost. The asymptotically best $k$-shortest-paths algorithm was proposed by
Eppstein~\cite{eppstein1998finding}, yet it is not the best choice for our use case.
In the words of its author, it is ``rather complicated'', thus it is unclear
how to extend it from path to \emph{tree} queries.
Since our DP problems are only concerned with \emph{multi-stage DAGs}
(Eppstein targets more general graphs), we propose a simpler and
easier-to-extend algorithm, \HEAP, that guarantees the same complexity as Eppstein.\footnote{Implementations
of ``Eppstein's algorithm'' exist, but they seem to implement a simpler variant with weaker
asymptotic guarantees that was also introduced in \cite{eppstein1998finding}.}

Below we explore algorithms that fall into two categories.
The first appeared in various optimization contexts as methods that
\emph{partition} the solution space and trace their roots to
Lawler~\cite{lawler72} and Murty~\cite{murty1968},
including recent work on subgraph isomorphism~\cite{chang15enumeration}. 
We call this family \RPDP; it includes \HEAP.
The second finds the {$k$-shortest paths} in a graph
via \emph{recursive equations}~\cite{dreyfus69shortest,jimenez99shortest}.
We refer to the application of this idea to our framework as \REDP.

\subsection{Repeated Partitioning DP (ANYK-PART)}
\label{sec:RP}

\subsubsection{The Lawler Procedure and DP}

Lawler~\cite{lawler72} proposed a procedure for ranked enumeration by
repeatedly {\em partitioning the solution space}, which can be applied to any
optimization problem over a fixed set of variables, not only DP.
In our problem, there is one variable per stage and it can take any state in that
stage as a value. Lawler only assumes the existence of a method $\best$
that returns the optimal variable assignment over any space
$\Sset'_1 \times\cdots\times \Sset'_\stages$,
where $\forall \sgiter: \Sset'_\sgiter \subseteq \Sset_\sgiter$.

The top-ranked solution $\langle s_1^* \ldots s_\stages^* \rangle$ is obtained by
executing $\best$ on the unconstrained space $\Sset_1 \times\cdots\times \Sset_\stages$.
To find the  second-best solution, Lawler creates $\stages$ \emph{disjoint subspaces}
such that subspace $\sgiter$ has 
the first $\sgiter-1$ variables fixed to the top-ranked solution's prefix
$\langle s_1^* \ldots s_{\sgiter-1}^* \rangle$ and the $\sgiter$-th variable restricted to
$\Sset_\sgiter - \{ s_\sgiter^* \}$.
Then it applies $\best$ to each of these subspaces to find the top solution 
in each. The second-best overall solution is the best of these $\stages$
subspace solutions. The procedure continues analogously by generating the
corresponding subspaces for the second-best solution, adding
them to a priority queue of candidates.

Chang et al.~\cite{chang15enumeration} showed that the $k^\textrm{th}$-ranked solution
$\langle s_1 \ldots s_\stages \rangle$ is the output of $\best$ on some subspace
\begin{equation}
P=\{s_1\} \times \! \cdots \! \times \{s_{\sglim-1}\} \times (\Sset_{\sglim} - U_{\sglim}) \times \Sset_{\sglim+1} \times \! \cdots \! \times \Sset_\stages,
\label{eq:subspace}
\end{equation}
with $U_{\sglim}$ being a set of states excluded from $\Sset_\sglim$.
The new candidates to be added to the candidate set for the $(k+1)^{\textrm{st}}$ 
result are the results obtained by executing $\best$ on the following
$\stages-\sglim+1$ subspaces:
\begin{tabbing}
$P_{\sglim} = \{s_1\} \times \! \cdots \! \times \{s_{\sglim\!-\!1}\} \times (\Sset_{\sglim} \!-\! U_{\sglim} \!-\! \{ s_{\sglim} \}) \times S_{\sglim\!+\!1} \times \! \cdots \! \times \Sset_\stages$ \\

$P_{\sglim+1} \!= \! \{s_1\} \times \! \cdots \! \times \{s_{\sglim\!-\!1}\} \times \{ s_{\sglim} \} \times (\Sset_{\sglim\!+\!1} \!-\! \{ s_{\sglim\!+\!1} \}) \times \! \cdots \! \times \Sset_\stages$ \\

$\qquad\qquad\vdots$ \\

$P_{\stages} = \{s_1\} \times \! \cdots \! \times \{s_{\sglim-1}\} \times \! \cdots \! \times  \{ s_{\stages-1} \} \times (\Sset_\stages \!-\! \{ s_{\stages} \} )$.     
\end{tabbing}    

\introparagraph{Efficient computation} Instead of calling $\best$ from scratch on
each subspace, we propose to exploit the structure of DP. Consider any subspace
$P$ as defined in \cref{eq:subspace}. Since prefix $\langle s_1 \ldots s_{\sglim-1} \rangle$
is fixed, we need to find the best suffix starting 
from state $s_{\sglim-1}$. In the next stage $\Sset_\sglim$, only states that are
\emph{not} in exclusion set $U_{\sglim}$ can be selected, i.e., the set of
choices at $s_{\sglim-1}$ is restricted by $U_{\sglim}$. Formally,
\begin{align}
\best(P) &= \langle s_1 \ldots s_{\sglim-1} s \rangle \concat \sol_1(s)
\textrm{, where}\\
s &= \arg\min_{s' \in \Sset_{\sglim} - U_{\sglim}} \{w(s_{\sglim-1}, s')+\pi_1(s') | \nonumber\\
&\qquad s_{\sglim-1} \concat \Pi_1(s') \in \Choices_1(s_{\sglim-1})\}, \label{eq:bestDP}
\end{align}
therefore \cref{eq:bestDP} can be solved \emph{using only information that was
already computed by the standard DP algorithm}. 
Note that all elements in a choice set other than the 
minimum-weight element are often referred to as \emph{deviations} from the optimal path.

\begin{example}[continued]
After returning $\sol_1(s_0) = \langle ``1" \ ``10" \ ``100" \rangle$, 
Lawler would solve three new optimization problems to find the second-best result. 
The first subspace is the set of paths that start at $s_0$, but cannot use state $``1"$. 
The second has prefix $\langle ``1" \rangle$ and cannot use state ``10''.
The third has prefix $\langle ``1"\ ``10" \rangle$ and cannot use state $``100"$. 
The best solution to the first subproblem is $\langle ``2" \ ``10" \ ``100" \rangle$, 
corresponding to deviation $s_0 \concat \solW_1(``2")$ of weight 112. 
For the second subproblem, the best result is found similarly as the second-best option 
$``1" \concat \solW_1(``20") = \langle ``1" \ ``20" \ ``100" \rangle$. 
For the third subproblem, the best subspace solution $\langle ``1" \ ``10" \ ``200" \rangle$ is
obtained analogously at state $``10"$.
\end{example}

\subsubsection{The ANYK-PART family of algorithms}

We propose a generic template for \RPDP algorithms and show how all existing approaches
and our novel \HEAP algorithm are obtained as specializations based on how the
Lawler-created subspace candidates are managed.
All \RPDP algorithms first execute standard DP, which produces for each state
$s$ the shortest path $\sol_1(s)$, its weight $\solW_1(s)$, and set of choices
$\Choices_1(s)$. The main feature of 
\RPDP is a set $\Cand$ of \emph{candidates}: it manages the best solution(s) found
in each of the subspaces explored so far.
To produce the next result, the \RPDP algorithm (\Cref{alg:dp-anyk})
(1) removes the lightest candidate from the candidate set $\Cand$, 
(2) expands it into a complete solution, and
(3) adds all new candidates found in the corresponding subspaces to $\Cand$.
We implement $\Cand$ using a priority queue with combined logarithmic time for
removing the top element and inserting a batch of new candidates.

\begin{algorithm}[t]
\setstretch{0.85} %
\small
\SetAlgoLined
\LinesNumbered
\textbf{Input}: DP problem with stages $\Sset_1, \ldots, \Sset_\stages$\\
\textbf{Output}: solutions in increasing order of weight\\

Execute standard DP algorithm to produce for each state $s$: $\sol_1(s)$, $\solW_1(s)$, and $\Choices_1(s)$\;

\algocomment{Initialize candidate set with top-1 result $\langle s_1^* \ldots s_{\stages}^* \rangle$}\;
\algocomment{A candidate consists of 4 fields: prefix $\langle s_1 \ldots s_{\sglim-1} \rangle$,
lastState $s_\sglim$, prefixWeight $\weight(\langle s_1 \ldots s_{\sglim-1} \rangle)$,
and choiceWeight $\weight(s_{\sglim-1}, s_\sglim) + \solW_1(s_\sglim)$.}\;

$\Cand.\mathrm{add}([\langle s_0^* \rangle, s_1^*, 0, \weight(s_0^*, s_1^*) + \solW_1(s_1^*)])$\label{line:initialCand}\;

\Repeat {query is interrupted or $\Cand$ is empty}{\label{line:repeat}

    \algocomment{Pop the candidate with the lowest sum of prefixWeight and choiceWeight. Let that be $[\langle s_1 \ldots s_{\sglim-1} \rangle, s_\sglim, \weight(\langle s_1 \ldots s_{\sglim-1} \rangle), \weight(s_{\sglim-1}, s_\sglim) + \solW_1(s_\sglim)]$}\;
    solution = $\Cand.\mathrm{popMin}()$\;

    \algocomment{Complete the partial solution with the optimal suffix and generate new candidates in all subspaces.}\;
    \For {stages from $\sglim$ to $\stages$\label{line:for1}}{
        \algocomment{Expand the prefix to the next stage. The tail of a prefix is its last element. $\Suc(x,y)$ returns an appropriate subset of $\Choices_1(x)$.}\;
        tail = solution.prefix.tail\;
        last = solution.lastState\;
        \For {$s \in \Suc(\mathrm{tail}, \mathrm{last})$ \label{line:for2}}{
            newCandidate = (solution.prefix, $s$, solution.prefixWeight, $\weight(\mathrm{tail}, s) + \solW_1(s)$)\;\label{line:newCand}
            $\Cand.\mathrm{add}(\mathrm{newCandidate})$\;
        }
        \algocomment{Update solution by appending the last state to the prefix.}\;
        solution.prefix.append(last)\label{line:solution}\;
        solution.prefixWeight.add($w$(tail, last))\;
        $s' = \arg \min_{s''} \{w(\mathrm{last}, s'') \aggr \solW_1(s'') \,|\, \mathrm{last} \concat \sol_1(s'') \in \Choices_1(\mathrm{last})\}$\;
        solution.lastState = $s'$\;
        solution.choiceWeight = $w(\mathrm{last}, s') \aggr \solW_1(s')$\;
    }
    output solution\label{line:output}\;
}
\caption{\RPDP}
\label{alg:dp-anyk}
\end{algorithm}

\begin{example}[continued]
The standard DP algorithm identifies $\langle ``1"\ ``10"\ ``100"\rangle$ as the
shortest path and generates the choice sets as shown in \cref{fig:LawlerDP}. 
Hence $\Cand$ initially contains only candidate
$(\langle s_0 \rangle, ``1", 0, 1+110=111)$ (\Cref{line:initialCand}), 
which is popped in the first iteration of the
repeat-loop (\Cref{line:repeat}), leaving $\Cand$ empty for now. 
The for-loop (\Cref{line:for1}) is executed for
stages 1 to $\stages=3$. For stage 1, we have $\mathrm{tail} = s_0$ and $\mathrm{last} = ``1"$. 
For the successor function (\Cref{line:for2}),
there are different choices as we discuss in more detail in \Cref{sec:anyKpartInstantiations}.
For now, assume $\Suc(x,y)$ returns the \emph{next-best choice at state $x$ after the
previous choice $y$}. Hence the successor of $``1"$ at state $s_0$ is $``2"$. 
As a result, newCandidate is set to $(\langle s_0 \rangle, ``2", 0, 2+110)$---it is
the winner for the first subspace---and added to $\Cand$. Then the solution is
expanded (\Cref{line:solution}) to $(\langle s_0\ ``1" \rangle, ``10", 1, 10+100)$,
because $``10"$ is the best choice from $``1"$.
The next iteration of the outer for-loop (\Cref{line:for1}) adds
candidate $(\langle s_0\ ``1" \rangle, ``20", 1, 20+100)$ to $\Cand$ and updates
the solution to $(\langle s_0\ ``1"\ ``10" \rangle, ``100", 11, 100)$. 
The third and final iteration adds candidate
$(\langle s_0\ ``1"\ ``10" \rangle, ``200", 11, 200)$ and updates
the solution to $(\langle s_0\ ``1"\ ``10"\ ``100" \rangle, t, 111, 0)$,
which is returned as the top-1 result.

At this time, $\Cand$ contains entries $(\langle s_0 \rangle, ``2", 0, 112)$,
$(\langle s_0\ ``1" \rangle, ``20", 1, 120)$,
and $(\langle s_0\ ``1"\ ``10" \rangle, ``200", 11, 200)$.
Note that each is the shortest path in the corresponding subspace as defined by the Lawler procedure. 
Among the three, $(\langle s_0 \rangle, ``2", 0, 112)$ 
is popped next, because it has the lowest sum of prefix-weight (0) and choice-weight (112). 
The first new candidate created for it is $(\langle s_0 \rangle, ``3", 0, 113)$, followed by $(\langle s_0\ ``2" \rangle, ``20", 2, 120)$, 
and $(\langle s_0\ ``2"\ ``10" \rangle, ``200", 12, 200)$. At the same time, the solution is expanded to $(\langle s_0\ ``2"\ ``10"\ ``100" \rangle, t, 112, 0)$. 
\end{example}

\subsubsection{Instantiations of ANYK-PART}
\label{sec:anyKpartInstantiations}

The main design decision in \Cref{alg:dp-anyk} is how to manage the choices at each state
and how to implement successor-finding (\Cref{line:for2}) over these choices.

\smallsection{Strict approaches}
A natural implementation of the successor function returns precisely the next-best choice.

Eager Sort (\EAGER):
Since a state might be reached repeatedly through different prefixes, it
may pay off to pre-sort all choice sets by weight and add pointers from each choice
to the next one in sort order. 
Then $\Suc(x, y)$ returns the next-best choice at $x$ in constant time
by following the next-pointer from $y$.

Lazy Sort (\LAZY):
For lower pre-processing cost, we can leverage the approach Chang et al.~\cite{chang15enumeration}
proposed in the context of graph-pattern search.
Instead of sorting a choice set, it constructs a binary heap in linear time.
Since all but one of the successor requests in a single repeat-loop execution
are looking for the second-best 
choice\footnote{During each execution of the repeat-loop, only the first iteration
of \Cref{line:for1} looks for a lower choice.}, 
the algorithm already pops the top two choices off the heap and moves them into a sorted list. 
For all other choices, the first access popping them from the heap will append them to the
sorted list that was initialized with the top-2 choices. 
As the algorithm progresses, the heap of choices gradually empties out, filling the sorted list
and thereby converging to \EAGER.

\smallsection{Relaxed approaches}
Instead of finding the \emph{single true successor}
of a choice,  what if the algorithm could return a set of \emph{potential successors}?
Correctness is guaranteed, as long as the true successor is contained in this set
or is already in $\Cand$. (Adding potential successors early to $\Cand$ does not
affect correctness, because they have higher weight and would not be popped
from $\Cand$ until it is ``their turn.'')
This relaxation may enable faster successor finding, but inserts
candidates earlier into $\Cand$.

All choices (\MIN):
This approach is based on a construction that Yang et al.~\cite{yang2018any} proposed
for any-k queries in the context of graph-pattern search.
Instead of trying to find the true successor of a choice, 
\emph{all} but the top choice are returned by $\Suc$. 
While this avoids any kind of pre-processing overhead, it inserts
$\O(n)$ potential successors into $\Cand$.

\HEAP:
We propose a new approach that has better asymptotic complexity than any of the above.
Intuitively, we want to keep pre-processing at a minimum (like \MIN),
but also return a few successors fast (like \EAGER). To this end, we organize each choice set
as a binary heap. In this tree structure, the root node is the minimum-weight choice and the
weight of a child is always 
greater than
its parent. Function $\Suc(x, y)$ (\Cref{line:for2}) returns the two children
of $y$ in the tree.
Unlike \LAZY, we never perform a pop operation and the heap stays intact for the entire
operation of the algorithm; it only serves as a partial order on the choice set, pointing
to two successors every time it is accessed.
Also note that the true successor does not necessarily
have to be a child of node $y$. Overall, returning two successors is asymptotically
the same as returning one and heap construction time is linear~\cite{Cormen:2009dp},
hence this approach asymptotically dominates the others.

\subsection{Recursive Enumeration DP (ANYK-REC)}
\label{sec:RE}

\begin{figure}
\centering
\begin{subfigure}{.5\textwidth}
    \centering
    \includegraphics[width=.8\linewidth]{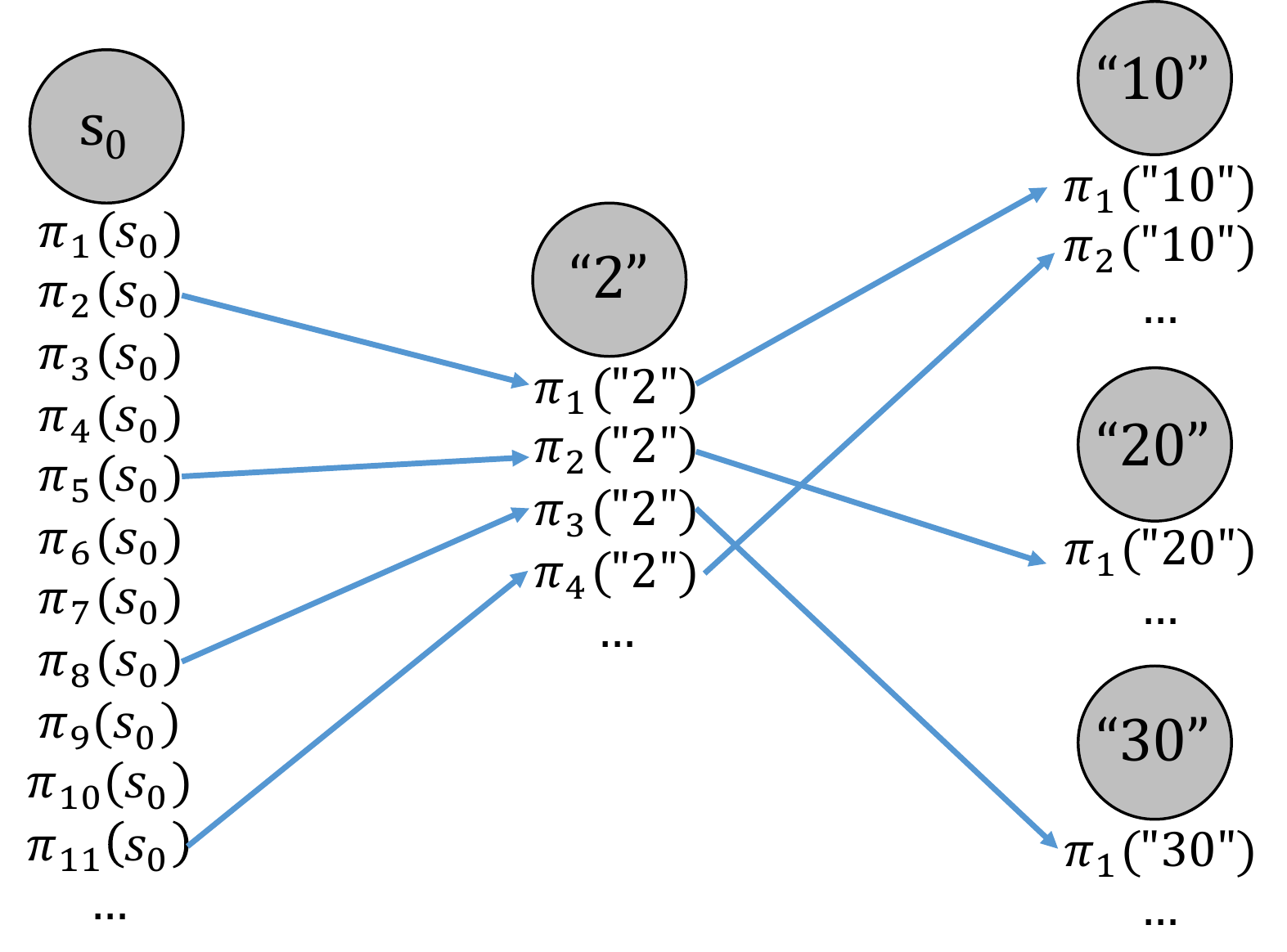}
    \caption{Pointers between solutions from and to $``2"$.}
    \label{fig:REApis}
\end{subfigure}%
\begin{subfigure}{.5\textwidth}
    \centering
    \includegraphics[width=.91\linewidth]{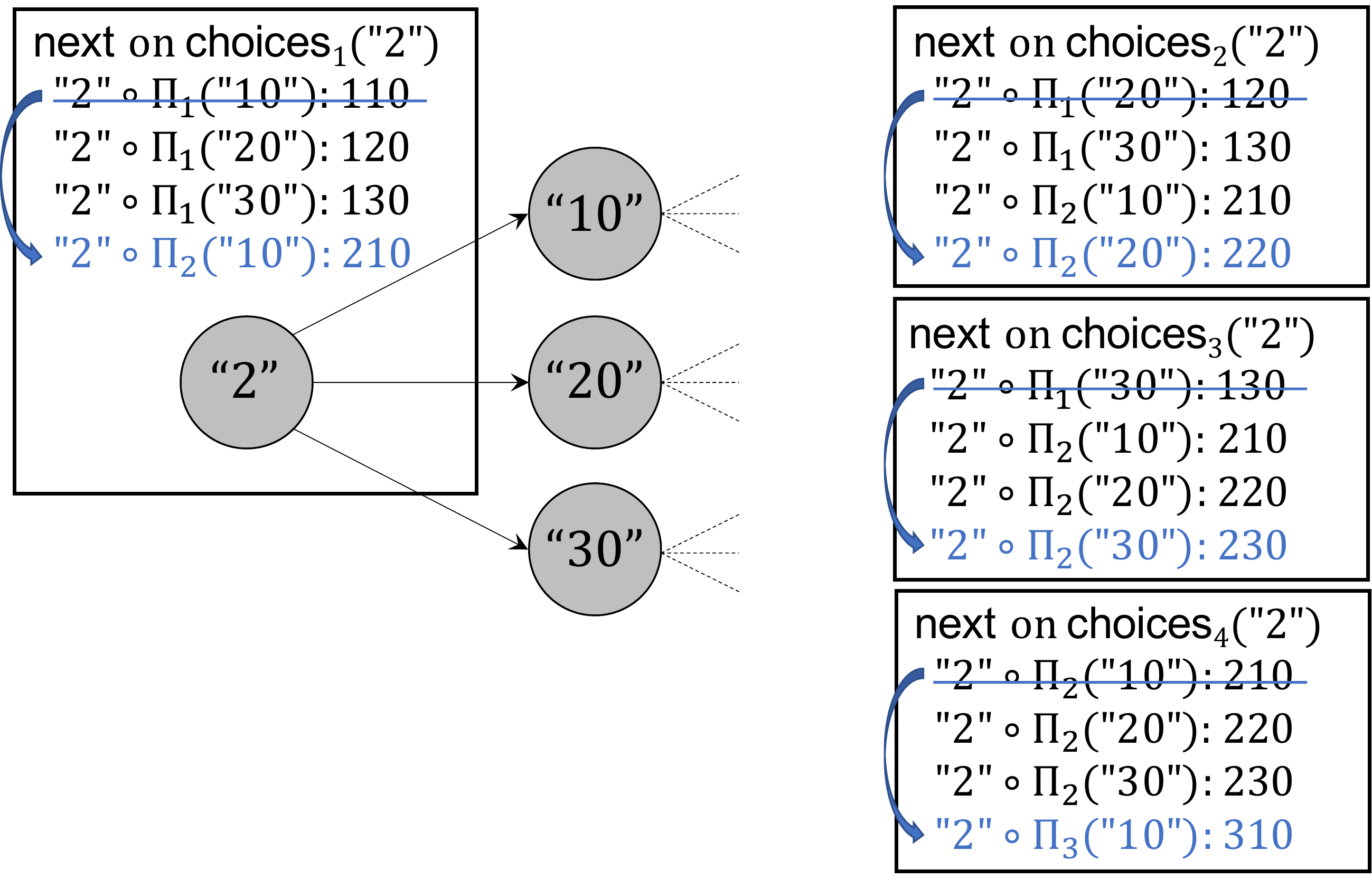}
    \caption{Recursive enumeration 
    at state $``2"$.}
    \label{fig:REApq}
\end{subfigure}
\caption{\Cref{ex:RE-DP}: Recursive enumeration}
\label{fig:REA}
\end{figure}

\begin{algorithm}[t]
\setstretch{0.85}   %
\small
\SetAlgoLined
\LinesNumbered
\SetKwFunction{RecFun}{next}
\SetKwProg{Fn}{Function}{:}{}
\textbf{Input}: DP problem with stages $\Sset_1, \ldots, \Sset_\stages$\\
\textbf{Output}: solutions in increasing order of weight\\

Execute standard DP algorithm to produce for each state $s$: $\sol_1(s)$, $\solW_1(s)$, and $\Choices_1(s)$\;

\algocomment{Initialization phase}\;
\For {stages $\sgiter$ from $\stages-1$ to $0$}{
    \For {states $s \in \Sset_\sgiter$}{
        $\Choices_1(s) = \{ s \concat \sol_1(s') \ | \ (s, s') \in \DecR \}$\;  \label{rea_line:init}
        $\sol_1(s) = \Choices_1(s)$.peek() \;
    }
}

\algocomment{Enumeration phase}\;
$k = 1$ \;
\Repeat {User stop process $\vee$ $\Choices_k(s_0)$ is empty}{
    $\sol_k(s_0) = \Choices_k(s_0)$.popMin() \;
    Output $\sol_k(s_0)$\;
    $\mathrm{\RecFun}(\sol_k(s_0))$\;
    $k = k + 1$\;
}
\;

\algocomment{Returns the next best solution starting from $s$}\;
\Fn{\RecFun{$\sol_{j_s}(s)$}}{
    \algocomment{Base case: Last stage}\;
    \If {$s \in \SsetR_\stages$}{
        \KwRet null\;
    }
    \algocomment{If $\sol_{j_s + 1}(s)$ has been computed by some previous call, it has been stored at state $s$}\;
    \If {$\sol_{j_s + 1}(s)$ has not been computed}{
        \algocomment{$\sol_{j_s}(s)$ is at the top of the priority queue, pop it so that we can construct $\Choices_{j_s + 1}(s)$ }\;
        $\Choices_{j_s}(s)$.popMin()\;
        \algocomment{Assume $\sol_{j_s}(s) = s \concat \sol_{j_{s'}}(s')$.}\;
        \algocomment{Compute $\sol_{j_{s'}+1}(s')$ recursively.}\;
        $\sol_{j_{s'} + 1}(s') = \mathrm{\RecFun}(\sol_{j_s}(s'))$\;
        \If {$\sol_{j_{s'} + 1}(s') \neq null$}{
            $\Choices_{j_s}(s)$.insert$(s \concat \sol_{j_{s'} + 1}(s'))$\; \label{rea_line:insert}
        }
        $\Choices_{j_s + 1}(s) = \Choices_{j_s}$\;   \algocomment{To get $\sol_{j_s + 1}(s)$, peek instead of popping. The pop will happen in the following call for $\mathrm{\RecFun}(\sol_{j_s + 1}(s))$.}\; 
        $\sol_{j_s + 1}(s) = \Choices_{j_s + 1}(s)$.peek()\;
    }
    \KwRet $\sol_{j_s + 1}(s)$\;
}

\caption{\RECURSIVE}
\label{alg:rea}
\end{algorithm}

\REDP relies on a generalized principle of
optimality~\cite{martins01kshortest}: if the k-th path from start node $s_0$ goes through
$s \in \SsetR_1$ and takes the $j_s$-lightest path $\sol_{j_s}(s)$ from there,
then the next lightest path from $s_0$ that goes through $s$ will take the
$(j_s+1)$-lightest path $\sol_{j_s+1}(s)$ from there.
We will refer to the prototypical algorithm in this space as
\RECURSIVE~\cite{jimenez99shortest}. Recall that lightest path
$\sol_1(s_0)$ from start node $s_0$ is found as the minimum-weight path in $\Choices_1(s_0)$. Assume it goes 
through $s \in \SsetR_1$. Through which node does the $2^\textrm{nd}$-lightest path $\sol_2(s_0)$ go? It has to be 
either the $2^\textrm{nd}$-lightest path through $s$, of weight $\weight(s_0, s) + \solW_2(s)$, or the lightest
path through any of the other nodes adjacent to $s_0$. In general, the $k$-th lightest path 
$\sol_k(s_0)$ is determined as the lightest path in some later version
$\Choices_k(s_0) = \{s_0 \concat \sol_{j_s}(s) \ |\ (s_0, s) \in \DecR\}$ of the set of choices,
for appropriate values of $j_s$.
Let $\sol_k(s_0) = s_0 \concat \sol_{j_{s'}}(s')$. Then the $(k+1)^\textrm{st}$
solution $\sol_{k+1}(s_0)$ is found 
as the minimum over the same set of choices, except that 
$s_0 \concat \sol_{j_{s'+1}}(s')$ replaces
$s_0 \concat \sol_{j_{s'}}(s')$.
To find $\sol_{j_{s'+1}}(s')$, the same procedure is applied recursively 
at $s'$ \emph{top-down}. Intuitively, an iterator-style \texttt{next} call at start node $s_0$
triggers a chain of  $\stages$ such \texttt{next} calls along the path that was found
in the previous iteration.

\begin{example}[continued]\label{ex:RE-DP}
Consider node ``2'' in \cref{fig:cartesian}. Since it has adjacent states ``10'', ``20'', and ``30'' in
the next stage, the lightest path $\sol_1(``2")$ is selected from
$\Choices_1(``2") = \{``2"\concat\sol_1(``10"), ``2"\concat\sol_1(``20"), ``2"\concat\sol_1(``30")\}$ as
shown in \cref{fig:LawlerDP}. The first \texttt{next} call on state ``2'' returns 
$``2"\concat\sol_1(``10")$, updating the set of choices for $\sol_2(``2")$ to 
$\{``2"\concat\sol_2(``10"), ``2"\concat\sol_1(``20"), ``2"\concat\sol_1(``30")\}$ as shown in the left
box in~\cref{fig:REApq}. The subsequent \texttt{next} call on state ``2'' then returns 
$``2"\concat\sol_1(``20")$ for $\sol_2(``2")$, causing $``2"\concat\sol_1(``20")$ in 
$\Choices_2(``2")$
to be replaced by $``2"\concat\sol_2(``20")$
for $\Choices_3(``2")$; and so on.
\end{example}

As the lower-ranked paths starting at various nodes in the graph are computed, each node keeps track of 
them for producing the results 
as shown in \cref{fig:REApis}. 
For example, the pointer from $\sol_1(``2")$ 
to $\sol_1(``10")$ at node ``10'' was created by the first \texttt{next} call on ``2'', which found 
$``2"\concat\sol_1(``10")$ as the lightest path in the choice set. %
\Cref{alg:rea} contains the detailed pseudocode.

\subsection{Any-k DP Algorithm Complexity}
\label{sec:complexity}

\definecolor{colorbest}{RGB}{77, 175, 74}

\begin{figure*}[t]
\centering
\footnotesize
\renewcommand{\tabcolsep}{1.3pt}
\begin{tabular}{|l|l|l|l|l|l|}
\hline
Algorithm 	& $\TTF$ 	& $\Del(k)$	& $\TTL$ for $|\mathrm{out}| = \Omega(\stages n)$ & $\TTL$ for $|\mathrm{out}| = \Theta(n^\stages)$ & $\MEM(k)$ \\ 
\hline

\RECURSIVE 	    &\cellcolor{colorbest!20}$\bigO(\stages n)$
                &$\bigO(\stages \log n)$ 
                &$\bigO(|\mathrm{out}| \stages \log n)$ 
				&\cellcolor{colorbest!20}$\bigO(n^\stages (\log n + \stages))$ 
				&\cellcolor{colorbest!20}$\bigO(\stages n + k \stages)$ \\
\hline

\HEAP           &\cellcolor{colorbest!20}$\bigO(\stages n)$
                &\cellcolor{colorbest!20}$\bigO(\log k + \stages)$ 
                &\cellcolor{colorbest!20}$\bigO(|\mathrm{out}| (\log |\mathrm{out}| + \stages))$ 
				& $\bigO(n^\stages \cdot \stages \log n)$ 
				&\cellcolor{colorbest!20}$\bigO(\stages n + k \stages)$ \\
				
\LAZY		    &\cellcolor{colorbest!20}$\bigO(\stages n)$
                & $\bigO(\log k + \stages + \log n)$ 
                &\cellcolor{colorbest!20}$\bigO(|\mathrm{out}| (\log |\mathrm{out}| + \stages))$ 
                & $\bigO(n^\stages \cdot \stages \log n)$ 
				&\cellcolor{colorbest!20}$\bigO(\stages n + k \stages)$ \\

\MIN     	    &\cellcolor{colorbest!20}$\bigO(\stages n)$
                & $\bigO(\log k + \stages n)$ 
                &\cellcolor{colorbest!20}$\bigO(|\mathrm{out}| (\log |\mathrm{out}| + \stages))$ 
                &\cellcolor{colorbest!20}$\bigO(n^\stages \cdot \stages \log n)$ 
				& $\bigO(\stages n + \min\{k n, |\mathrm{out}|\} \stages)$ \\

\EAGER          &$\bigO(\stages n \log n)$
                &$\bigO(\log k + \stages)$
                &$\bigO(|\mathrm{out}| (\log |\mathrm{out}| + \stages))$ 
                & $\bigO(n^\stages \cdot \stages \log n)$ 
				&$\bigO(\stages n + k \stages)$ \\

\hline
\NAIVE     	    & $\bigO(\stages n + |\mathrm{out}| (\log |\mathrm{out}| + \stages))$ 
                & $\bigO(\stages)$ 
                & $\bigO(|\mathrm{out}| (\log |\mathrm{out}| + \stages))$ 
                & $\bigO(n^\stages \cdot \stages \log n)$ 
				&$\bigO(\stages n + |\mathrm{out}| \stages)$ \\
\hline
\end{tabular} 
\caption{Complexity of ranked-enumeration algorithms for equi-joins.
Best performing any-$k$ algorithms with linear $\TTF$ $\O(\stages n)$ in each column are colored in green).
}
\label{tab:complexity_dp}
\end{figure*}

In contrast to the discussion in \Cref{sec:optimalityDef},
which focused on data complexity and treated query size as a constant,
we now include query size in the analysis
to uncover more subtle performance tradeoffs between the different any-k approaches.
Since each input relation has at most $n$ tuples, the DP problem has $\O(\stages n)$
nodes, each with at most $n$ outgoing edges. Based on our equi-join construction
(\cref{fig:equiJoinGraph}), it is easy to see that the total number of edges
is $|\Dec| = \O(\stages n)$. For simplicity we make the following assumptions:
(1) the maximum arity of a relation is bounded by a constant
\cite{DBLP:journals/jacm/Grohe07}
, thus $|Q| = \stages$
, and
(2) the operations $\oplus$ and $\otimes$ of the selective dioid over which the ranking
function is defined take $\gamma = \O(1)$ time to execute.
It is straightforward to extend our analysis to scenarios where those assumptions
do not hold. Note that (2) holds for many practical problems, e.g.,
tropical semiring $(\mathbb{R}^\infty, \min, +, \infty, 0)$, but not for lexicographic
ordering where weights are $\stages$-dimensional vectors and hence
$\gamma = \O(\stages)$.
With \NAIVE, we refer to an algorithm that sorts the full output produced by the
Yannakakis algorithm~\cite{DBLP:conf/vldb/Yannakakis81}.

\subsubsection{Time to First}

All any-k algorithms first execute DP to find the top result and create all choice sets
in time $\O(\stages n)$.
\EAGER requires $\bigO(\stages n \log n)$ for sorting of choice sets. 
Heap construction for \LAZY and \HEAP takes time linear in input size.

\subsubsection{Delay \new{and \TTL}}

Each algorithm requires $\bigO(\stages)$ to assemble an output tuple. In addition, the following
costs are incurred:

\introparagraph{\RPDP}
For all \RPDP algorithms, popMin and bulk-insertion of all new candidates during result expansion take $\bigO(\log |\Cand|)$. For efficient candidate generation (\Cref{line:for2} in \Cref{alg:dp-anyk}) the new candidates do not copy the solution prefix, but simply create a pointer to it.
Therefore, a new candidate is created in $\O(1)$.

\EAGER finds each successor in constant time.
Since $|\Cand| \le k \stages$, its total delay is
$\bigO(\log (k \stages) + \stages) = \bigO(\log k + \stages)$.
For \LAZY, in the first iteration of the
main for-loop (\Cref{alg:dp-anyk}, \Cref{line:for1}), finding the successor
(\Cref{line:for2}) requires at most one pop on a heap
storing $\bigO(n)$ choices. All later iterations find the successor
in constant time. Hence total delay is $\bigO(\log k + \stages + \log n)$.
The \MIN algorithm might insert up to $\stages n$ new candidates to $\Cand$
for each result produced. Hence access to $\Cand$ after producing $k$ results
takes a total of $\bigO(\log (k \stages n))$. All together, delay is
$\bigO(\log k + \log \stages + \log n + \stages n) = \bigO(\log k + \stages n)$.
Finally, \HEAP finds up to two successor candidates of a choice in constant time.
Delay therefore is $\bigO(\log k + \stages)$.
It is easy to see that all these algorithms have
worst-case TTL of $\bigO(n^\stages \cdot \stages \log n)$, the same as \NAIVE
(refer to \cite{yang2018any} for \MIN).

\introparagraph{\REDP}
In \RECURSIVE each \texttt{next} call on $s_0$ triggers 
$\bigO(\stages)$ \texttt{next} calls in later stages---at most one per stage. The call deletes the 
top choice at the state and replaces it with the next-heavier path through the same child node in 
the next stage (see \cref{fig:REApq}). With a priority queue, these operations together take
time $\bigO(\log n)$ per state accessed, for a total delay of $\bigO(\stages \log n)$
between consecutive results. In total, it takes
$\bigO(\stages n + k \stages \log n)$ to produce the top $k$ results.
The resulting $\TTL$ bound of $\bigO(\stages n + |\mathrm{out}| \cdot \stages \log n)$
can be loose because it does not take into account that in later iterations
many \texttt{next} calls will stop early because the corresponding suffixes $\sol_i$
had already been computed by an earlier call:

\begin{theorem}\label{TH:TTL}
There exist DP problems where \RECURSIVE has strictly lower $\TTL$ complexity than \NAIVE.
\end{theorem}

\begin{proof}
Regardless of the implementation of $\NAIVE$, before it terminates it has to
(i) process the input in $\Omega(n \stages)$,
(ii) enumerate all results in $\Omega(|\mathrm{output}| \cdot \stages)$ and
(iii) use a standard comparison-based sort algorithm to batch-rank the entire output
in $(|\mathrm{out}| \log |\mathrm{out}|)$. In total, it needs
$\Omega(n \stages + |\mathrm{out}| (\log |\mathrm{out}| + \stages))$.

For \RECURSIVE, when computing the full result, for each suffix $\solW_\sgiter(s)$
of any state $s$, it holds that the suffix is \emph{exactly once}
inserted into and removed from the priority queue managing $\Choices$ at $s$.
Hence the total number of priority queue operations,
each costing $\bigO(\log n)$, equals the number
of suffixes. Let $\sol_*(\sgiter)$ denote the number of suffixes in stage $\sgiter$,
i.e., the total number of paths starting from any node in $\SsetR_\sgiter$.
Then the total cost for all priority-queue operations is
$\bigO(\log n \sum_{\sgiter=1}^{\stages} \sol_*(\sgiter))$. If
$\sum_{\sgiter=1}^{\stages} \sol_*(\sgiter) = \bigO(\sol_*(1))$,
then this cost is $\bigO(|\mathrm{output}| \cdot \log n)$.
(To see this, note that the set of paths starting at nodes in stage 1 is
the set of all possible paths, i.e., the full output.)
Together with pre-processing time and time to assemble each output
tuple, total $\TTL$ complexity of \RECURSIVE then adds up to
$\bigO(\stages n + |\mathrm{output}| (\log n + \stages))$.
To complete the proof, we show instances where the condition $\sum_{\sgiter=1}^{\stages} \sol_*(\sgiter) = \bigO(\sol_*(1))$ holds and in which the running time of \NAIVE is strictly worse.

Consider the instances with worst-case output $\Theta(n^\stages)$ such as a Cartesian product. 
Recall that the size of the output is the same as the number of suffixes in the first stage, thus $\sol_*(1) = \Theta(n^\stages)$.
Now consider the ratio between $\sol_*(\sgiter)$ and $\sol_*(\sgiter + 1)$ for some stage $i \in \N_1^{\stages-1}$.
That ratio can't be more than $n$ which occurs when $i$ and $i + 1$ are fully connected.
It follows that in order to get that many suffixes in the first stage, 
\emph{every stage} $i$
has to increase the number of suffixes of stage $i + 1$ by a factor of $\Theta(n)$.
Therefore, $\sol_*(1)$ asymptotically dominates the sum $\sum_{\sgiter=1}^{\stages} \sol_*(\sgiter)$, similarly to a geometric series.
Also note that the running time of \NAIVE in these instances is $\Omega(n^\stages \cdot \stages \log n )$, which is higher than $\O(n^\stages (\log n + \stages))$ of \RECURSIVE.

\end{proof}
The lower $\TTL$ of \RECURSIVE is at first surprising,
given that \NAIVE is optimized for bulk-computing and bulk-sorting the entire output.
Intuitively, \RECURSIVE wins because it exploits the multi-stage structure of the
graph---which enables the re-use of shared path suffixes---while \NAIVE uses a general-purpose comparison-based sort algorithm.
We leave as future work a more precise characterization of graph properties that ensure better $\TTL$ for
\RECURSIVE over \NAIVE.

\subsubsection{\new{$\TT(k)$}}

\new{
For all algorithms, $\TT(k) = \O(\TTF + k \cdot \Del(k))$. 
Thus for \HEAP, $\TT(k) = \O(\stages n + k (\log k + \stages))$, while for \LAZY, $\TT(k) = \O(\stages n + k (\log k + \stages + \log n))$.
However, a more careful analysis for the \RPDP variants, gives us the following result:
}

\new{
\begin{proposition}\label{PROP:TTK}
\LAZY achieves $\TT(k) = \O(\stages n + k (\log k + \stages))$, the same as \HEAP.
\end{proposition}
}
\begin{proof}
We will show for all values of $k$ that 
$\stages n + k (\log k + \stages) = \Omega(\stages n + k (\log k + \stages + \log n))$,
thus the seemingly lower $\TT(k)$ complexity of \HEAP is lower bounded 
by the seemingly higher $\TT(k)$ complexity \LAZY.
Since $\stages \geq 1$ and 
$\log n$ is dominated by $\log k$ for $k \geq n$, 
it suffices to show that
$n + k \log k = \Omega(n + k \log n)$ for $k < n$.

For any $1 \leq k \leq n$, it holds that 
$n / k \geq \log(n / k)$ and 
$\log k \geq 0$
and therefore
\begin{align*}
& \frac{n}{k} \geq \log n - \log k \geq \log n - 2 \log k \\
\Rightarrow & n \geq k \log n - 2 k \log k \\
\Rightarrow & (\frac{1}{0.5} - 1) n \geq k \log n - \frac{1}{0.5} k \log k\\
\Rightarrow & \frac{1}{0.5} (n + k \log k) \geq n + k \log n
\end{align*}
This means that there exists an $a > 0$ ($a=0.5$ here) for which 
$n + k \log k \geq a (n + k \log n)$ for all values of $n$,
which completes the proof.
\end{proof}

\begin{figure*}[t]
\centering
\includegraphics[width=0.7\linewidth]{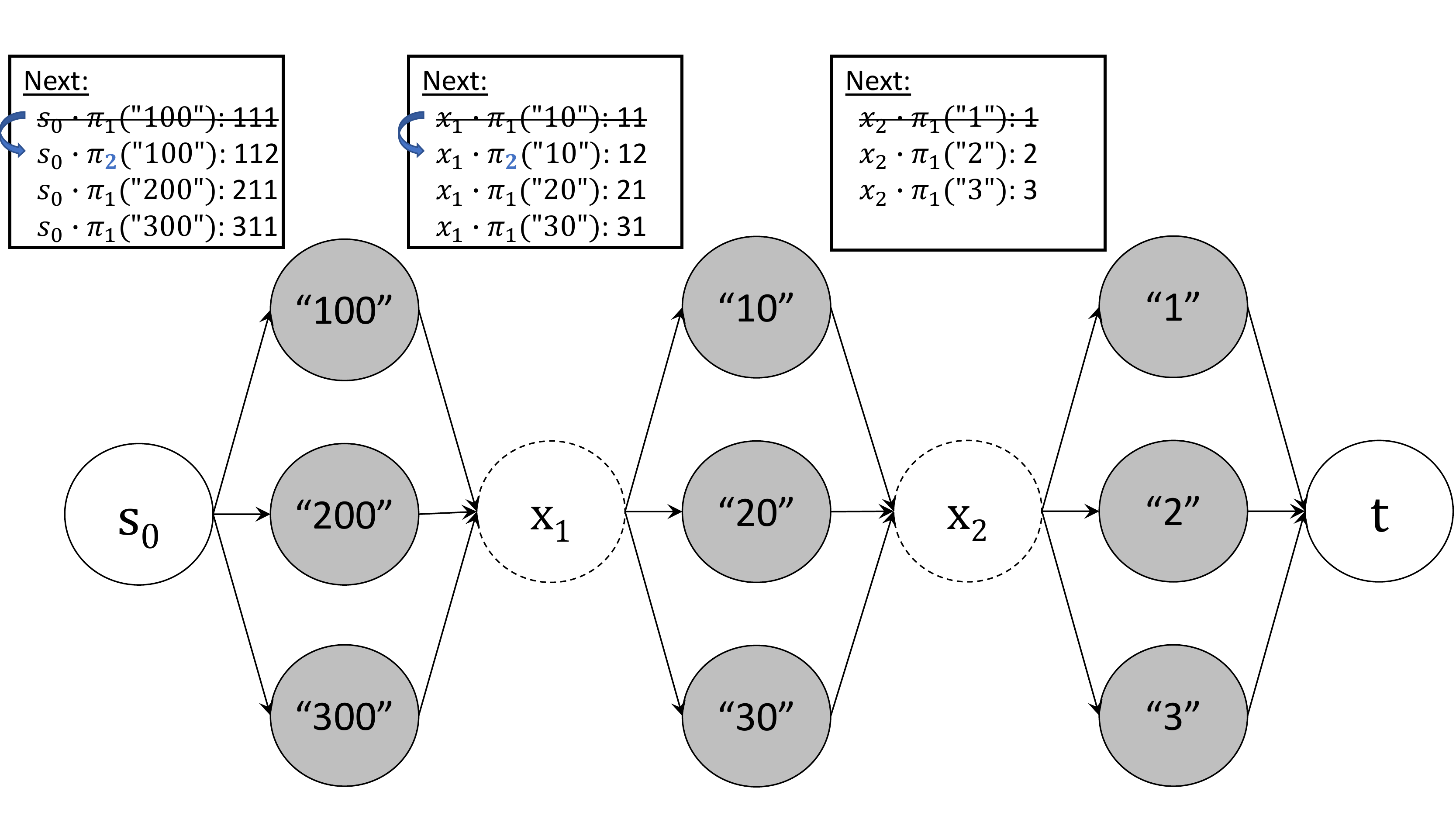}
\caption{A worst-case example for \RECURSIVE and $k=n$. Notice the sharing of data structures between tuples due to our special equi-join encoding (\cref{fig:equiJoinGraph}). Each returned query result entails a sequence of $\ell$ priority queue operations.
}
\label{fig:rea_wc}
\end{figure*}

For \RECURSIVE, our analysis shows that when the number of the $k$ returned results is not ``too large'', the best \RPDP approaches are asymptotically faster.
For instance, when $k=\bigO(n)$, \HEAP achieves $\bigO(n \log n + n \cdot \ell)$ compared to $\bigO(n \cdot \ell \log n)$ of \RECURSIVE.
One might be inclined to think that this gap is an just an artifact of our analysis and it can potentially be closed with arguments similar to the proof of \Cref{TH:TTL}.
However this is not the case, as we now show that the aforementioned bound is {\em tight}, i.e., there exists an instance for which \RECURSIVE needs $\Theta(n \cdot \ell \log n)$ time to return $k = n$ results.

\begin{proposition}\label{TH:REC_WC}
\RECURSIVE is strictly slower than the best implementation of \RPDP for $\TT(n)$ in the worst case.
\end{proposition}
\begin{proof}
To achieve worst-case behavior for \RECURSIVE, we need to (i) create large priority queues and (ii) minimize the sharing of common suffixes between different results.
\Cref{fig:rea_wc} depicts the simplest such example, corresponding to a Cartesian product between $\ell=3$ relations. 
As before, tuple weight is equal to tuple value. 
Notice that each of the first $k = n$ results uses a different tuple from $R_\ell$. 
It is straightforward to set the weights appropriately in order to achieve the same for arbitrary values of $n$, $\ell$. 
To retrieve the $k$'th result, a \texttt{next} call at $s_0$ will trigger a chain of $\ell-1$ other recursive \texttt{next} calls, each one computing $\sol_k$ for a different stage. 
Every \texttt{next} call (except maybe the last one) involves a pop and a push from a priority queue of size $\Theta(n)$, hence $\Theta(n \cdot \ell \log n)$ in total.
At the same time, the worst-case bound for \HEAP is $\bigO(n \log n + n \cdot \ell)$.
\end{proof}

\subsubsection{Memory}

All algorithms need $\O(\stages n)$ memory for storing the input.
The memory consumption of \RPDP approaches depends on the size of $\Cand$.
\MIN grows $\Cand$ by $\O(\stages n)$ elements in each iteration, but creates
at most $|\mathrm{out}|$ candidates in total. The others create only $\O(\stages)$
new candidates per iteration, thus $\MEM(k) = \O(\stages n + k \stages)$.
For \RECURSIVE, size of a choice set $\Choices_k(s)$ is bounded by the out-degree
of $s$, hence cannot exceed $n$.
However, we need to store the suffixes $\sol_i(s)$ produced by the algorithm,
whose number is $\O(\stages)$ per iteration, thus $\MEM(k) = \stages n + k \stages$.
\NAIVE first materializes the output and then sorts it in-place,
therefore has $\MEM(k) = \O(\stages n + |\mathrm{out}| \stages)$, regardless of $k$.

\subsubsection{Summary}

\Cref{tab:complexity_dp} summarizes the analysis for $\TTF$, 
for $\Del(k)$,
for $\TTL$ where the output is sufficiently big (so that result-enumeration time dominates
pre-processing time),
for $\TTL$ on worst-case outputs where we can see the advantage of \RECURSIVE,
and for memory $\MEM(k)$.
All any-k algorithms except
\EAGER have optimal $\TTF = \bigO(\stages n)$.
In contrast, \NAIVE has to sort the full output in
$\bigO(|\mathrm{out}| \log |\mathrm{out}|)$.
\EAGER and \HEAP have the lowest delay $\bigO(\log k + \stages)$.
Only our new algorithm \HEAP achieves optimal $\Del(k)$ after linear $\TTF$ (\Cref{sec:optimalityDef}).

While \RECURSIVE has higher delay than \HEAP, \LAZY, and \EAGER, it has the lowest
$\TTL$ for a worst-case-size output.
This seemingly paradoxical result stems from the fact that as \RECURSIVE outputs
results, it builds up state (ranking of suffixes) that speeds up computation for
later results. Hence even though its delay complexity is tight for small $k$,
our amortized accounting showed that it ultimately must achieve
lower delay for large $k$.

All any-k algorithms but \MIN require minimal space, depending only on input size
and the number of iterations $k$ times query size $\stages$.
\MIN has higher memory demand because it overloads the candidate set early,
while \NAIVE materializes the complete output.

\section{Extension to General CQs}
\label{sec:nsdp}

We extend our ranked enumeration framework from serial to Tree-Based DP (T-DP),
and then to a Union of T-DPs (UT-DP). This enables optimal ranked enumeration of
\emph{arbitrary conjunctive queries}.

\begin{figure}[tb]
\centering
\includegraphics[width=.45\textwidth]{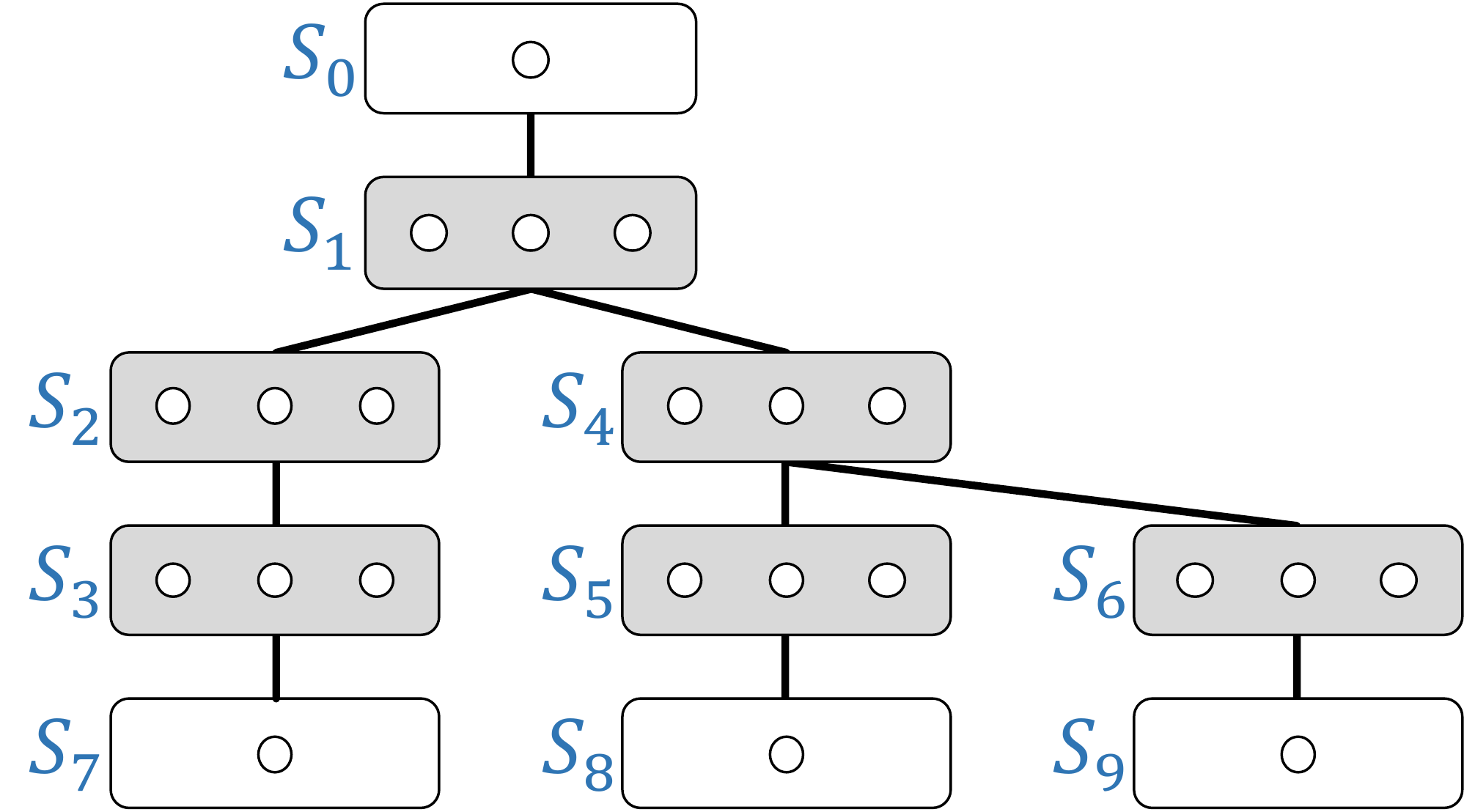}
\caption{\emph{Tree-Based DP} (T-DP) problem structure. 
Rounded rectangles are \emph{stages}, small circles are \emph{states}.}
\label{fig:treedp}
\end{figure}

\subsection{Tree-Based DP (T-DP)}
\label{sec:tdp}

We first consider problems where the stages are organized in a \emph{rooted tree}
with $S_0 = \{ s_0 \}$ as the root stage. 
In these problems, there is a distinct set of decisions $\Dec_{pc}$ for each parent-child pair $p-c$.
\Cref{fig:treedp} depicts an example with 10 stages. 
We assume that all leaf stages contain only one (terminal) state\footnote{
Artificial stages can be introduced to meet this assumption.}, 
thus every root-to-leaf path represents an instance of serial DP as discussed in \cref{sec:DPtoAnyK,sec:DPalgorithms}.
We now extend our approach to Tree-based DP problems (T-DP)
and adapt all any-$k$ algorithms accordingly.

We serialize the stages by assigning a \emph{tree order} that places every parent
before its children, e.g., by a topological sorting of the tree.
To simplify the notation we force the $t$ leaf nodes to be numbered last, i.e.,
$S_i = \{s_i\}$ for $i \in \{\ell+1, \ldots, \ell + t\}$.
We define $\Ch(\Sset_\sgiter)$ 
to be the set of indices of child stages and 
$\parent(\Sset_\sgiter)$ to be the index of the parent stage of $\Sset_\sgiter$. 
In our example, $\Ch(\Sset_4) = \{ 5, 6 \}$ and $\parent(\Sset_4) = 1$. 
By $\lceil S_v \rceil$ we denote the stages of the subtree rooted at $S_v$, 
while $\llceil S_v \rrceil := \lceil S_v \rceil \setminus \{ v \}$. 
In our example, $\llceil S_4 \rrceil := \{  5, 6, 8, 9\}$.
Slightly overloading the notation, we also use 
$\Ch(s_\sgiter) := \Ch(\Sset_\sgiter)$
for a state $s_\sgiter \in \Sset_\sgiter$.
Analogously for 
$\parent(s_\sgiter)$, 
$\lceil s_\sgiter \rceil$, and
$\llceil s_\sgiter \rrceil$.

A T-DP \emph{solution} $\sol = \langle s_1 \ldots s_\stages \rangle$ 
\footnote{
Notice that as in DP, we do not include the unique root state and the $t$ terminal states in the $t$ leaf nodes 
(i.e.\ stages with unique states) in the solution.
}
is a tree with one state per stage and is admissible, i.e.,
$\forall c \in \N_1^{\stages+t}$, if $\parent(c)=p$ 
then $(s_{p}, s_{c}) \in \Dec_{pc}$. 
The \emph{objective function} aggregates the weights of decisions across the entire tree structure:
\begin{align}
	\weight(\sol) = \aggrsum_{c=1}^{\stages+t} \weight(s_{\parent(s_c)}, s_{c}) \label{eq:cost_treeDP}
\end{align}

\introparagraph{T-DP Bottom-up}
The optimal solution is then computed bottom-up by following the serial order
of the stages in reverse.
A bottom-up step for a state $s$ solves a \emph{subproblem} which corresponds to finding an optimal subtree $\sol(s)$.
If $\Ch(s) = \{ i_1, \ldots, i_\chnum \}$, then that subtree consists of 
$s$ and a list of other subtrees rooted at its children $s_{i_1}, \ldots, s_{i_{\chnum}}$. 
To solve a subproblem, we \emph{independently} choose the best decision for each child stage.
The equations 
describing the bottom-up phase in T-DP are
recursively defined for all states and stages
by
\begin{equation}
\begin{aligned}
    \solW_1(s) &= 0, \textrm{ for the $t$ terminals with } \Ch(s) = \emptyset		 \\
    \solW_1(s) &= \!\!\!\aggrsum_{c \in \Ch(s)} 
		\!\min_{(s, s_c) \in \Dec_{pc}} 
		\!\big\{\weight(s, s_c) \aggr \solW_1(s_c)\big\}, \\[-3mm]
		& \hspace{40mm}
			\textrm{ for } 
			s \in \Sset_p, 
			p \in \N_0^\stages	 
	\label{eq:TDP_recursion}
\end{aligned}
\end{equation}

\introparagraph{T-DP top-down}
Similarly to serial DP, after the bottom-up phase we get reduced sets of states
$\SsetR_\sgiter \subseteq \Sset_\sgiter$, $\DecR_{pc} \subseteq \Dec_{pc}$
and the top-1 solution $\sol_1(s_0)$ is found by a top-down
phase that follows optimal decisions.

\introparagraph{T-DP principle of optimality}
Comparing the above formulation to serial DP, we now may have multiple terminals (i.e.\ leaves in the tree) 
that are initialized with 0 cost, 
but we still have only one single root node. 
Comparing the objective functions of T-DP \cref{eq:cost_treeDP} with DP \cref{eq:costDP},
we changed the indexing to reflect the fact that each state has exactly one parent
(but not the other way around). 
Consider \cref{fig:treedp} after removing
the subtree rooted at $S_4$; then our problem degenerates to standard serial DP and we
are back at \cref{fig:cartesian}.
Contrasting the new \emph{principle of optimality} formulation in
\cref{eq:TDP_recursion} against \cref{eq:DP_recursion}, 
we now have that a minimum-cost solution contains other subtree solutions 
that achieve themselves minimum cost for their respective subproblems.

\begin{theorem}[T-DP]\label{TH:T-DP}
	\Cref{eq:TDP_recursion}
	finds an optimal solution to the problem of minimizing \cref{eq:cost_treeDP}.
\end{theorem}

\begin{proof}
We will show by an induction on the tree stages
in reverse serial order
that for all states $s \in S$:
\begin{align}
\label{eq:tdp_induction}
 \min_{\sol(s)}
 \big\{
 \aggrsum_{i \in \llceil s \rrceil}
 \weight(s_{\parent(s_i)}, s_{i})
 \big\}
	= \solW_1(s)
\end{align}
The base case for the (terminal) leaf states follows by definition from
\Cref{eq:TDP_recursion}.
For the inductive step, assume that the above holds for all
descendant states
$s_{d} \in S_d, d \in \llceil s \rrceil$ of a state $s$.
In particular for any ``child state'' $s_c$ with $s \in S_{\parent(s_c)}$:
\begin{align}
\min_{\sol(s_c)}
 \big\{
 \aggrsum_{i \in \llceil s_c \rrceil}
 \weight(s_{\parent(s_i)}, s_{i})
 \big\} 
 =
	\aggrsum_{g \in \Ch(s_c)}
 \min_{(s_c, s_g) \in \Dec_{cg}}
	\big\{
	\weight(s_c, s_g) \aggr \solW_1(s_g)
	\big\}
\end{align}

Then for any state $s \in S_p$:
\begin{align*}
&\solW_1(s)
=
	\!\!\!\!\aggrsum_{c \in \Ch(s)}
 \min_{(s, s_c) \in \Dec_{pc}}
	\!\big\{
	\weight(s, s_c) \aggr \solW_1(s_c)
	\!\big\}\\
&=
	\!\!\!\!\aggrsum_{c \in \Ch(s)}
	\min_{(s, s_c) \in \Dec_{pc}}
	\!\big\{
	\weight(s, s_c) \aggr
	\min_{\sol(s_c)}
	\!\big\{
	\!\!\!\!\aggrsum_{g \in \llceil s_c \rrceil}
	\weight(s_{\parent(g)}, s_{g})
	\big\}
	\big\}\\
&=
	\min_{s_c}
	\!\big\{
	\min_{\sol(s_c)}
	\!\big\{
	\!\!\!\!\aggrsum_{c \in \Ch(s)}
	\!\big\{
	\weight(s, s_c) \aggr
	\!\!\!\!\aggrsum_{g \in \llceil s_c \rrceil}
	\weight((s_{\parent(g)}, s_{g})
	\big\}
	\big\}
	\big\}\\
&=
 \min_{\sol(s)}
 \!\big\{
 \aggrsum_{i \in \llceil s \rrceil}
 \weight(s_{\parent(s_i)}, s_{i})
 \!\big\}
\end{align*}
Since \Cref{eq:tdp_induction} hold for any state $s$, it also holds for the starting state $s_0$, thus the theorem follows.
\end{proof}

To enumerate lower-ranked results for T-DP, we need to extend the path-based any-k algorithms.

\introparagraph{Changes to \RPDP}
All \RPDP algorithms are straightforward to extend to the tree case by following the serialized order of the stages.
Intuitively, the $i^\textrm{th}$ stage in this tree order is treated like the
$i^\textrm{th}$ stage in the path problem, except that the sets of choices are determined
by the actual parent-child edges in the tree. 
For illustration, assume a tree order as indicated
by the stage indices in \Cref{fig:treedp}. 
Given a prefix $\langle s_1 s_2 s_3 \rangle$,
the choices for $s_4 \in \Sset_4$ are not determined by $s_3$ (as they would be for a path
with stages $S_1$, $S_2$,\ldots), but by $s_1 \in \Sset_1$, because $\Sset_1$ is the parent
of $\Sset_4$ in the tree. 
In general, at stage $\Sset_c$, we have to find the successors $\Suc(s_p, s_c)$ 
where $p = \parent(s_c)$. 
Similarly, to optimally expand a prefix
$\langle s_1 \ldots s_{c-1} \rangle$ by one stage, 
we append $s_{c}$ such that $\sol_1(s_{c})$ is a subtree of  $\sol_1(s_p)$.
Thus, we can run \Cref{alg:dp-anyk} unchanged as long as we
define the choice sets based on the parent-child relationships in the
tree. 
Hence the complexity analysis in \Cref{sec:complexity} still applies
as summarized in \Cref{tab:complexity_dp}.

\introparagraph{Changes to \REDP}
Unfortunately, for \REDP the situation appears more challenging, because each state processes a
\texttt{next} call by recursively calling \texttt{next} on its children.
The challenge is to combine the lower-ranked solutions from the children
and to rank these combinations efficiently.
First, we give a high-level overview:
Consider a state $s_1 \in \Sset_1$ with children $\Sset_2$ and $\Sset_4$.
A solution rooted at $s_1$ consists of two parts: one solution rooted at the first child
$\Sset_2$ and the other at $\Sset_4$. Suppose this solution contains the $2^\textrm{nd}$-best
path from $\Sset_2$ and the $3^\textrm{rd}$-best path from $\Sset_4$---$[\Pi_2, \Pi_3]$
for short. Then the next-best solution from $s_1$ could be either $[\Pi_3, \Pi_3]$ or
$[\Pi_2, \Pi_4]$. Since any combination of child solutions $[\Pi_{j_1}, \Pi_{j_2}]$ is valid
for the parent, the problem is essentially to rank the Cartesian product space of
subtree solutions. 
This produces duplicates when directly applying
the recursive algorithm~\cite{deep19}, or requires a different approach such as \RPDP
for this Cartesian product problem to avoid duplicates. 
We adopt the latter approach.

In more detail, let $\sol_j(s, c)$ be the $j$-th best solution that starts from state $s$ but is restricted only to a single branch $c \in \Ch(s)$.
Thus, $\sol_j(s, c)$ consists of state $s$, then a state from stage $S_c$ and from there, a list of pointers to other solutions (i.e., subtrees) that have their own rank.
We write that as 
$\sol_j(s, c) = s \tree [ \sol_{j_1}(s_c, i_1), \ldots, \sol_{j_\chnum}(s_c, i_\chnum) ]$ 
for $s_c \in S_c$, $\Ch(s_c) = \{ i_1, \ldots, i_\chnum \}$ and appropriate values $j_1, \ldots, j_\chnum$.
For example, in \Cref{fig:treedp}, 
$\sol_k(s_1, 4) = s_1 \tree [\sol_{j_1}(s_{4}, 5), \sol_{j_2}(s_{4}, 6) ]$ for some values $j, j_1, j_2$.
Notice that this definition matches the one in \Cref{sec:RE} for $|\Ch(s_c)| = 1$ and
since $S_0$ always has a single child $S_1$, we have that $\sol_k(s_0) = \sol_k(s_0, 1)$ for all values of $k$.

A state $s \in S_p$ maintains one data structure per branch $c \in \Ch(s)$ for storing and comparing solutions $\sol_j(s, c)$.
At the beginning of the algorithm, we initialize it as 
$\Choices_1(s, c) \!\! = \!\! 
\{ s \tree [ \sol_1(s_c, i_1), \ldots, \sol_1(s_c, i_\chnum) ] 
\: | \: 
(s, s_c) \in \DecR_{pc} \}$. 
To process a \texttt{next} call, we pop the best solution from the data structure but unlike DP, we now have to replace it with more than one new candidates.
To compute the \texttt{next} of 
$\sol_j(s, c) = s \tree [ \sol_{j_1}(s_c, i_1), \ldots, \sol_{j_\chnum}(s_c, i_\chnum) ]$,
we have to consider as new candidates all the following solutions:
\begin{align*}
& s \tree [ \sol_{j_1 + 1}(s_c, i_1), \ldots, \sol_{j_\chnum}(s_c, i_\chnum) ] \\
&\qquad\vdots \\
& s \tree [ \sol_{j_1}(s_c, i_1), \ldots, \sol_{j_\chnum + 1}(s_c, i_\chnum) ]
\end{align*}
There are two problems associated with this.
First, we could have up to $\stages$ children, hence up to $\stages$ new candidates and each one could have size (i.e., number of pointers) up to $\stages$. 
Thus, in order to create them we would have to pay $\bigO(\stages^2)$ delay.
Second, it is easy to see that this process creates duplicates.

An elegant way to address both issues is to notice that from a specific parent node (e.g. $s_4$) all possible solutions belong to the Cartesian Product space formed by the sub-solutions of its children.
Therefore, we need to apply a ranked enumeration algorithm that enumerates this space of solutions with low delay and without duplicates.
We apply \RPDP which does not need to precompute all the 
elements from the beginning of the algorithm (we do not want to materialize all the children solutions from the start) given that they are accessed in a sorted order.
This is equivalent to running \EAGER over the Cartesian product of children solutions, except
that we do not need to sort to find the successors -- the sorted order is guaranteed by how \RECURSIVE pulls lower ranked solutions in-order.

As a result, \REDP behaves similar to the (path) DP case for nodes with a single child,
but similar to \RPDP when encountering branches. In the extreme case of star queries
(where a root stage is directly connected to all leaves),
\RECURSIVE degenerates to an \RPDP variant.

\subsection{DP over a Union of Trees (UT-DP)}
\label{sec:utdp}

We define a \emph{union of T-DP problems} as a 
set of T-DP problems where a solution to 
any of the T-DP problems is a valid solution to the UT-DP problem.
Thus, we are given a set of $u$ functions
$F = \big\{ f^{(i)} \big\}$, 
each defined over a solution space $\sol^{(i)}, i \in \N^u$.
The UT-DP problem is then to \emph{find the minimum solution across all T-DP instances}.

\introparagraph{Changes to ranked enumeration}
The necessary changes to any of our any-k algorithms are now straightforward:
We add one more top-level data structure $\Union$ that maintains the
last returned solution of each separate T-DP algorithm in a single priority queue.
Whenever a solution is popped from $\Union$, it gets replaced by the
next best solution of the corresponding T-DP problem.

\subsection{Cyclic Queries}\label{sec:cycles}

Recent work on cyclic join queries indicates that a promising approach is to reduce
the problem to the acyclic case via a \emph{decomposition algorithm}~\cite{GottlobGLS:2016}.
Extending the notion of
tree decompositions for graphs~\cite{RobertsonS:1986}, hypertree
decompositions~\cite{GottlobLS:2002} organize the
relations into ``bags'' and arrange the bags into a
tree~\cite{DBLP:conf/sebd/Scarcello18}. Each decomposition is associated with
a width parameter that captures the degree of acyclicity in the query
and affects the complexity of subsequent evaluation: smaller width implies
lower time complexity.
\emph{Our approach is orthogonal to the decomposition algorithm used and it adds
ranked enumeration capability virtually ``for free.''}

The state-of-the-art decomposition algorithms rely on the submodular width $\subw(Q)$
of a query $Q$. Marx~\cite{Marx:2013:THP:2555516.2535926} describes an algorithm that runs
in $\O(f(|Q|)n^{(2+\delta)\subw(Q)})$ for $\delta > 0$ and a function $f$ that depends
only on query size. 
\PANDA~\cite{khamis17panda} runs in
$\bigO(f_1(|Q|)n^{\subw(Q)}(\log n)^{f_2(|Q|)})$ for query-dependent functions
$f_1$ and $f_2$. Since this is an active research area,
we expect these algorithms to be improved and we believe our framework is general
enough to accommodate future decomposition algorithms. Sufficient conditions for
applicability of our approach and for achieving optimal delay are, respectively,
(1) the full output of $Q$ is the union of the output produced by the trees
in the decomposition and
(2) the number of trees depends only on query size $|Q|$.
Both are satisfied by current decompositions and it is hard to imagine how
this would change in the future. 

We can execute any decomposition algorithm almost as a blackbox to create
a union of acyclic queries to which we then apply our UT-DP framework.
\emph{However, there are subtle challenges:} For correctness, we have to
(1) properly compute the weights of tuples in the bags (i.e., tree nodes) and
(2) deal with possible output duplicates when a decomposition creates
multiple trees. For (1), we slightly modify the decomposition algorithm
to track the lineage for bags at the schema level: We only need to know from which input
relation a tuple originates and if that relation's weight values had already been accounted
for by another bag that is a descendent in the tree structure.

For (2), note that if all output tuples have distinct weights,
then an output tuple's duplicates will be produced by
our any-k algorithm one right after the other, making it trivial to eliminate them on-the-fly.
Since the number of trees depends only query size $|Q|$, total delay induced by duplicate filtering
is $\O(1)$ (data complexity). When different output tuples can have the same weight,
we break ties using lexicographic ordering on their witnesses
, as we describe in \cref{sec:ties}.

\subsubsection{Simple Cycle Decomposition}
\label{sec:cycle_details}

For $\stages$-cycle queries $Q_{C\stages}$ we use the standard
decomposition~\cite{Khamis0S16:pandaArxiv,DBLP:conf/sebd/Scarcello18},
which was pioneered
by Alon et al.~\cite{Alon1997} in the context of graph-pattern queries. It does
not produce output duplicates and achieves $\O(n^{2-1/\lceil \ell/2 \rceil})$ for $\TTF$.
On the other hand, for a worst-case optimal join algorithm such as \NPRR~\cite{ngo2018worst} or \GenJoin~\cite{Ngo:2014:SSB:2590989.2590991},
$\TTF$ is $\O(n^{\ell/2})$.
We show in
\Cref{sec:NPRR_and_TTF}
that those algorithms can indeed not be modified to overcome this problem.

\begin{figure}[h]
    \centering
    \begin{subfigure}{0.31\linewidth}
        \centering
        \includegraphics[page=1, scale=0.38]{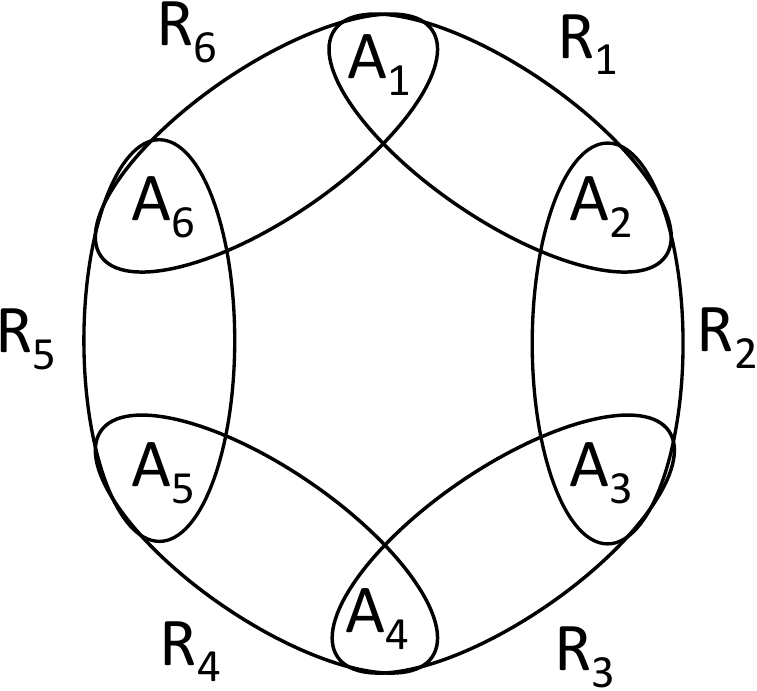}
        \caption{Query $Q_{C6}$}\label{fig:6cycle}
    \end{subfigure}%
    \begin{subfigure}{0.33\linewidth}
        \centering
        \includegraphics[page=2, scale=0.37]{figs/6cycle.pdf}
        \caption{Heavy decomp.}\label{fig:heavy_dec}
    \end{subfigure}%
    \begin{subfigure}{0.34\linewidth}
        \centering
        \includegraphics[page=3, scale=0.38]{figs/6cycle.pdf}
        \caption{All-light decomp.}\label{fig:light_dec}
    \end{subfigure}%
    \caption{Simple cycle of length 6 and two decompositions.}
    \label{fig:cycle}
\end{figure}

We illustrate with 6-cycle query $Q_{C6}$, depicted in \cref{fig:6cycle}.
(Here $x_i$ in
{{\cref{ex:path_cycle}}}
is replaced by $A_i$ to better distinguish
the concrete 6-cycle from the general $\ell$-cycle case.)
First, we horizontally partition each relation $R_i$ into $R_{iH}$ and $R_{iL}$ according to
whether the tuples are {\em heavy} or {\em light}: $R_{iH}$ receives all heavy
tuples; $R_{iL}$ the others (light ones).
A tuple $t$ in relation $R_i$ is heavy~\cite{Alon1997} iff value
$t.A_i$ occurs at least $n^{2/\stages}=n^{1/3}$ times in column $R_i.A_i$.
Then the maximum number of distinct heavy values in a column is at most $n^{1-2/\stages}=n^{2/3}$.
We create $\ell+1 = 7$ database partitions:
{
\renewcommand{\tabcolsep}{0.5mm}
\begin{center}
\begin{tabular}{@{\hspace{40pt}} cccccccc}
$T_1 =$   & $\{R_{1H},$ & $R_{2},$ & $R_{3},$ & $R_{4},$ & $R_{5},$ & $R_{6} \}$ \\ 
$T_2 =$   & $\{R_{1L},$ & $R_{2H},$ & $R_{3},$ & $R_{4},$ & $R_{5},$ & $R_{6} \}$ \\ 
         & \multicolumn{2}{c}{$\vdots$} \\
$T_6 =$   & $\{R_{1L},$ & $R_{2L},$ & $R_{3L},$ & $R_{4L},$ & $R_{5L},$ & $R_{6H} \}$ \\ 
$T_7 =$   & $\{R_{1L},$ & $R_{2L},$ & $R_{3L},$ & $R_{4L},$ & $R_{5L},$ & $R_{6L} \}$ \\ 
\end{tabular} 
\end{center}
}

It is easy to verify that each output tuple will be produced by exactly one partition.
The first $\ell=6$ partitions use a ``heavy'' tree decomposition where the cycle is
``broken'' at the heavy attribute. For instance, $A_1$ is the heavy attribute for $R_{1H}$;
the resulting tree is shown in \cref{fig:heavy_dec}. Each tree node is a bag whose content
is materialized in time 
$\bigO(n^{2-2/6}) = \bigO(n^{5/3})$
by appropriately joining the corresponding relations. Consider top
bag $(A_1, A_2, A_3)$ derived from relations $R_1$ and $R_2$. Since $R_{1H}$
contains at most $n^{2/3}$ distinct values, we can compute the bag with a simple
nested-loop join. It goes through all pairs $A_1-(A_2, A_3)$ of distinct heavy values of
$A_1$ in $R_{1H}$ and tuples $(A_2, A_3)$ in $R_2$. Since there are at most $n^{2/3}$
distinct heavy $A_1$-values and $n$ tuples in $R_2$, there are $\bigO(n^{5/3})$ such pairs.
For each pair we can verify in $\bigO(1)$ if the corresponding $(A_1, A_2)$ combination
exists in $R_{1H}$. The other bag computations and heavy decompositions are
analogous.

For $T_7$, which only contains light partitions, we use a different
``all-light'' tree decomposition shown in \cref{fig:light_dec}.
It materializes each bag with a join chain: For each tuple in one ``endpoint'' relation,
find all matches in the next relation, and so on. Consider $(A_1, A_2, A_3, A_4)$,
which is derived from $R_{1L}$, $R_{2L}$, and $R_{3L}$. For each tuple $t_1 \in R_{1L}$, we
find all matches in $R_{2L}$, then join with $R_{3L}$. There are $\O(n)$ tuples
in $R_{1L}$, but since all relations are light, each of them joins with at most
$n^{1/3}$ in the next relation. Hence total complexity for materializing
bag $(A_1, A_2, A_3, A_4)$ is
$\O(n \cdot n^{1/3} \cdot n^{1/3}) = \O(n^{5/3})$.

In general, the tree decomposition for an $\ell$-cycle query produces a union of $\ell+1$ trees, out of
which $\ell$ use the heavy decomposition and 1 uses the light one. By setting the heavy-light threshold to
$n^{2/\ell}$, we can materialize all bags of all trees in time $\bigO(n^{2-2/\ell})$. Note
that the number of tuples in a bag is $\bigO(n^{2-2/\ell})$.
Any such union of trees can be handled by our UT-DP framework.

\subsection{Putting everything together}
\label{sec:main_res}

Our main result follows from the above analysis when using \HEAP for the acyclic CQ base case:

\begin{theorem}\label{TH:MAIN}
Given a decomposition algorithm $\mathcal{A}$ that takes time $T(\mathcal{A})$ and space $S(\mathcal{A})$,
ranked enumeration of the results of a full conjunctive query can be performed with 
\new{
$\TTF = \O(T(\mathcal{A}))$,
$\Del(k) = \O(\log k)$,
}
and $\MEM(k) = \O(S(\mathcal{A}) + k)$ in data complexity.
\end{theorem}

\begin{proof}
In all cases we use the \HEAP algorithm.

First, consider the case of path queries and recall that we had made two assumptions in \cref{sec:complexity} for simplicity: (1) that the arities of the relations are bounded, thus $|Q| = \stages$ and (2) that
the operations $\oplus$ and $\otimes$ of the selective dioid take $\gamma = \O(1)$.
Extending the analysis of \cref{sec:complexity} to the general case of unbounded arities
and $\gamma = f(|Q|)$, we get 
\new{
$\TTF = \O(f(|Q|)n)$,
$\Del(k) = O(f(|Q|) (\log k + |Q|))$,
}
and $\MEM(k) = \O(|Q|n + k|Q|)$.
Thus, 
\new{
$\TTF = \O(n)$,
$\Del(k) = O(\log k)$,
}
and $\MEM(k) = \O(n + k)$ in data complexity.

For tree queries, 
\new{
$TTF$ and $\MEM(k)$
}
stay the same, however 
$\Del(k)$
has an additional term that is quadratic in $\stages$ in the absence of the inverse element as we discuss in \cref{sec:inverse}.
Still, the data complexity remains the same as above.

For cyclic queries, first we apply the decomposition algorithm $\mathcal{A}$ to obtain a set of acyclic queries $\boldsymbol{\mathcal{Q}}$.
The number of acyclic queries we get is $g(|Q|)$ (according to our 
assumptions on $\mathcal{A}$ in \cref{sec:cycles}) and also note that the ``bags'' (i.e., the derived input relations) in each acyclic query can be at most the number of attributes $m$~\cite{khamis17panda}. 
Then, we run our UT-DP framework on top of $\boldsymbol{\mathcal{Q}}$ using \HEAP.
The top-level priority queue $\Union$ takes $\O(\log g(|Q|))$ to pop an element or $\O(1)$ in data complexity and then $\O(\log k)$ to pull the next result from the corresponding tree.
To avoid duplicates, we apply our construction of \cref{sec:ties} which imposes a lexicographic order and thus, increases the complexity of pulling results from each acyclic query by an $\O(|Q|)$ factor.
Also, to filter the duplicates we have to spend an additional $\O(g(|Q|))$ factor in delay since the number of duplicates cannot exceed the number of acyclic queries. 

Overall, 
\new{
$\TTF$ is $\O(T(\mathcal{A})$ and $\Del(k)$ is $\O(\log k)$ in data complexity.
}
For the space consumption, note that the total size of the derived input relations of $\boldsymbol{\mathcal{Q}}$ is bounded by $S(\mathcal{A})$ and our framework only adds an $\O(k)$ term in data complexity.
\end{proof}

\section{Ranking Function}
\label{sec:ranking_functions}

We now look deeper into the ranking functions that our framework supports.

\subsection{Attribute weights}
\label{sec:attribute_weights}

In order to keep our formalism clean and easy to follow, we focused only on weights on tuples.
It is however straightforward to also handle weights on attributes
by adding unary tables with weights on single columns. We illustrate next.

\begin{example}[Attribute weights]
	Consider the query $Q(x,y) \datarule R(x,y)$
	over a database $R(A,B)$ with weight function
	$w_R : r \in R \rightarrow \R^+$ on tuples, and two weight functions
	$w_A : a \in \adom(A) \rightarrow \R^+$, and
	$w_b : b \in \adom(B) \rightarrow \R^+$ on attributes.
	The problem can then be translated into one with only weight functions on tuples
	by introducing two new relations
	$S(A) = \adom(A)$ and $T(B) = \adom(B)$ with associated weight functions
	$w_S : s \in S \rightarrow \R^+$, and
	$w_T : t \in T \rightarrow \R^+$ and translated query
	$Q'(x,y) \datarule R(x,y), S(A), T(B)$.
\end{example}

\subsection{On the existence of the inverse (groups vs.\ monoids)}
\label{sec:inverse}

When we presented and analysed the algorithms in
\cref{sec:DPalgorithms},
we assumed for simplicity the existence of an inverse element for the $\otimes$ operator of the selective dioid $(W, \oplus, \otimes, \0, \1)$.
We now discuss what happens in the absence of that inverse element.
We start with some definitions.

\introparagraph{The inverse of an operation}
An \emph{Abelian group} is a commutative monoid $(W, \otimes, \1)$
for which there exists an \emph{inverse} for each element.
More formally, for each $x$ in $W$, there is an inverse element $x'$ in $W$ such that $x \otimes x^{-1} = x^{-1} \otimes x = \1$.
We also write 
$y \oslash x$ 
as short form for 
$y \otimes x^{-1}$ (i.e., ``$\oslash y$" 
composes $y$ with the inverse of $x$).

\begin{savenotes} %
\begin{example}[Groups vs. Monoids]
The archetypical Abelian group is $(\R, +, 0)$, i.e.\ the real numbers with addition.
An example commutative monoid that is not a group (and thus has no inverse in general) is logical conjunction: 
$(\{0, 1\}, \wedge, 1)$.\footnote{We write $1$ for \textsf{true} and $0$ for \textsf{false}.}
Here, $1$ is the identity 
element because 
$x \wedge 1 = 1 \wedge x = x$ for $x \in \{0, 1\}$.
However, for $0$ there is no inverse $0'$ such that $0 \wedge 0' = 0' \wedge 0 = 1$.
Another operation that has no inverse is the minimum: 
$(\R, \min, \infty)$. 
Here $\infty$ is the identity element because 
$\min (x, \infty) = \min (\infty, x) = x$ for $x \in \R$.
However, for no $x$ there is an inverse $x'$ such that $\min(x, x') = \infty$.
\end{example}
\end{savenotes}

In general, the inverse element allows us to 
perform calculations that would be otherwise impossible.
Thus, it can be used to short-circuit long calculations
by reusing prior results.

\begin{example}[Benefit of inverse elements]
Consider a commutative monoid $(W, \otimes, \1)$ and the composition $x \otimes y = z$.
Assume we are given $z$ and $y$ and would like to calculate $x$. 
Then this is only possible in general, if each element has an inverse; in other words, if the monoid is actually a group.
To illustrate this issue, consider first the real numbers with addition
$(\R, +, 0)$
and assume
$(x,y,z) = (1, 2, 3)$. 
Then we can calculate $x=1$ from $z=3$ and $y=2$ as $x = z + y' = 3 + (-2) = 1$.
Next consider logical conjunction $(\{0, 1\}, \wedge, 1)$ with 
$(x,y,z) = (1, 0, 0)$. Then we cannot calculate $x=1$ from $z=0$ and $y=0$ (both $x=1$ or $x=0$ are possible).
Similarly, consider minimum $(\R, \min, \infty)$ with 
$(x,y,z) = (3, 2, 2)$. Then we cannot calculate $x=3$ from $z=2$ and $y=2$ ($x$ could be any value in $[2, \infty]$).
\end{example}

\introparagraph{Ranked enumeration without an inverse}
In the context of our algorithms, the inverse element is not a hard requirement. However, it can help simplify and speed up certain variants of ranked numeration.
First, notice that \RECURSIVE never uses an inverse since it always constructs solutions by appending one state to a suffix or a list of subtrees
(see \Cref{rea_line:init,rea_line:insert} of \Cref{alg:rea}).
Therefore the cost of the solution can easily be calculated by applying $\otimes$.

For the \emph{\RPDP algorithms over T-DP}, 
we have to make a minor adjustment which will incur an additional $\O(\stages^2)$ delay term in the complexities presented in \Cref{tab:complexity_dp}. 
To illustrate this, we use the terminology of \Cref{alg:dp-anyk}.
First, notice that in the path case the inverse element is not needed. 
In \Cref{line:newCand},
a new candidate is inserted into the priority queue with weight 
$\textrm{solution.prefixWeight} \otimes w(\textrm{tail}, s) \otimes \solW_1(s)$.
Intuitively this means that the ``future cost'' of the candidate (when optimally expanded in a solution)
is the weight of its prefix composed with the weight of the new decision 
and with the optimal weight from there onward.
Thus, we are able in $\O(1)$ to calculate the weight it will have if we expand it without actually spending $\O(\stages)$ to expand it.
In T-DP this calculation is not possible because the weight of the optimal extension from $s$ (which was the $\solW_1(s)$ term in the path case) 
involves subtrees that are not in $\lceil s \rceil$, thus it
is not available at state $s$.

One way to circumvent this is to use the inverse and still get an $\O(1)$ computation per new candidate.
Let $\textrm{prevWeight}$ be the weight of the prefix we popped from $\Cand$ in the current iteration.
Then, the weight of the new candidate is $\textrm{prevWeight} \oslash w(\textrm{tail}, \textrm{last}) \oslash \solW_1(last) \otimes w(\textrm{tail}, s) \otimes \solW_1(s)$.
Intuitively, this means that to compute the weight of the new candidate we 
``subtract'' the old decision weight and the optimal subtree weight of its target state and we ``add'' the new ones. 
If we don't have the inverse element, then the above computation is not possible;
instead, we expand each of the $\O(\stages)$ new candidates before inserting them into the $\Cand$ priority queue and traverse each one of them to compose the decision weights (as in \Cref{eq:cost_treeDP}).
This costs $\O(\stages^2)$ in total because we have $\O(\stages)$ candidates and each one has $\O(\stages)$ size.

\subsection{Tie-breaking the output}
\label{sec:ties}

We now elaborate on how to break ties between result weights consistently.
This is a key element for handling cyclic queries with existing decomposition algorithms.
Recall from \Cref{sec:cycles} that we could use a decomposition  
(e.g.\ \PANDA~\cite{khamis17panda}) that generates 
a set of trees whose outputs are not necessarily disjoint. 
Thus the same result tuple could potentially be produced by multiple trees. 
It is easy to detect and remove those duplicates
if they arrive in consecutive order (the step is then linear in number of trees, but constant in data complexity).
This consecutive arrival is guaranteed if there are no ties in the weights of output tuples.
If there are ties, however, the arrival between identical output tuples could be in the order of number of output tuples produced so far.
To see why, imagine an extreme scenario where all the output tuples have the same weight and duplicates arrive in arbitrary order; 
in that case, 
the delay between consecutive results could be in the order of $k$, i.e.\ in the order of the number of already seen output tuples. 
For instance, assume 5 output tuples $\{r, s, t, u, v\}$ 
with the same weight, and assume 10 tree decompositions. 
Then a possible enumeration 
could be 
$(r, s, t, u, r, r, r, r, r, r, r, r, r, s, s, \ldots, t, t, u, \ldots)$.
To prevent this, we redefine our ranking function slightly 
so that it breaks ties in a consistent way 
and thus no two output tuples will have the same weight.
This guarantees again that only duplicates can have the same weight and hence all
the duplicates of a tuple arrive consecutively.

Intuitively, we add a second dimension to our ranking function that captures a lexicographic order on the input tuples. 
Whenever two weights are equal, the tie will be broken by the value of that extra dimension, 
ensuring that only identical results have the same overall weight. 
In the end, the true weight can be recovered by looking only at the first dimension of the weight function.

Given two partially ordered sets $A$ and $B$, 
the lexicographic order on the Cartesian product $A \times B$ is defined as
\begin{align*}
    (a,b) \leq (a',b') \textrm{ iff } a < a' \textrm{ or } (a = a' \textrm{ and } b \leq b')
\end{align*}
It is well known that this order is a total order if and only if the factors of the Cartesian product are totally ordered. 

However, what is less known is that this order is a total order even if the first factor is just a 
\emph{total preorder} (also called preference relation). 
Recall that a total preorder is reflective ($a \leq a$), 
transitive (if $a \leq b$ and $b \leq c$ then $a \leq c$), 
complete (for every $a,b$, $a \leq b$ or $b \leq a$), 
however not necessarily antisymmetric
($a \leq b$ and $b \leq a$ does not imply $a = b$). 
To illustrate this point, 
consider binary output tuples with domain $\{a,b,c,d,e\}$
under the attribute weight model (\cref{sec:attribute_weights}).
Assume a total preorder on the domain values with $a=b<c<d<e$.
Then the lexicographic order for three particular output tuples could be $(a,c) \rightarrow (b,d) \rightarrow (a,e)$.
Thus the three tuples imply a total order although the domain values of the first column do not.

We now show how to use this property to force our any-k enumeration to enumerate the same output tuple
with a delay that depends only on the query
even if we use a decomposition method (such as PANDA) that is not disjoint.
The key idea is to force that each output tuple will be enumerated \emph{consecutively}
even if there are ties, i.e.\ multiple output tuples with the same weights.

Assume that for an output tuple $r$, the original ranking function $w(r)$ was defined with operators $\oplus, \otimes$ and a total order $\leq$. 
Then the new ranking function is the Cartesian product $w'(r) = (w(r), \vec{id}(r))$, 
with $\vec{id}(r)$ capturing a lexicographic order as in 
\cref{sec:ranking},
and the following two operators:

\begin{enumerate}[nolistsep]

\item 
$w'(r_1) \otimes w'(r_2) = w'(r_1)$ iff 
$(w(r_1) \leq w(r_2) \wedge w(r_2) 
	\nleq w(r_1)) \vee (w(r_1) \leq w(r_2) \wedge w(r_2) \leq w(r_1) \wedge \vec{id}(r_1) 
	\preceq_\lex \vec{id}(r_2))$, 
else $w'(r_1) \otimes w'(r_2) = w'(r_2)$
and

\item
$w'(r_1) \otimes w'(r_2) = (w(r_1) \otimes w(r_2), \vec{id}(r_1) \otimes_\lex \vec{id}(r_2))$.
\end{enumerate}

As can be easily seen, the new ranking function is also defined over a \emph{selective dioid}, 
and our any-k algorithms immediately apply.

\subsection{Other examples of ranking functions}

Throughout the main paper we focused on the ranking function that consists of the
operators $\min$ (which is selective) and addition. These correspond to
the \emph{tropical semiring} $(\R_{\min}, \min, +, \infty, 0)$ 
with $\R_{\min} \define \R \cup \{\infty\}$, which is an instance of a \emph{selective dioid} (see the definition in \Cref{sec:ranking}). 
Under this perspective, Bellman's famous \emph{principle of optimality} discussed in \cref{sec:DPtoAnyK}
is a re-statement of the more general \emph{distributivity of addition over minimization}:
$\min(x+z, y+z) = \min(x,y)+z$.

Another example of a selective dioid that our approach works on is the Boolean semiring $(\{0, 1\}, \vee, \wedge, 0, 1)$,
where the disjunction is also selective.
Interestingly, our algorithms can also perform standard query evaluation by inverting
the order to $1 \leq 0$.
Since maintaining priority queues (or sorting) with $\{0, 1\}$ elements takes linear time,
it follows that our algorithms can enumerate answers to a 4-cycle query with $\TTF = \O(n^{1.5})$
and $\TTL = \O(n^{1.5} + |\mathrm{out}|)$. 
These match the best known algorithms for Boolean and full query evaluation which use the submodular width ($\subw(Q_{C4}) = 1.5$)
For worst-case output instances, i.e., $|\mathrm{out}| = n^2$
we also match the AGM bound, i.e., our algorithm, like \NPRR, is worst-case optimal.

Other examples of selective dioids that we can use are $(\R_{\max}, \max, +, -\infty, 0)$ with $\R_{\max} \define \R \cup \{-\infty\}$ or $(\R_{\geq 0}, \max, \times, 0, 1)$ with $\R_{\geq 0} \define [0, \infty)$.
The former finds the heaviest tuples or equivalently, the ``longest'' paths in a graph (according to the input weights).
The latter can be used to simulate bag semantics; if the weight of each input tuple reflects its multiplicity in the input relation, then by using that ranking function we first get the output tuple with the biggest multiplicity in the result and its output weight is that multiplicity.

\section{Experiments}
\label{sec:experiments}

Since asymptotic complexity only tells part of the story, we compare all algorithms
in terms of actual running time.

\introparagraph{Algorithms}
All algorithms are implemented in the same Java environment and use the same data structures
for the same functionality. We compare:
(1) \RECURSIVE representing the \REDP approach,
(2) \HEAP, (3) \LAZY~\cite{chang15enumeration}, (4) \EAGER, (5) \MIN~\cite{yang2018any}
representing the \RPDP approach,
and (6) \NAIVE, which computes the full result using the Yannakakis
algorithm~\cite{DBLP:conf/vldb/Yannakakis81} for acyclic queries and \NPRR~\cite{ngo2018worst} 
for cyclic queries, both followed by sorting.

\introparagraph{Queries}
We explore \emph{paths}, \emph{stars}, and \emph{simple cycles} over binary relations. 
The SQL queries are listed in 
\cref{sec:sql_queries}.
A path is the simplest acyclic query, making it ideal for studying core differences
between the algorithms.
The star represents a typical join in a data warehouse and by treating it as a single root
(the center) with many children, we can study the impact of node degree.
The simple cycles apply our decomposition method as described in \Cref{sec:cycles}.

\introparagraph{Synthetic data}
Our goal for experiments with {synthetic data} 
is to create input with regular structure that 
allows us to
identify and explain the core
differences between the algorithms.
For path and star queries, we 
create tuples 
with values uniformly sampled from the domain $\N_1^{n/10}$. 
That way, tuples join with $10$ others in the next relation, on average.
For cycles, we follow a construction by \cite{ngo2018worst} that creates a worst-case output:
every relation consists of $n/2$ tuples of the form $(0, i)$ and $n/2$ of the
form $(i, 0)$ where $i$ takes all the values in $\N_1^{n/2}$.
Tuple weights are real numbers uniformly drawn from $[0, 10000]$.

\introparagraph{Real Data}
We use two {real networks}.
In Bitcoin OTC \cite{Bitcoin_dataset2,Bitcoin_dataset1}, edges
have weights representing the degree of trust between users.
Twitter \cite{Twitter_dataset} edges model followership among users.
Edge weight is set to the sum of the PageRanks~\cite{brin98pagerank}
of both endpoints. 
To control input size, we only retain edges between users whose IDs are below a given
threshold.
Since the cycle queries are more expensive, we run them on a
smaller sample (TwitterS) than the path queries (TwitterL).
\Cref{tab:datasets} summarizes relevant statistics.
Note that the size of our relations $n$ is equal to the number of edges.

\begin{figure}[!tb]
\footnotesize
\renewcommand{\tabcolsep}{1.0mm}
\begin{center}
\begin{tabular}{|l|r|r|c|c|}
\hline
Dataset                                                 & Nodes & Edges & Max/Avg Degree & Weights \\ \hline
Bitcoin \cite{Bitcoin_dataset2,Bitcoin_dataset1}        & 5,881 & 35,592 & 1,298 / 12.1 & Provided  \\
TwitterS \cite{Twitter_dataset}                          & 8,000 & 87,687 & 6,093 / 21.9 & PageRank \\
TwitterL \cite{Twitter_dataset}                          & 80,000 & 2,250,298 & 22,072 / 56.3 & PageRank \\
\hline
\end{tabular} 
\caption{Datasets used for experiments with real data.}
\label{tab:datasets}
\end{center}
\vspace{-4mm}
\end{figure}

\begin{figure*}[h]

    \centering
    \begin{subfigure}{\linewidth}
        \centering
        \includegraphics[width=0.7\linewidth]{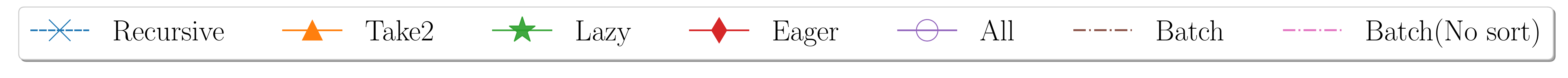}
    \end{subfigure}
    \vspace{-3mm}

    \begin{subfigure}{0.24\linewidth}
        \centering
        \includegraphics[width=\linewidth]{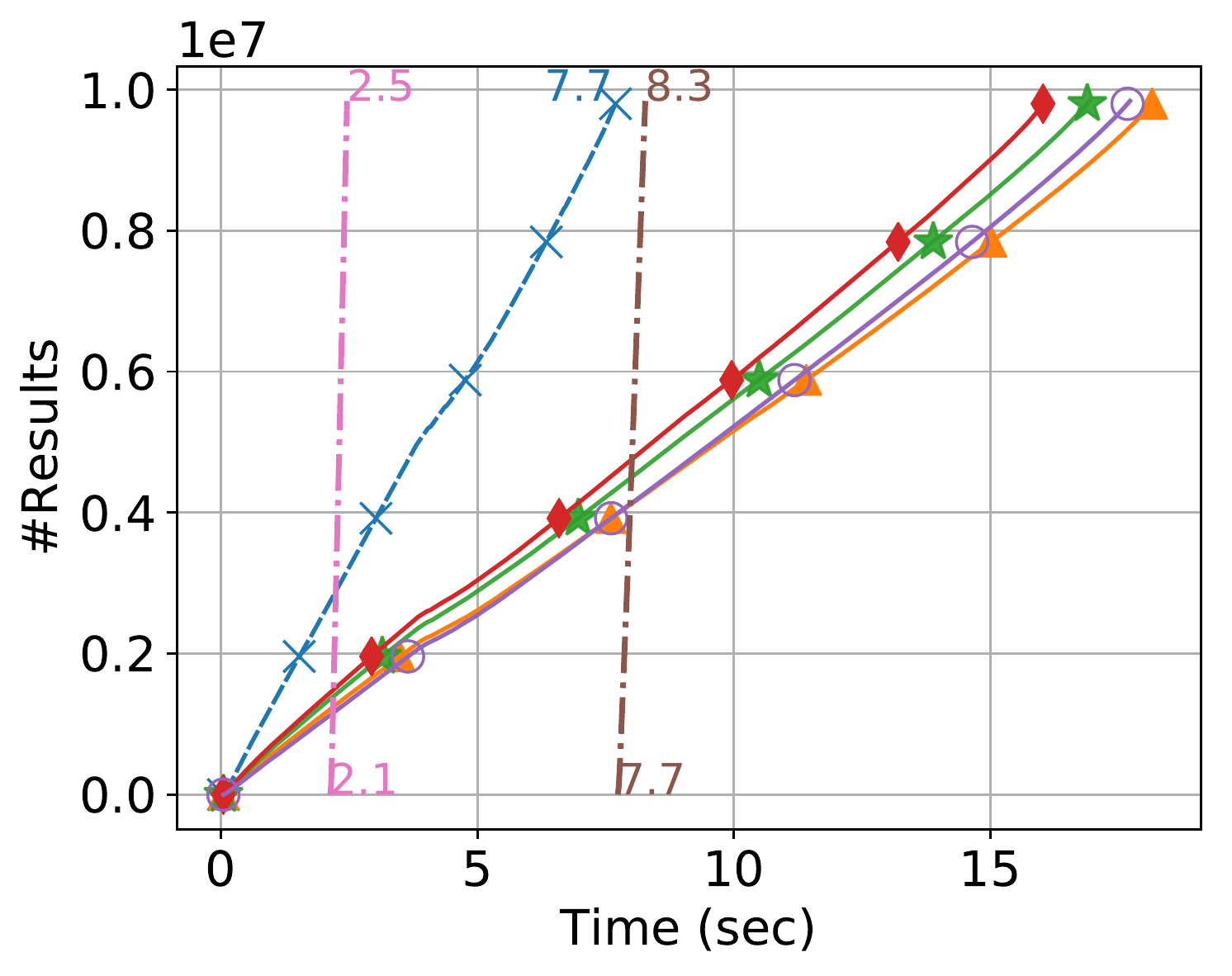}
        \caption{4-Path Synthetic\\($n\!=\!10^4$):\\All $\sim \!10^{7}$ results.}
		\label{exp:4path_syn_small}
    \end{subfigure}%
    \hfill
    \begin{subfigure}{0.24\linewidth}
        \centering
        \includegraphics[width=\linewidth]{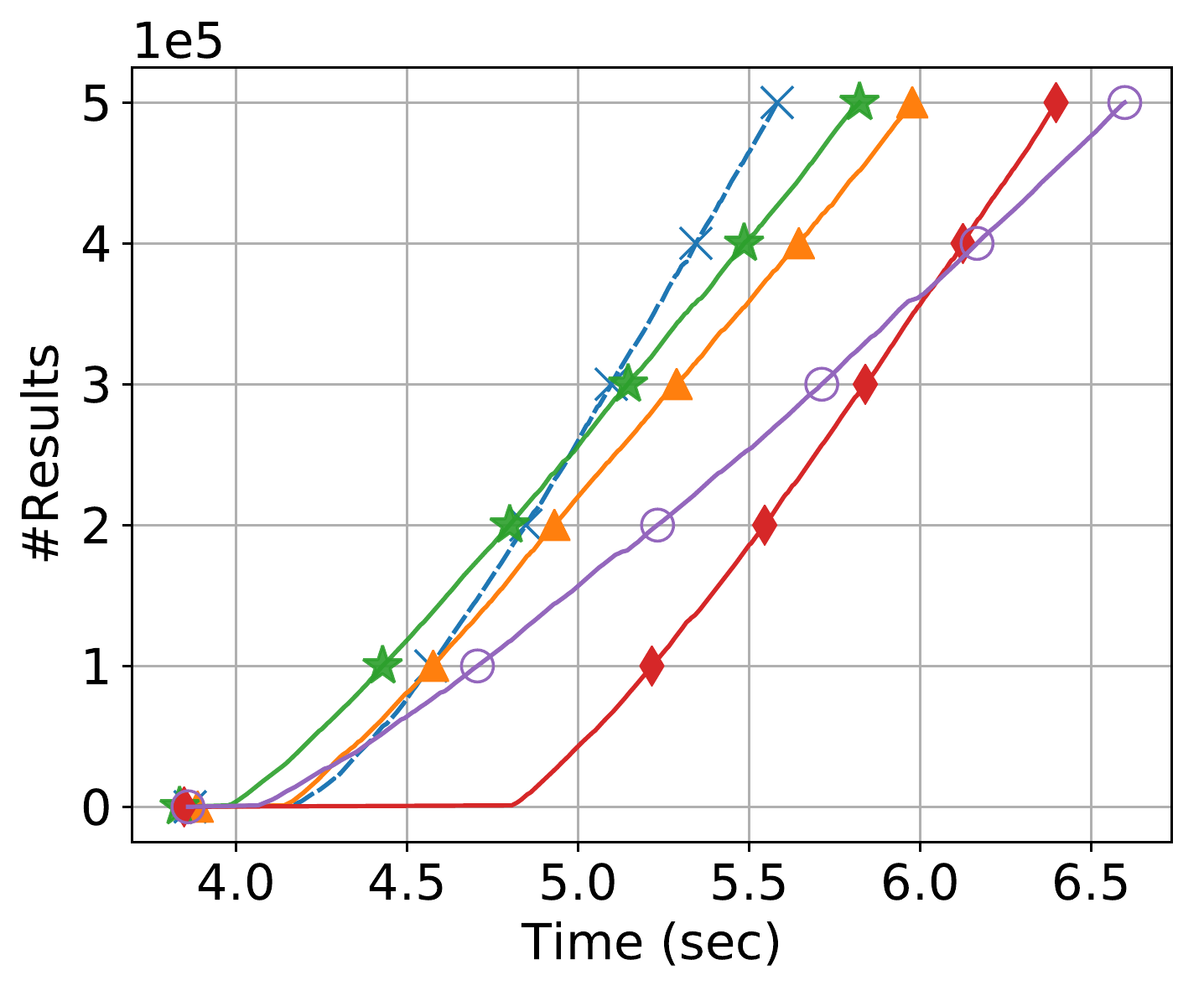}
        \caption{4-Path Synthetic\\($n\!=\!10^6$):\\Top $n/2$ of $\sim 6 \cdot 10^{9}$ results.}
		\label{exp:4path_syn_large}
    \end{subfigure}%
    \hfill
    \begin{subfigure}{0.24\linewidth}
        \centering
        \includegraphics[width=\linewidth]{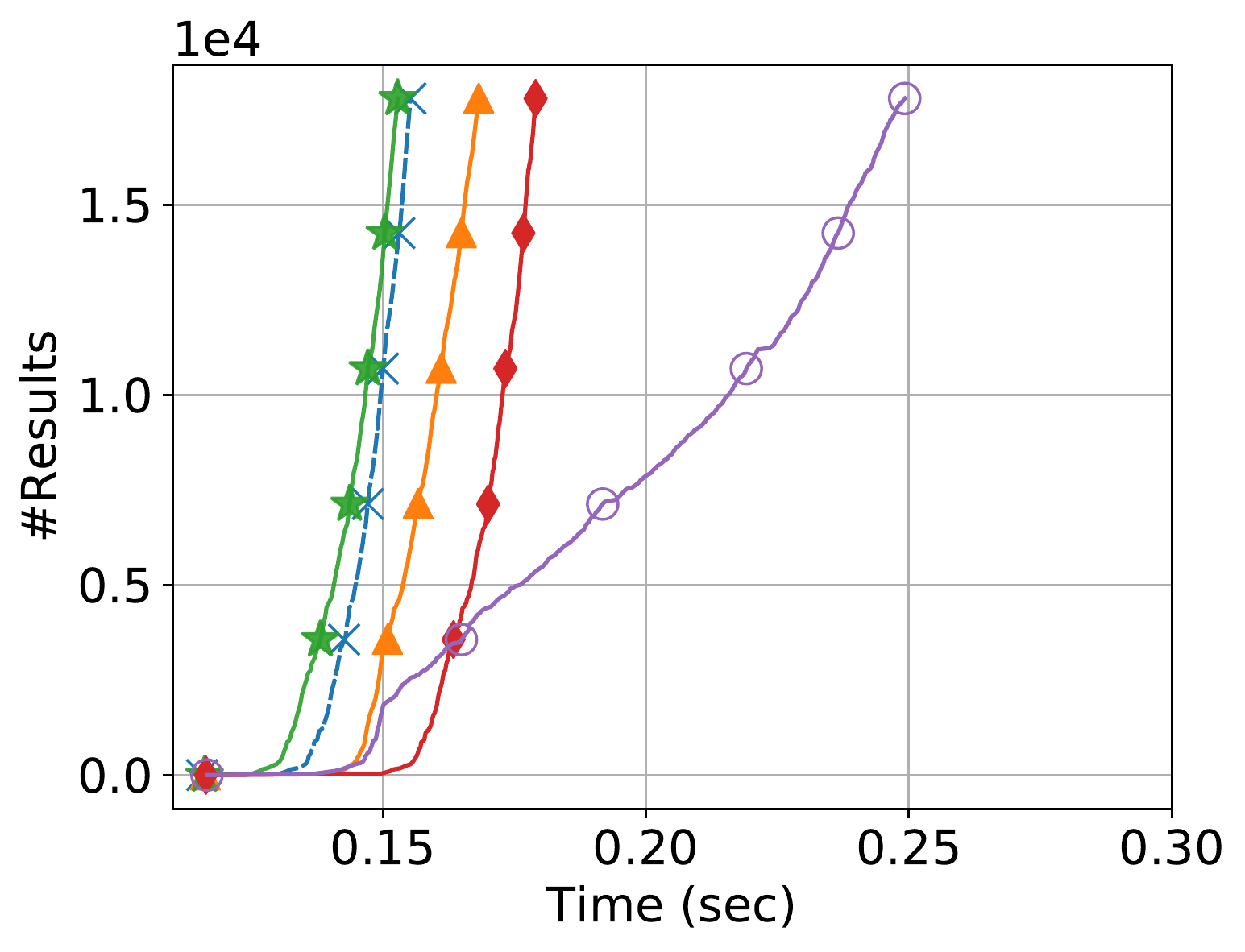}
        \caption{4-Path Bitcoin\\($n\!\sim\!3.6 \!\cdot\! 10^4$):\\Top $n/2$ of $\sim 4 \cdot 10^9$ results.}
		\label{exp:4path_bitcoin}
    \end{subfigure}%
    \hfill
    \begin{subfigure}{0.24\linewidth}
        \centering
        \includegraphics[width=\linewidth]{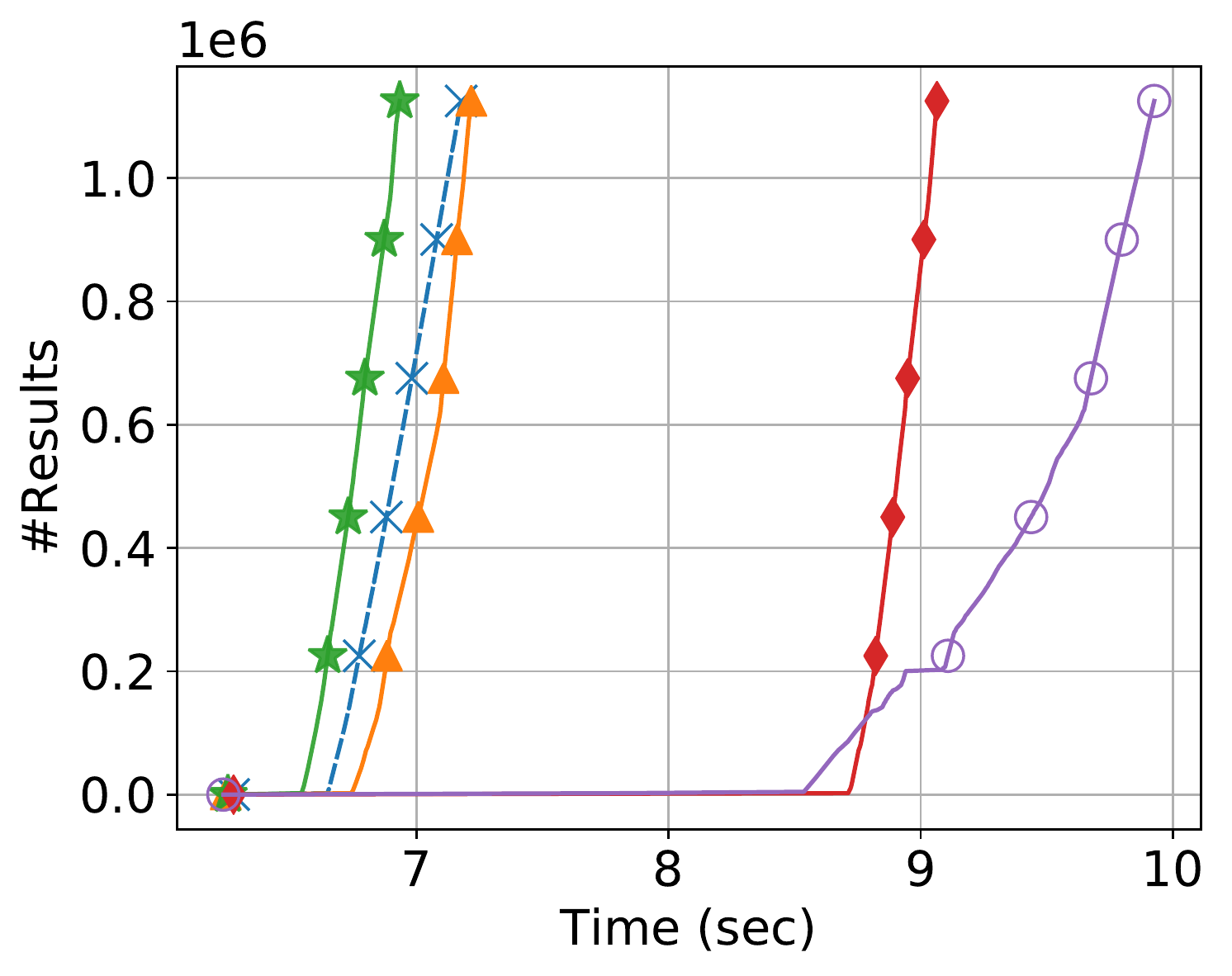}
        \caption{4-Path TwitterL\\($n\!\sim\!2.3 \!\cdot\! 10^6$):\\Top $n/2$ of $\sim \!10^{15}$ results.}
		\label{exp:4path_twitter}
    \end{subfigure}
    \vspace{-1mm}

    \begin{subfigure}{0.24\linewidth}
        \centering
        \includegraphics[width=\linewidth]{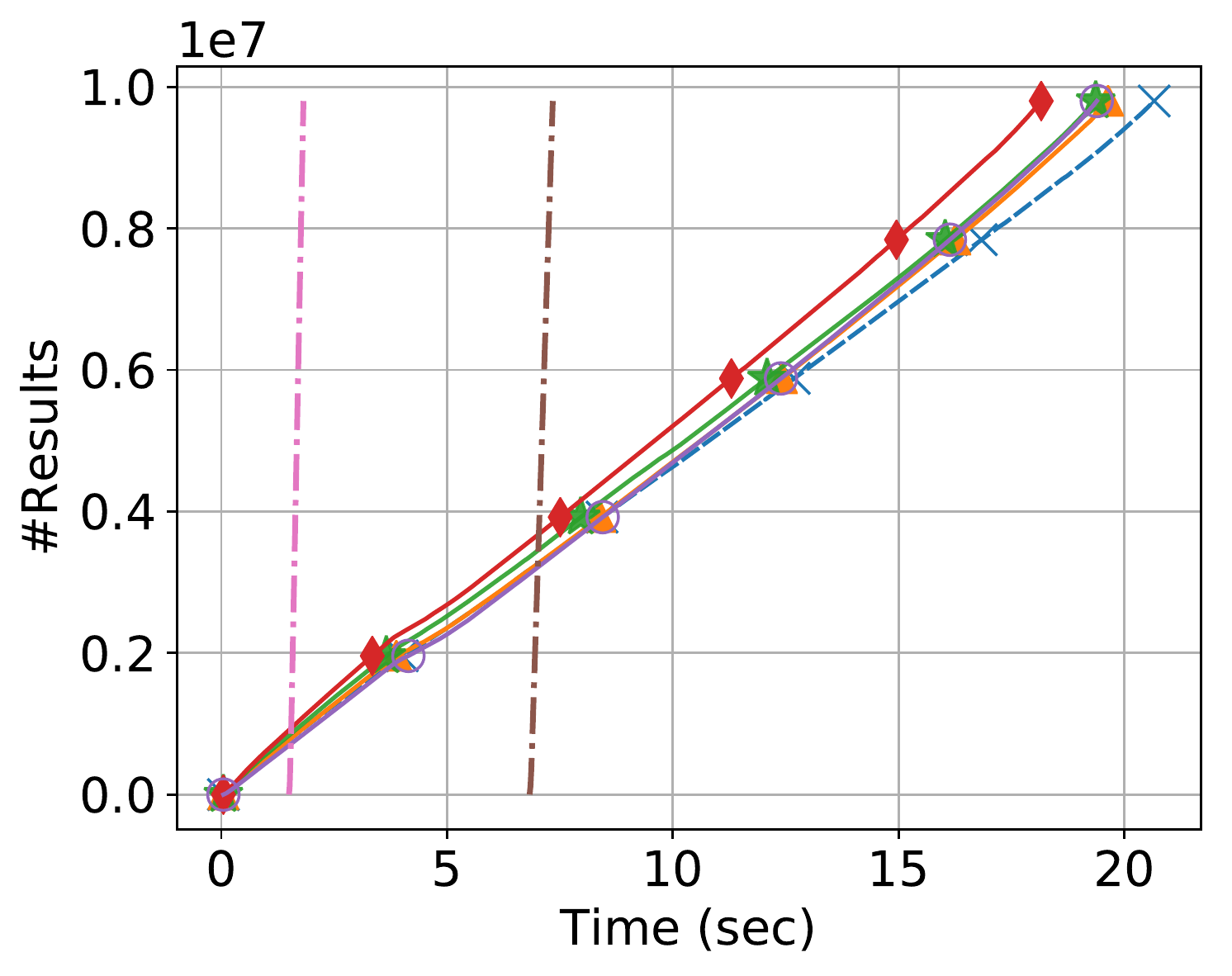}
        \caption{4-Star Synthetic\\($n\!=\!10^4$):\\All $\sim \!10^{7}$ results.}
		\label{exp:4star_syn_small}
    \end{subfigure}%
    \hfill
    \begin{subfigure}{0.24\linewidth}
        \centering
        \includegraphics[width=\linewidth]{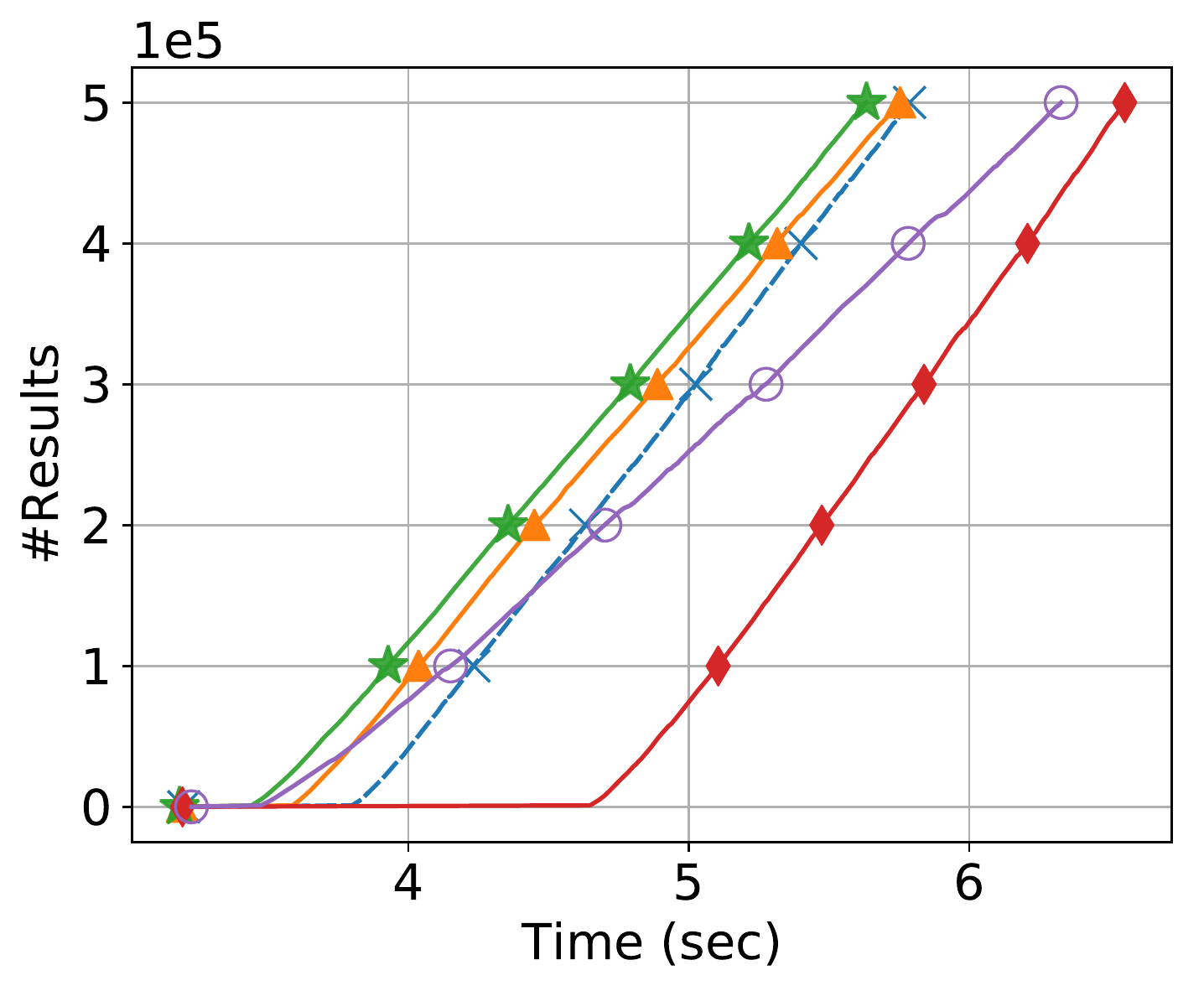}
        \caption{4-Star Synthetic\\($n\!=\!10^6$):\\Top $n/2$ of $\sim 6 \cdot 10^{9}$ results.}
		\label{exp:4star_syn_large}
    \end{subfigure}%
    \hfill
    \begin{subfigure}{0.24\linewidth}
        \centering
        \includegraphics[width=\linewidth]{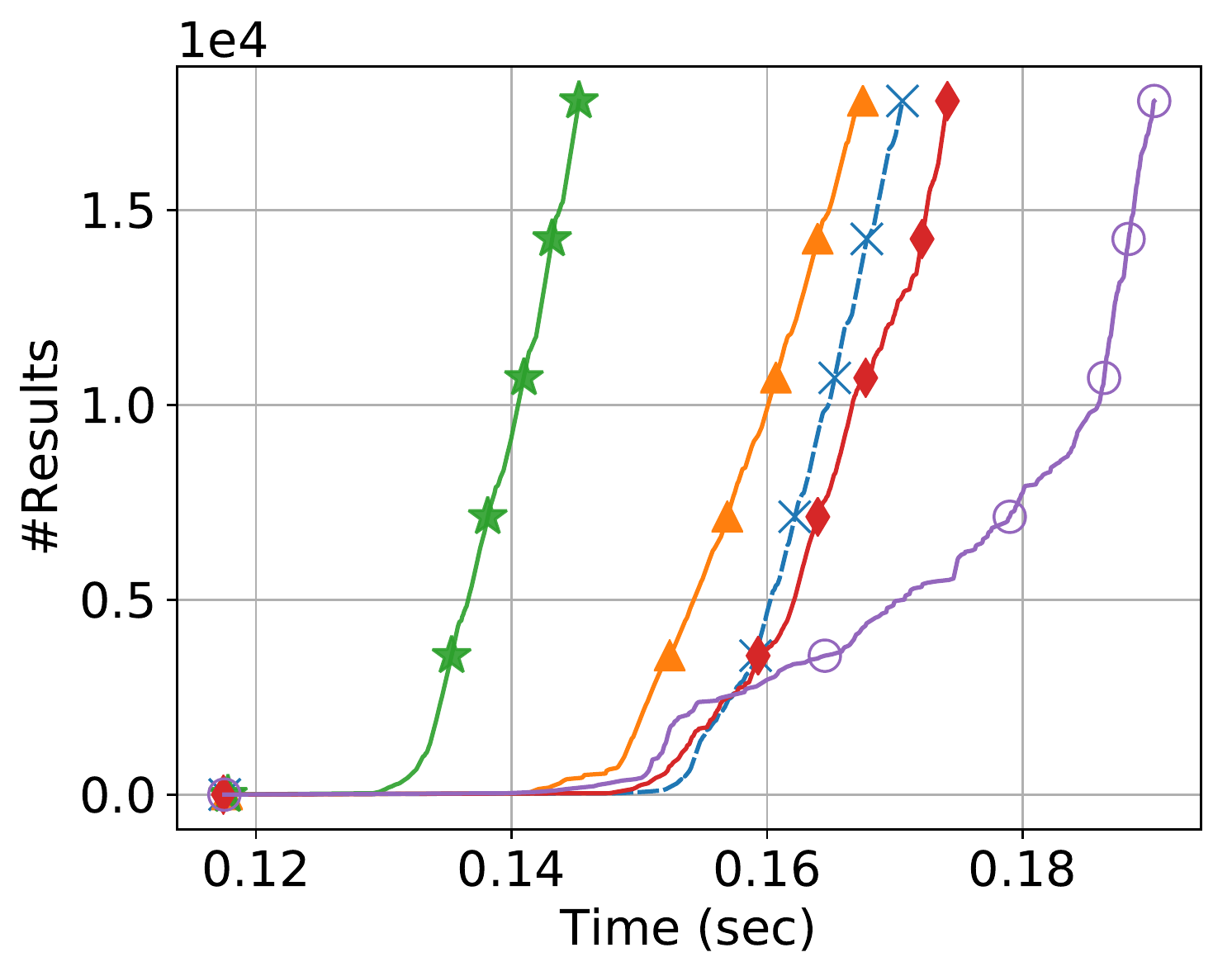}
        \caption{4-Star Bitcoin\\($n\!\sim\!3.6 \!\cdot\! 10^4$):\\Top $n/2$ of $\sim 2 \!\cdot\! 10^{10}$ results.}
		\label{exp:4star_bitcoin}
    \end{subfigure}%
    \hfill
    \begin{subfigure}{0.24\linewidth}
        \centering
        \includegraphics[width=\linewidth]{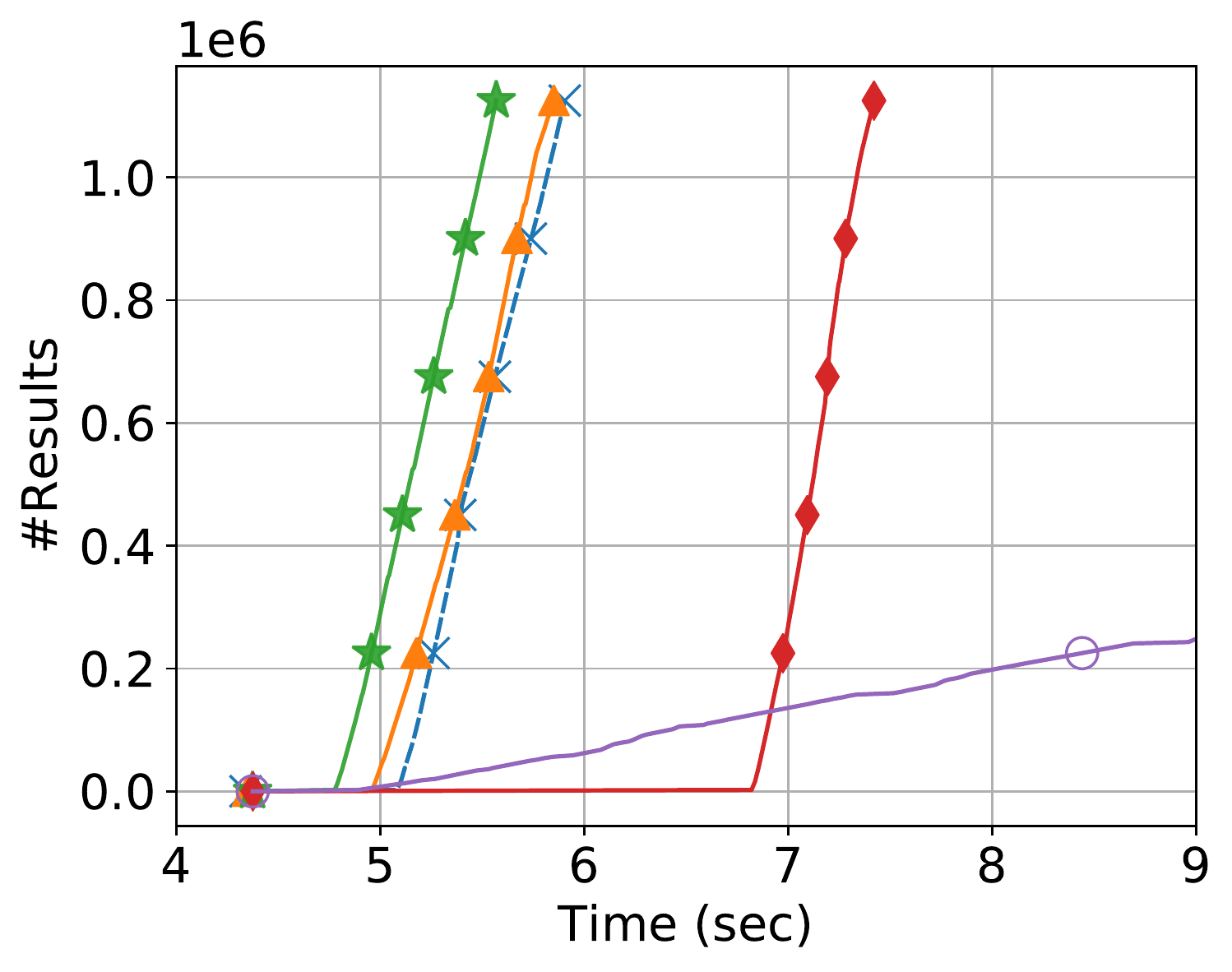}
        \caption{4-Star TwitterL\\($n\!\sim\!2.3 \!\cdot\! 10^6$):\\Top $n/2$ of $\sim 6 \!\cdot\! 10^{15}$ results.}
		\label{exp:4star_twitter}
    \end{subfigure}
    \vspace{-1mm}

    \begin{subfigure}{0.24\linewidth}
        \centering
        \includegraphics[width=\linewidth]{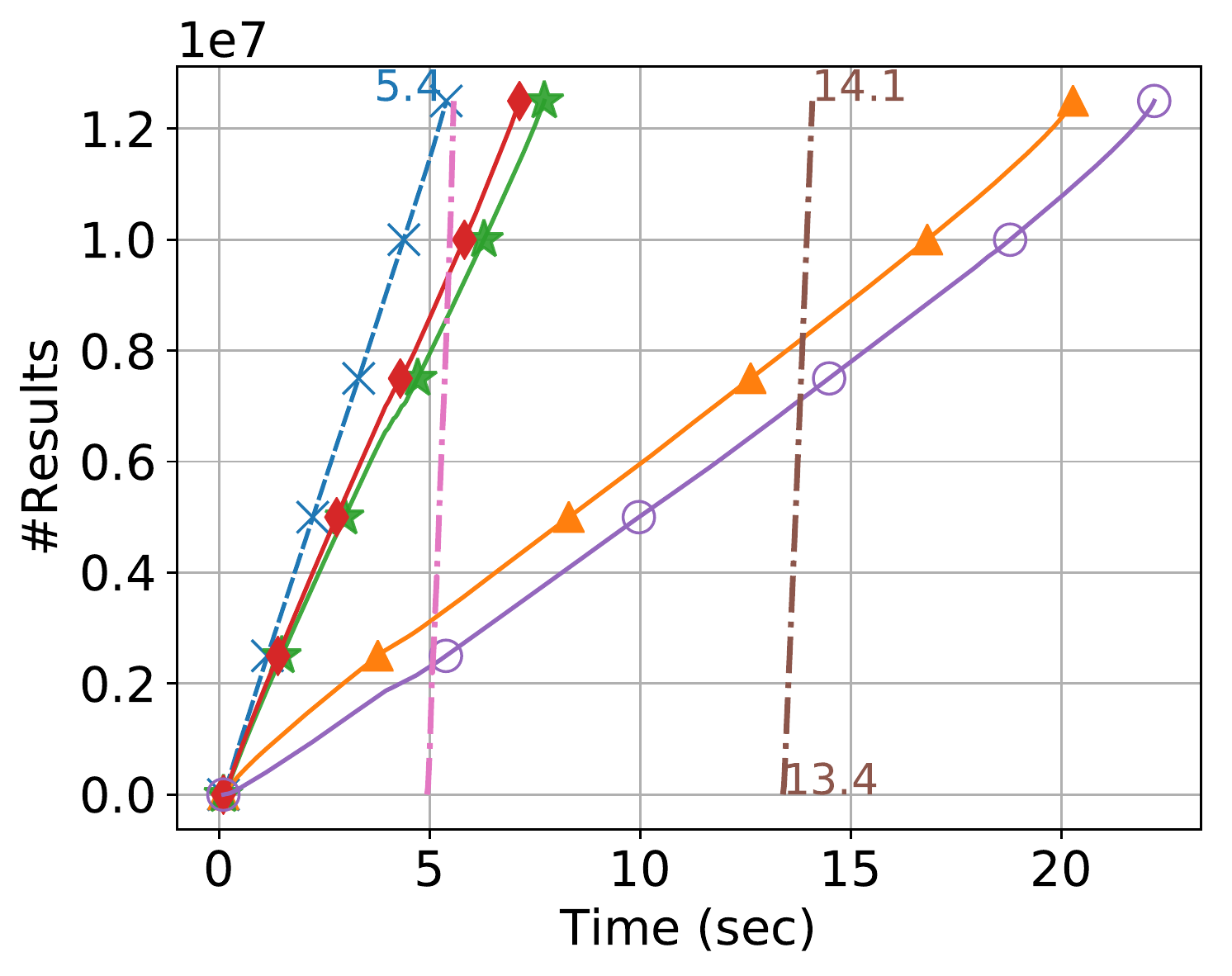}
        \caption{4-Cycle Synthetic\\($n\!=\!5 \!\cdot\! 10^3$):\\All $\sim \!1.2 \!\cdot\! 10^{7}$ results.}
		\label{exp:4cycle_syn_small}
    \end{subfigure}%
    \hfill
    \begin{subfigure}{0.24\linewidth}
        \centering
        \includegraphics[width=\linewidth]{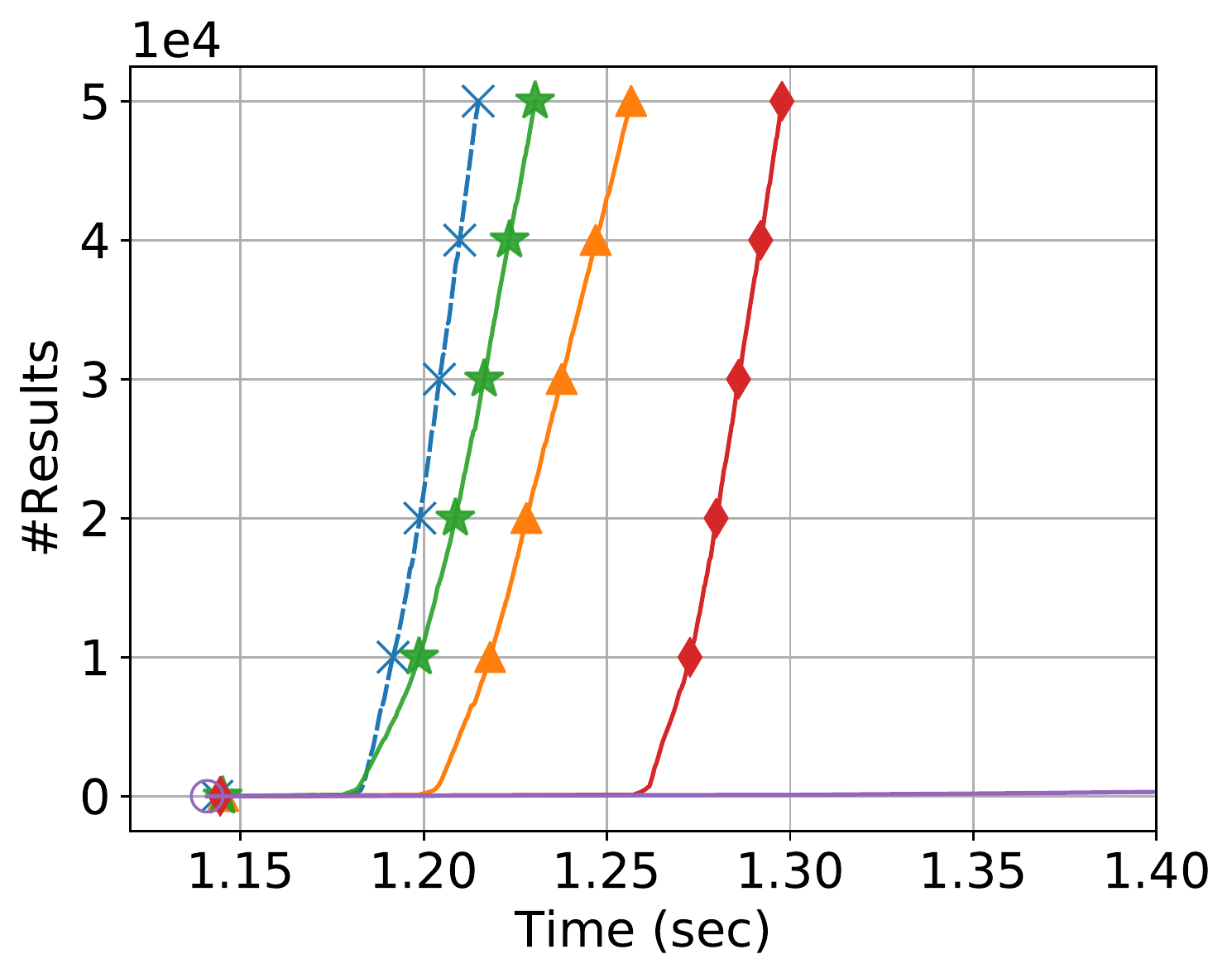}
        \caption{4-Cycle Synthetic\\($n\!=\!10^5$):\\Top $n/2$ of $\sim \!5 \!\cdot\! 10^{9}$ results.}
		\label{exp:4cycle_syn_large}
    \end{subfigure}%
    \hfill
    \begin{subfigure}{0.24\linewidth}
        \centering
        \includegraphics[width=\linewidth]{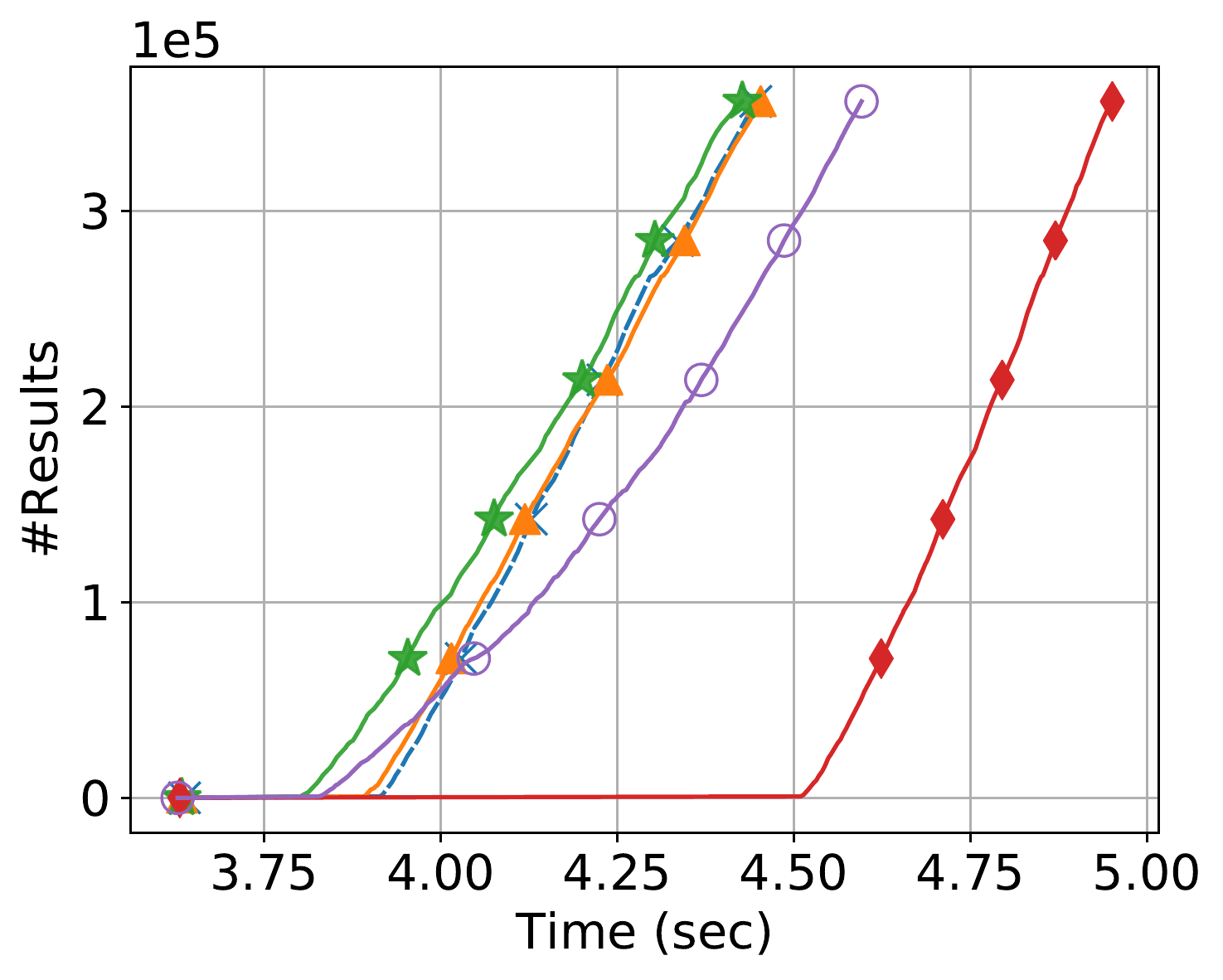}
        \caption{4-Cycle Bitcoin\\($n\!\sim\!3.6 \!\cdot\! 10^4$):\\Top $10n$ of $\sim 7 \cdot 10^{6}$ results.}
		\label{exp:4cycle_bitcoin}
    \end{subfigure}%
    \hfill
    \begin{subfigure}{0.24\linewidth}
        \centering
        \includegraphics[width=\linewidth]{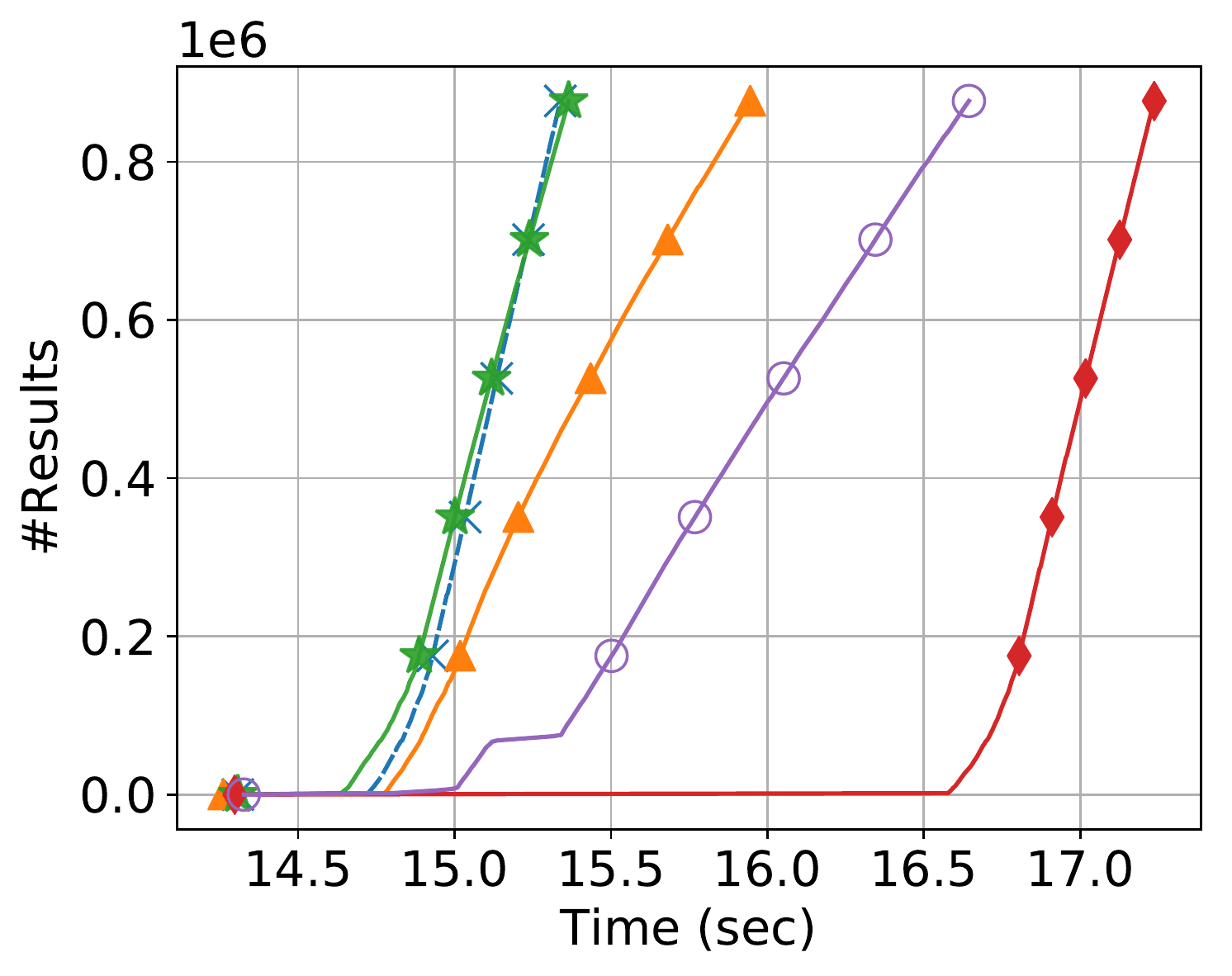}
        \caption{4-Cycle TwitterS\\($n\!\sim\!8.8 \!\cdot\! 10^4$):\\Top $10n$ of $\sim 3 \cdot 10^{8}$ results.}
		\label{exp:4cycle_twitter}
    \end{subfigure}

    \caption{Experiments on queries of size $4$ (\cref{sec:evaluation_4}).
	}
    \label{exp}
\end{figure*}

\introparagraph{Implementation details}
All algorithms are implemented in Java and run on an
Intel Xeon E5-2643 CPU with 3.3Ghz and 128 GB RAM with Ubuntu Linux.
Each data point 
is the median of 200 runs. 
We initialize all data structures lazily when they are accessed for
the first time. For example, in \EAGER, we do not sort the $\Choices$ set of a
node until it is visited.
This can significantly reduce $\TT(k)$ for small $k$, 
and we apply this optimization to all algorithms. 
Notice that our complexity analysis in \cref{sec:complexity}
assumes constant-time inserts for priority queues, which is important for algorithms
that push more elements than they pop per iteration.
This bound is achieved by data structures that are well-known to perform poorly
in practice~\cite{cherkassky96shortest,LarkinSenTarjan2004:PQs}. 
To address this issue in the experiments, we use ``bulk inserts'' which heapify
the inserted elements~\cite{chang15enumeration}
or standard binary heaps when query size is small.

\subsection{Experimental results}
\label{sec:evaluation_4}

\Cref{exp} reports the number of output tuples returned in ranking order over time
for queries of size $4$.
On the larger input, \NAIVE runs out of memory or we terminate it after 2 hours.
This clearly demonstrates the need for our approach. 
We then set a limit on the number of returned results
and compare our various any-k algorithms for relatively small $k$.
We also use a fairly small synthetic input to be able to compare TTL performances against \NAIVE.

\resultbox{%
\introparagraph{Results}
For $\TTL$, \RECURSIVE is fastest on paths and cycles, finishing even before \NAIVE.
This advantage disappears in star queries due to the small depth of the tree.
For small $k$, \LAZY is consistently the top-performer and
is even faster than the asymptotically best \HEAP.
\NAIVE is impractical for real-world data since it attempts to compute the full result, which is extremely large.}

For path and cycle queries on the small synthetic data, \RECURSIVE is faster than \NAIVE
(\cref{exp:4path_syn_small,exp:4cycle_syn_small}) due to the large number of suffixes shared between
different output tuples. It returns the \emph{full sorted result faster (7.7 sec and 5.4 sec) than
\NAIVE (8.3 sec and 14.1 sec)}.
Especially for cycles, our decomposition method really pays off compared to \NAIVE~\cite{ngo2018worst},
as \RECURSIVE terminates around the same time \NAIVE starts to sort.
For star queries, \RECURSIVE behaves like an \RPDP approach 
because of the shallowness of the tree (\cref{exp:4star_syn_small}).
When many results are returned,
the strict \RPDP variants
(\EAGER, \LAZY) have an advantage over the relaxed ones (\HEAP, \MIN) 
as they produce fewer candidates per iteration and maintain a smaller priority queue.
\EAGER is slightly better than \LAZY because sorting is faster than
incrementally converting a heap to a sorted list.
This situation is reversed for small $k$ where
\emph{initialization time} becomes a crucial factor:
Then \EAGER and \RECURSIVE lose their edge, while \LAZY shines
(\cref{exp:4path_bitcoin,exp:4star_bitcoin,exp:4star_twitter,exp:4cycle_bitcoin,exp:4cycle_twitter}).
\RECURSIVE starts off slower,
but often overtakes the others for sufficiently large $k$
(\cref{exp:4path_syn_large,exp:4cycle_syn_large}).
\EAGER is also slow in the beginning
because it has to sort each time it accesses a new choice set.
\HEAP showed mixed results, performing near the top (\cref{exp:4star_syn_large})
or near the bottom (\cref{exp:4cycle_twitter}).
\MIN performs poorly overall due to the large number of successors it inserts into its priority
queue.

\begin{figure*}[t]
    \centering
    \begin{subfigure}{\linewidth}
        \centering
        \includegraphics[width=0.7\linewidth]{figs/experiments/legend.pdf}
    \end{subfigure}
    \vspace{-3mm}
    
    \begin{subfigure}{0.24\linewidth}
        \centering
        \includegraphics[width=\linewidth]{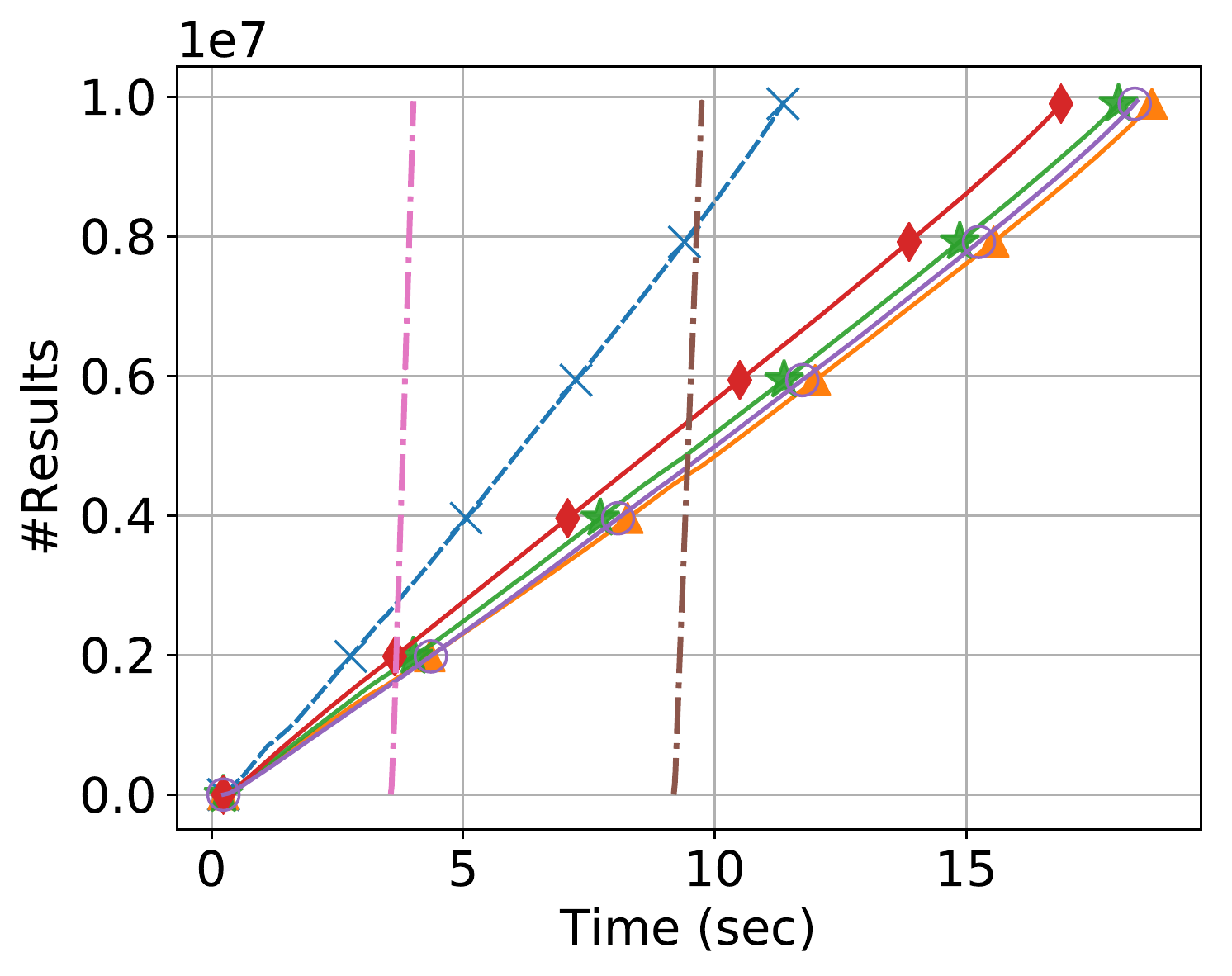}
        \caption{3-Path Synthetic\\($n\!=\!10^5$):\\All $\sim \!10^{7}$ results.}
		\label{exp:3path_syn_small}
    \end{subfigure}%
    \hfill
    \begin{subfigure}{0.24\linewidth}
        \centering
        \includegraphics[width=\linewidth]{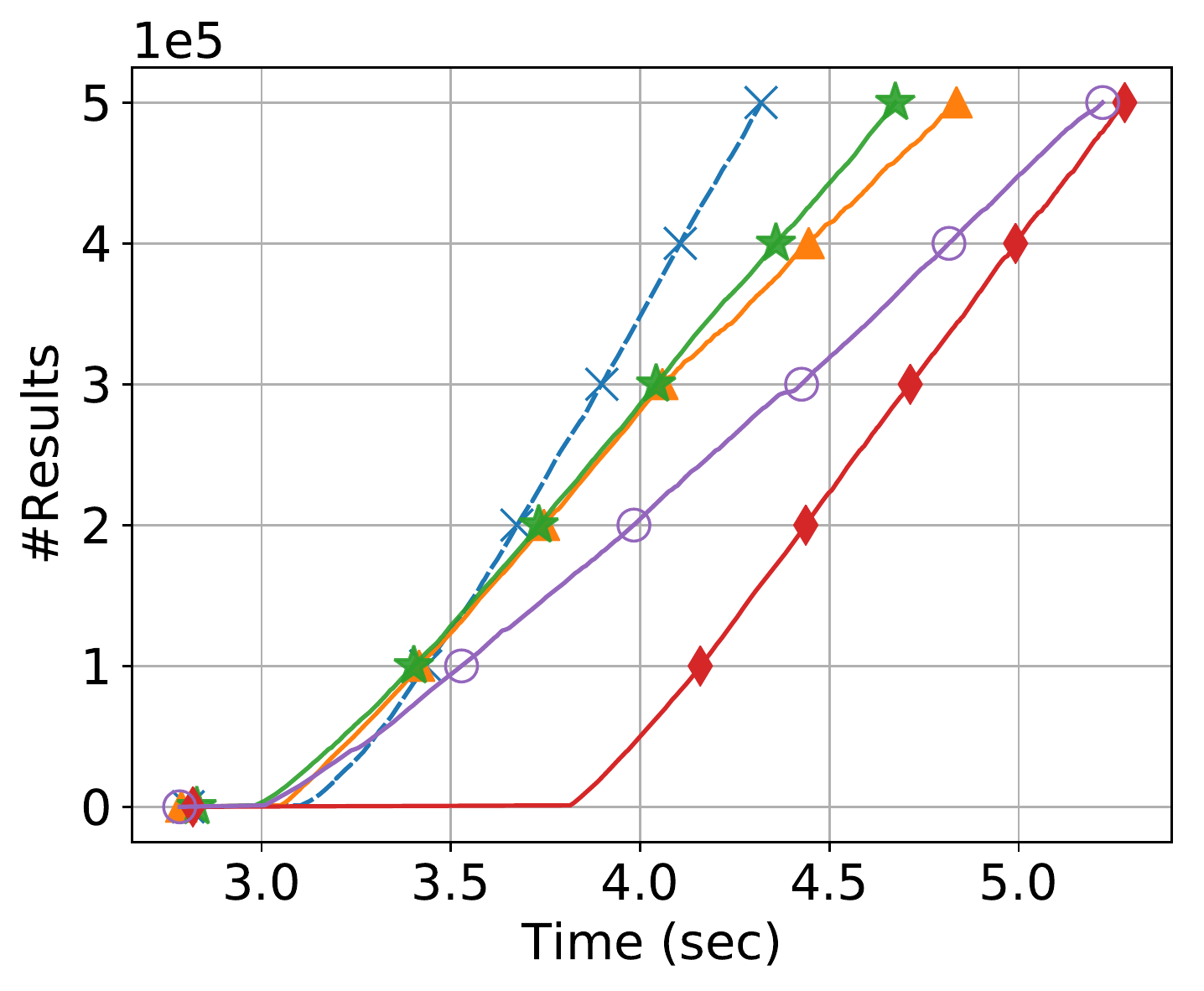}
        \caption{3-Path Synthetic\\($n\!=\!10^6$):\\Top $n/2$ of $\sim 10^8$ results.}
		\label{exp:3path_syn_large}
    \end{subfigure}%
    \hfill
    \begin{subfigure}{0.24\linewidth}
        \centering
        \includegraphics[width=\linewidth]{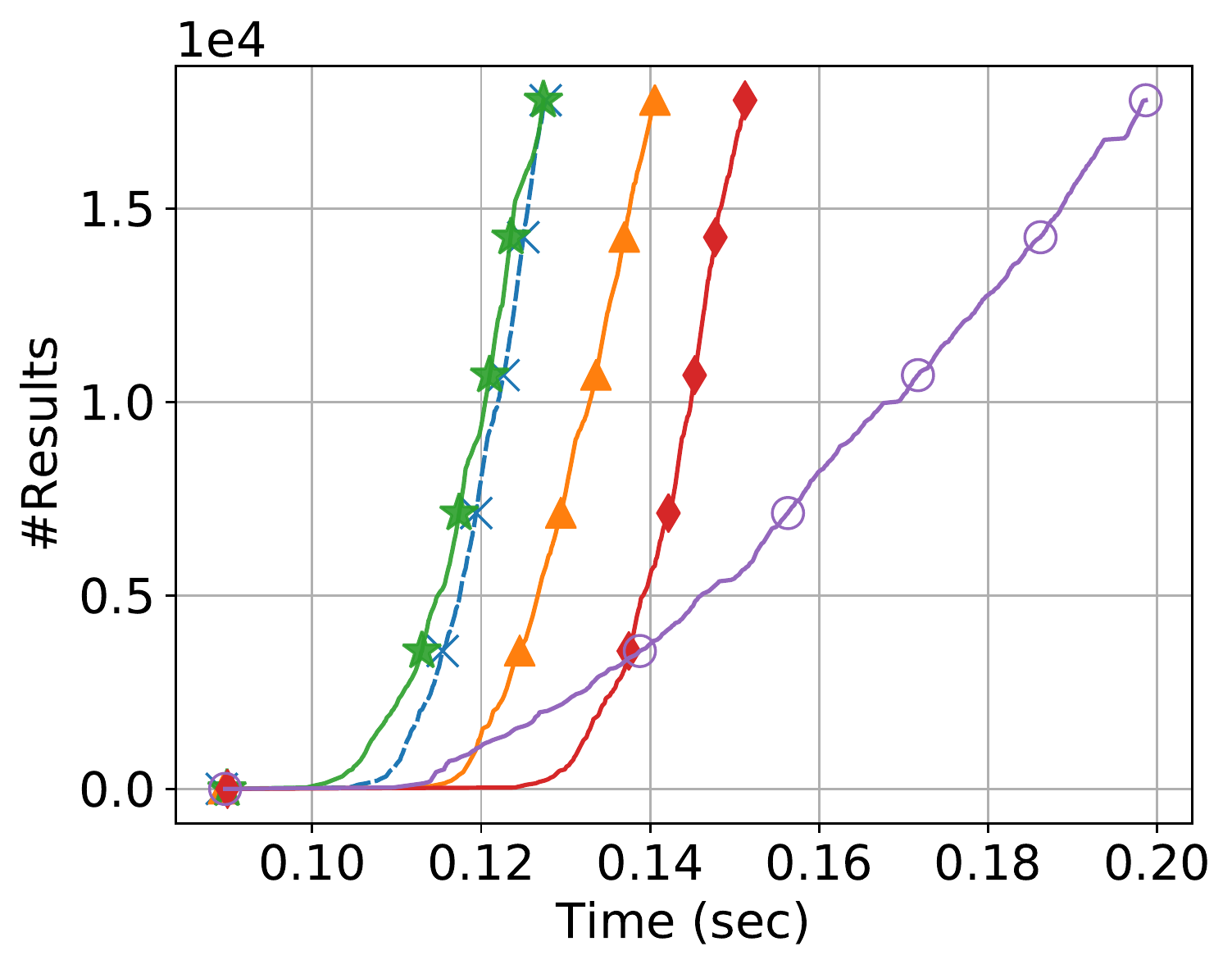}
        \caption{3-Path Bitcoin\\($n\!\sim\!3.6 \!\cdot\! 10^4$):\\Top $n/2$ of $\sim 8 \cdot 10^7$ results.}
		\label{exp:3path_bitcoin}
    \end{subfigure}%
    \hfill
    \begin{subfigure}{0.24\linewidth}
        \centering
        \includegraphics[width=\linewidth]{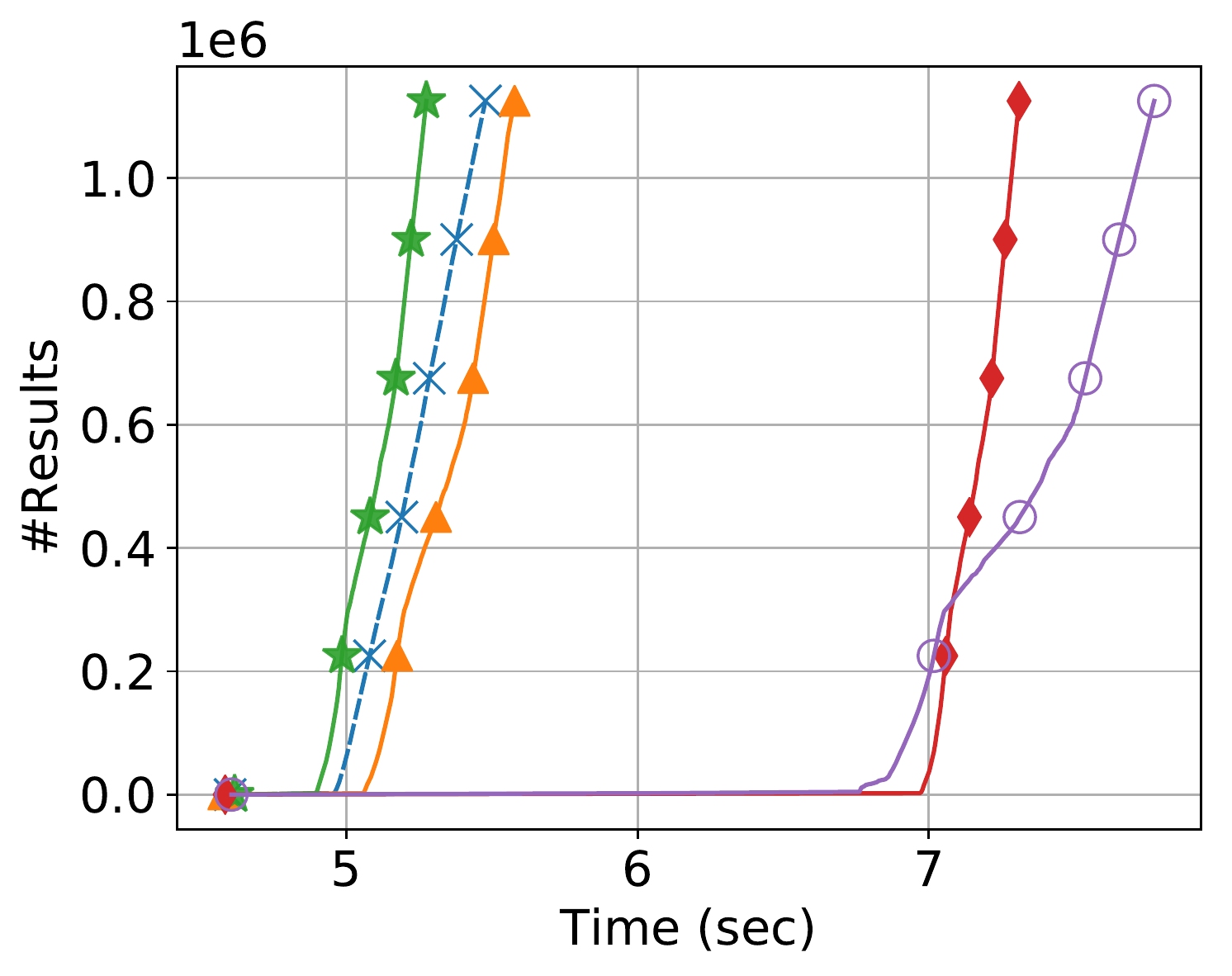}
        \caption{3-Path TwitterL\\($n\!\sim\!2.3 \!\cdot\! 10^6$):\\Top $n/2$ of $\sim 9 \!\cdot\! 10^{11}$ results.}
		\label{exp:3path_twitter}
    \end{subfigure}
    \vspace{-1mm}
    
        \begin{subfigure}{0.24\linewidth}
        \centering
        \includegraphics[width=\linewidth]{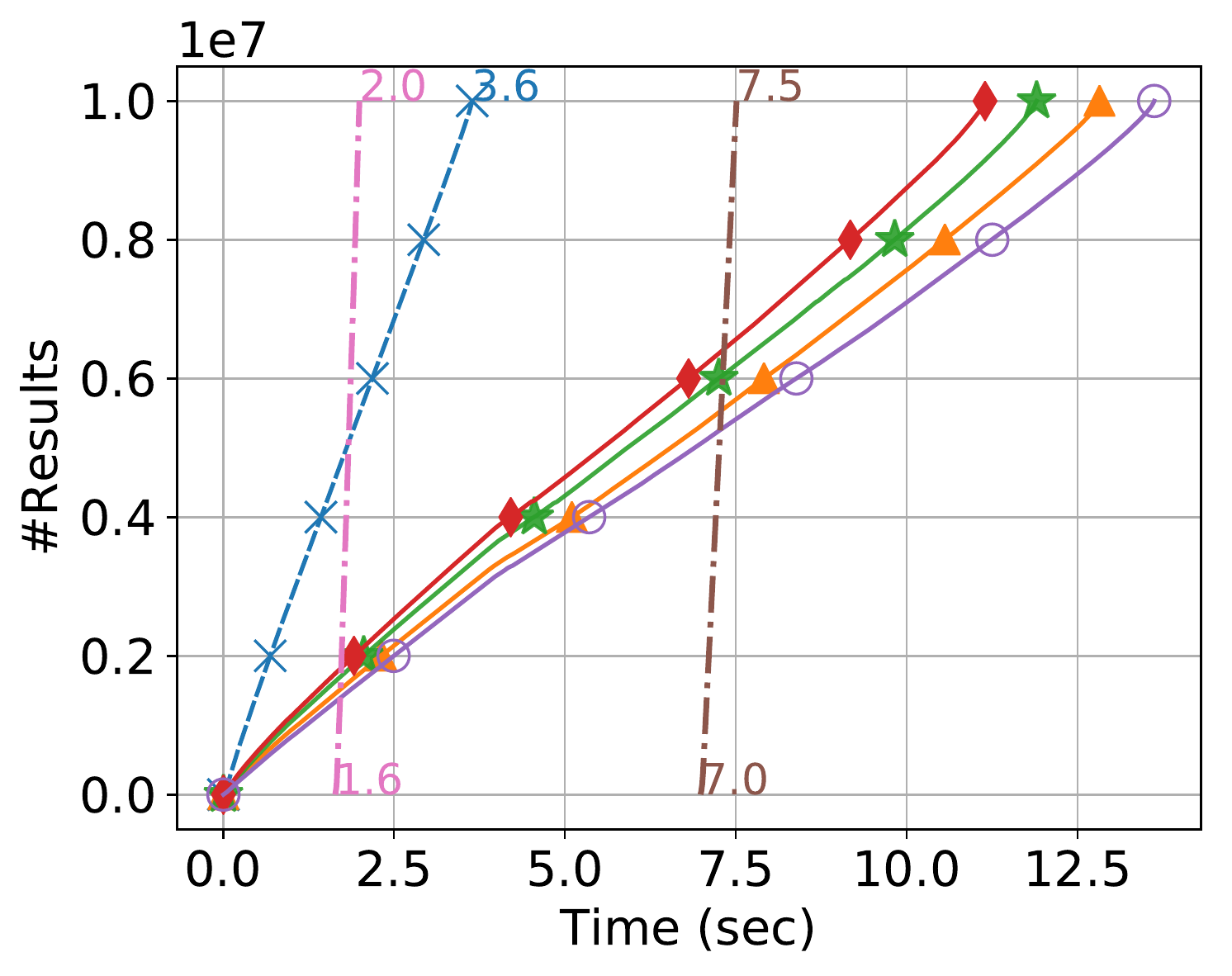}
        \caption{6-Path Synthetic\\($n\!=\!10^2$):\\All $\sim \!10^{7}$ results.}
		\label{exp:6path_syn_small}
    \end{subfigure}%
    \hfill
    \begin{subfigure}{0.24\linewidth}
        \centering
        \includegraphics[width=\linewidth]{figs/experiments/path_n1000000_l4_d100000.pdf}
        \caption{6-Path Synthetic\\($n\!=\!10^6$):\\Top $n/2$ of $\sim 10^{11}$ results.}
		\label{exp:6path_syn_large}
    \end{subfigure}%
    \hfill
    \begin{subfigure}{0.24\linewidth}
        \centering
        \includegraphics[width=\linewidth]{figs/experiments/path_bitcoinotc_l4.pdf}
        \caption{6-Path Bitcoin\\($n\!\sim\!3.6 \!\cdot\! 10^4$):\\Top $n/2$ of $\sim 10^{12}$ results.}
		\label{exp:6path_bitcoin}
    \end{subfigure}%
    \hfill
    \begin{subfigure}{0.24\linewidth}
        \centering
        \includegraphics[width=\linewidth]{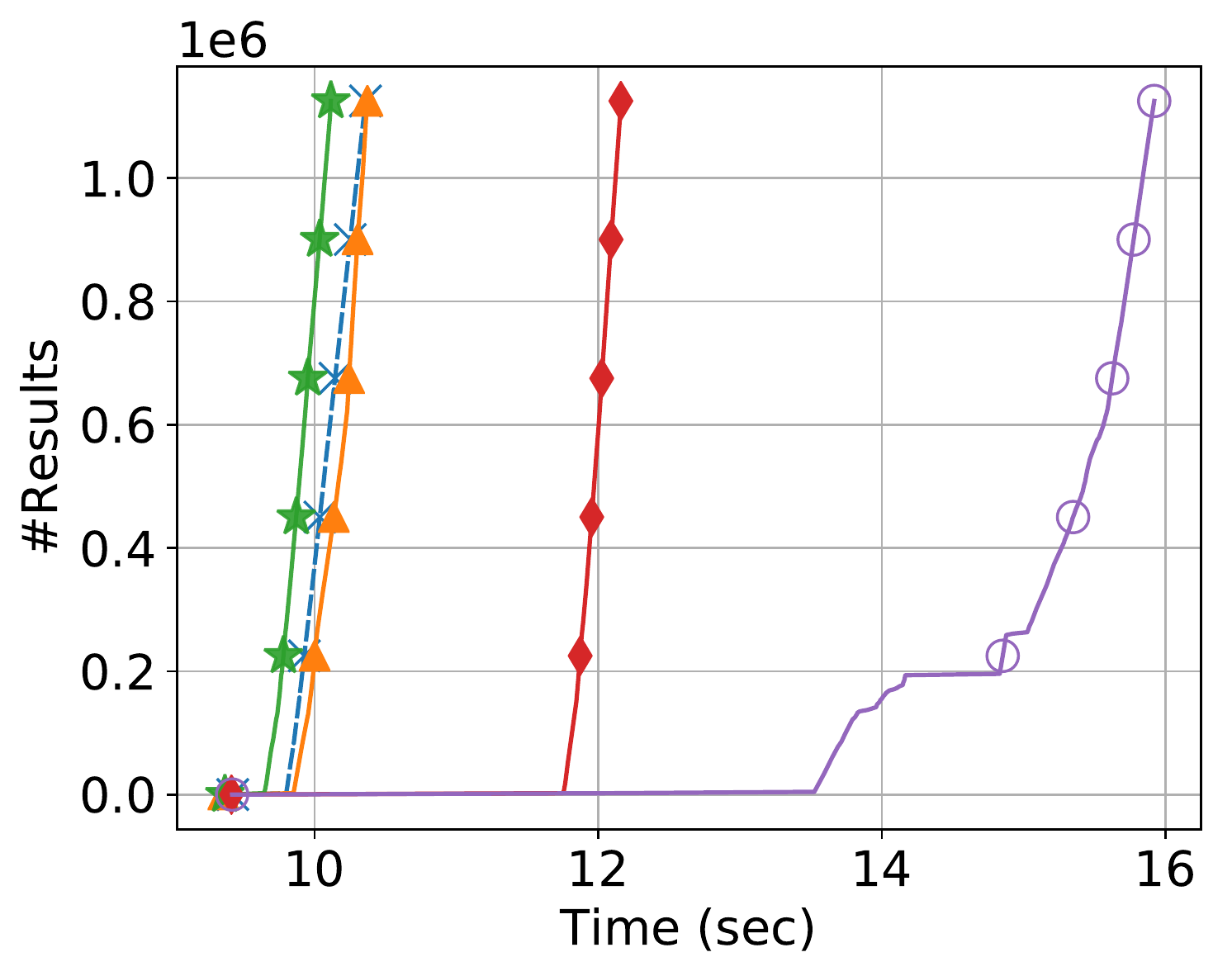}
        \caption{6-Path TwitterL\\($n\!\sim\!2.3 \!\cdot\! 10^6$):\\Top $n/2$ of $\sim 4 \!\cdot\! 10^{19}$ results.}
		\label{exp:6path_twitter}
    \end{subfigure}
    \vspace{-1mm}

    \caption{Experiments on path queries of sizes $3$ and $6$.}
    \label{exp_extra_paths}
\end{figure*}

\begin{figure*}[t]
    \centering
    \begin{subfigure}{\linewidth}
        \centering
        \includegraphics[width=0.7\linewidth]{figs/experiments/legend.pdf}
    \end{subfigure}
    \vspace{-3mm}
    
    \begin{subfigure}{0.24\linewidth}
        \centering
        \includegraphics[width=\linewidth]{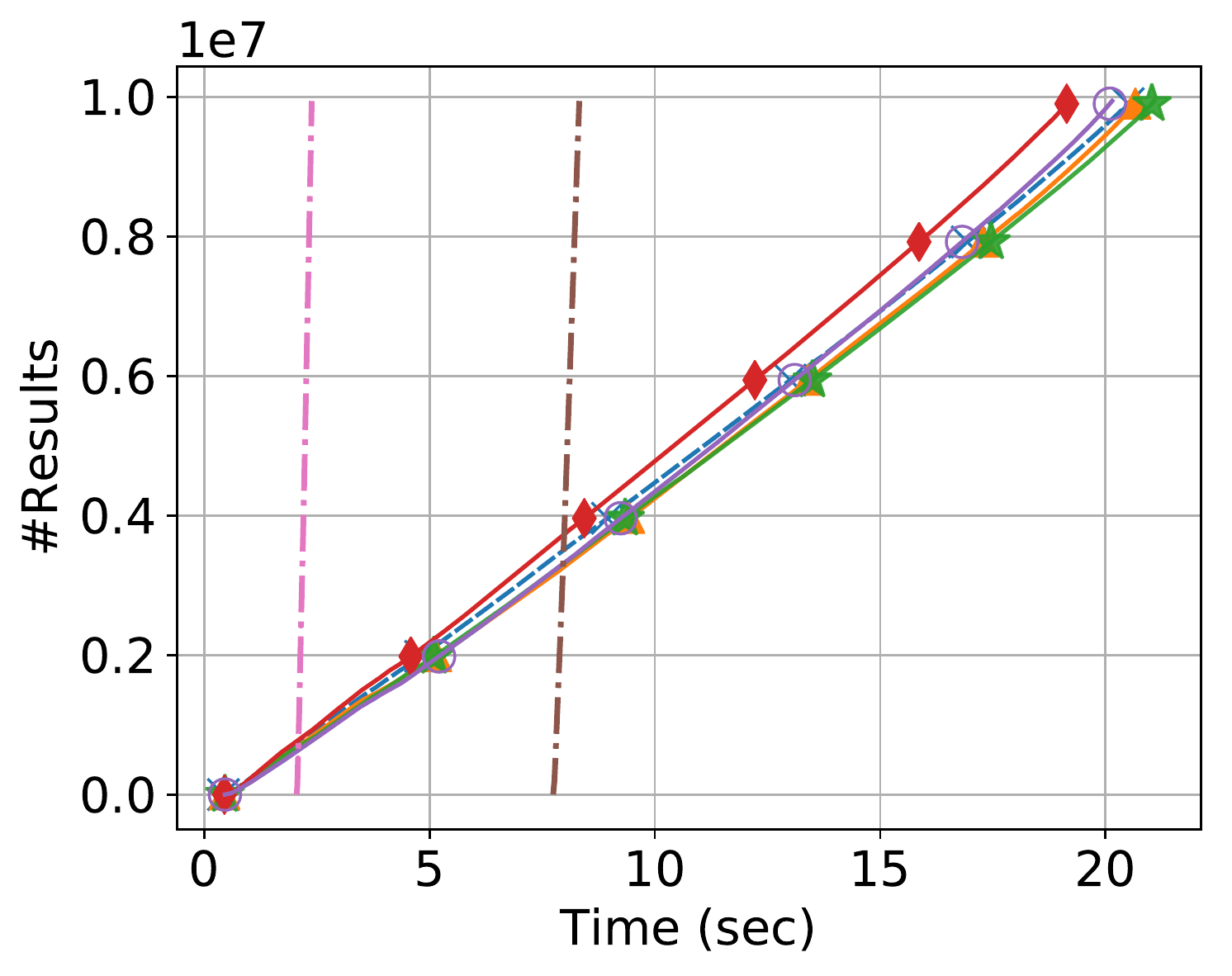}
        \caption{3-Star Synthetic\\($n\!=\!10^5$):\\All $\sim \!10^{7}$ results.
        }
		\label{exp:3star_syn_small}
    \end{subfigure}%
    \hfill
    \begin{subfigure}{0.24\linewidth}
        \centering
        \includegraphics[width=\linewidth]{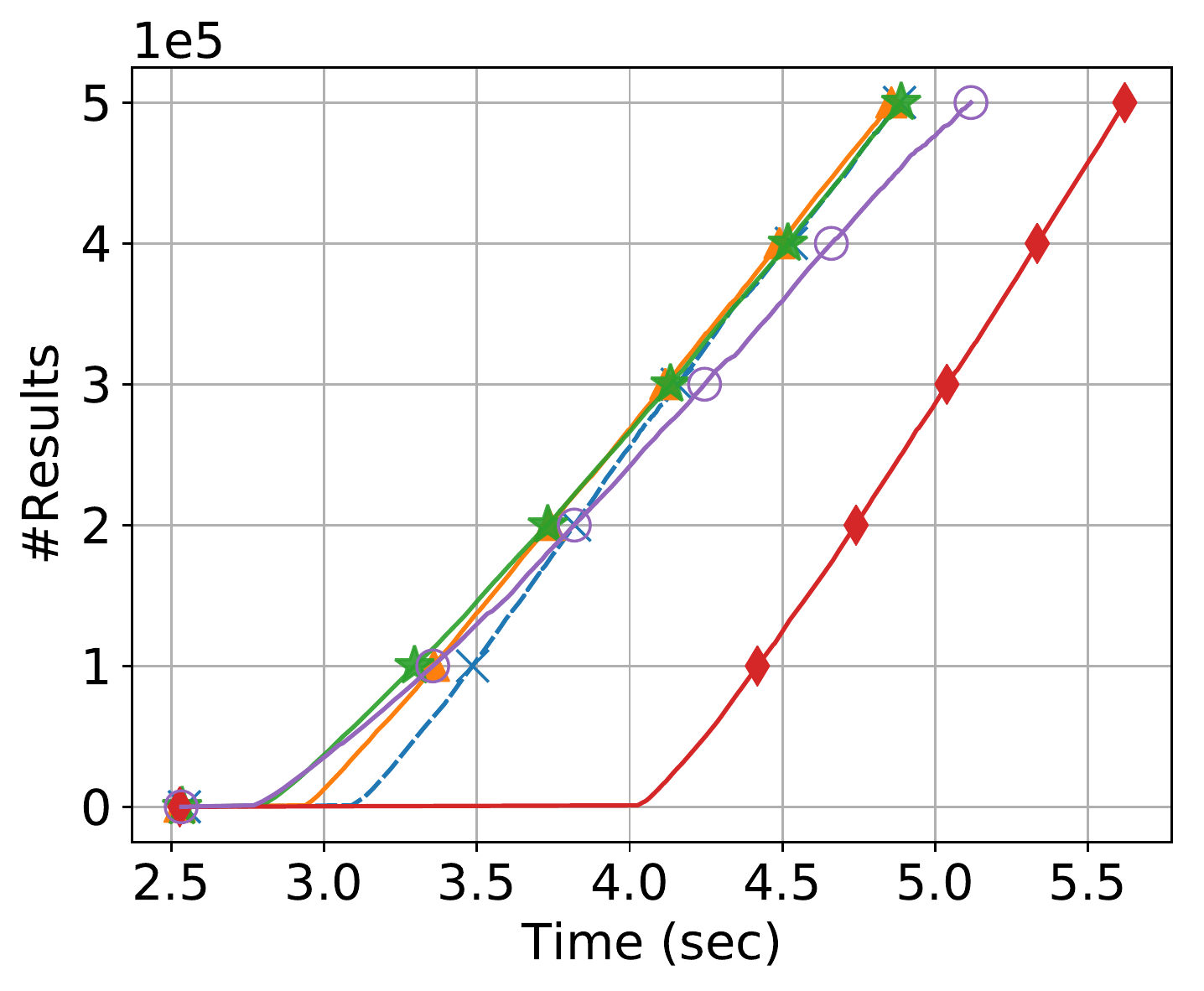}
        \caption{3-Star Synthetic\\($n\!=\!10^6$):\\Top $n/2$ of $\sim 10^{8}$ results.}
		\label{exp:3star_syn_large}
    \end{subfigure}%
    \hfill
    \begin{subfigure}{0.24\linewidth}
        \centering
        \includegraphics[width=\linewidth]{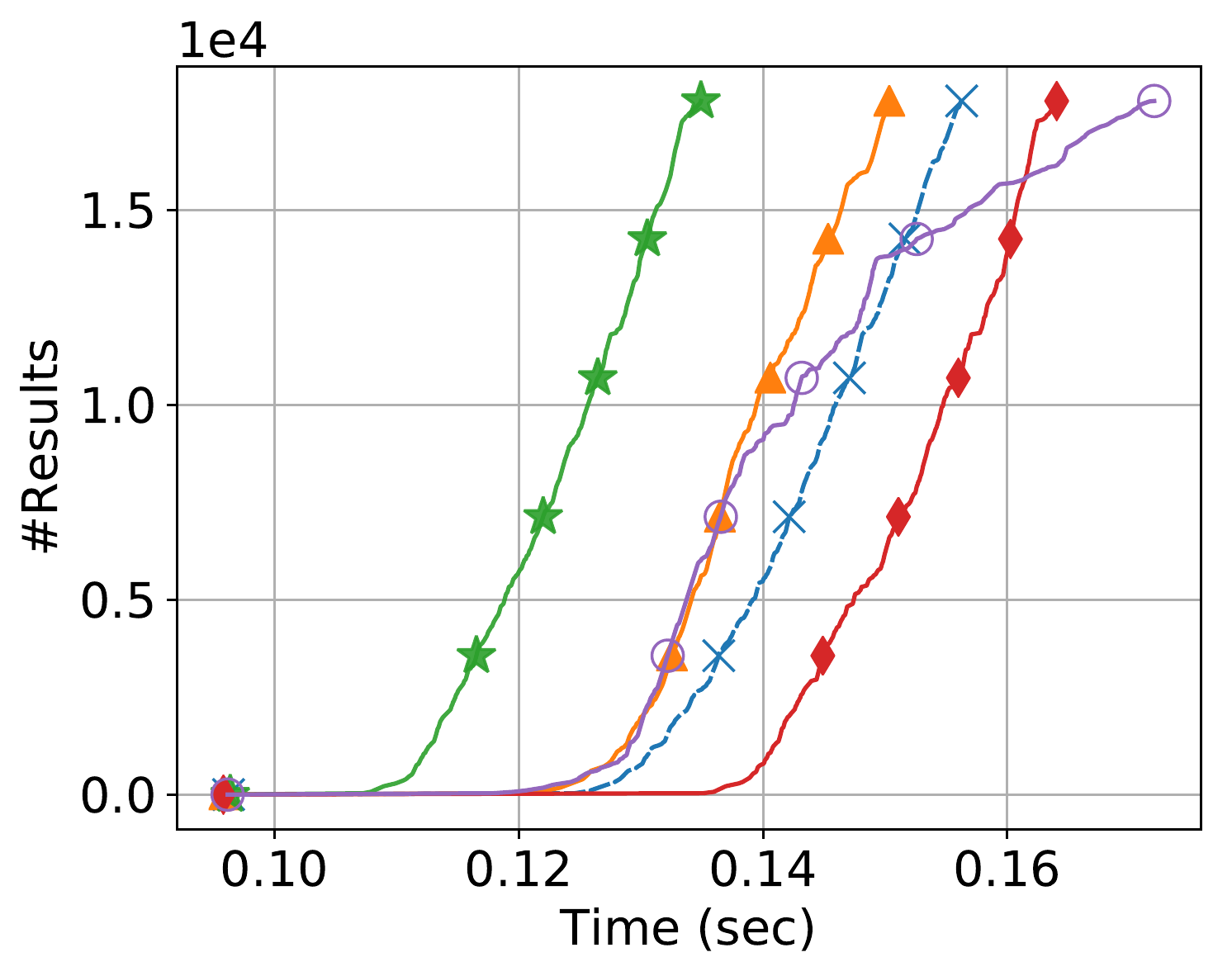}
        \caption{3-Star Bitcoin\\($n\!\sim\!3.6 \!\cdot\! 10^4$):\\Top $n/2$ of $\sim 10^{8}$ results.
        }
		\label{exp:3star_bitcoin}
    \end{subfigure}%
    \hfill
    \begin{subfigure}{0.24\linewidth}
        \centering
        \includegraphics[width=\linewidth]{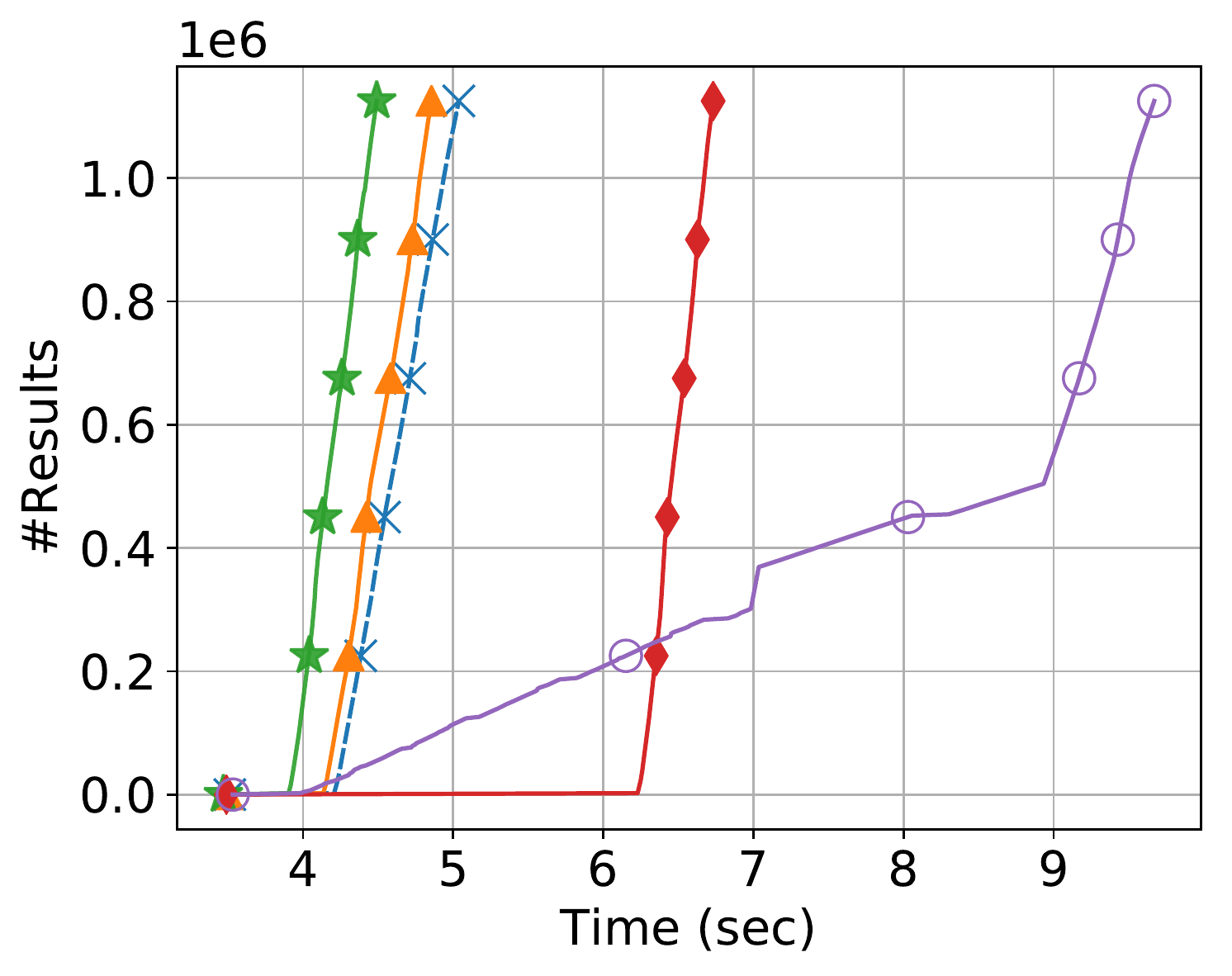}
        \caption{3-Star TwitterL\\($n\!\sim\!2.3 \!\cdot\! 10^6$):\\Top $n/2$ of $\sim 4 \!\cdot\! 10^{12}$ results.
        }
		\label{exp:3star_twitter}
    \end{subfigure}
    \vspace{-1mm}
    
    \begin{subfigure}{0.24\linewidth}
        \centering
        \includegraphics[width=\linewidth]{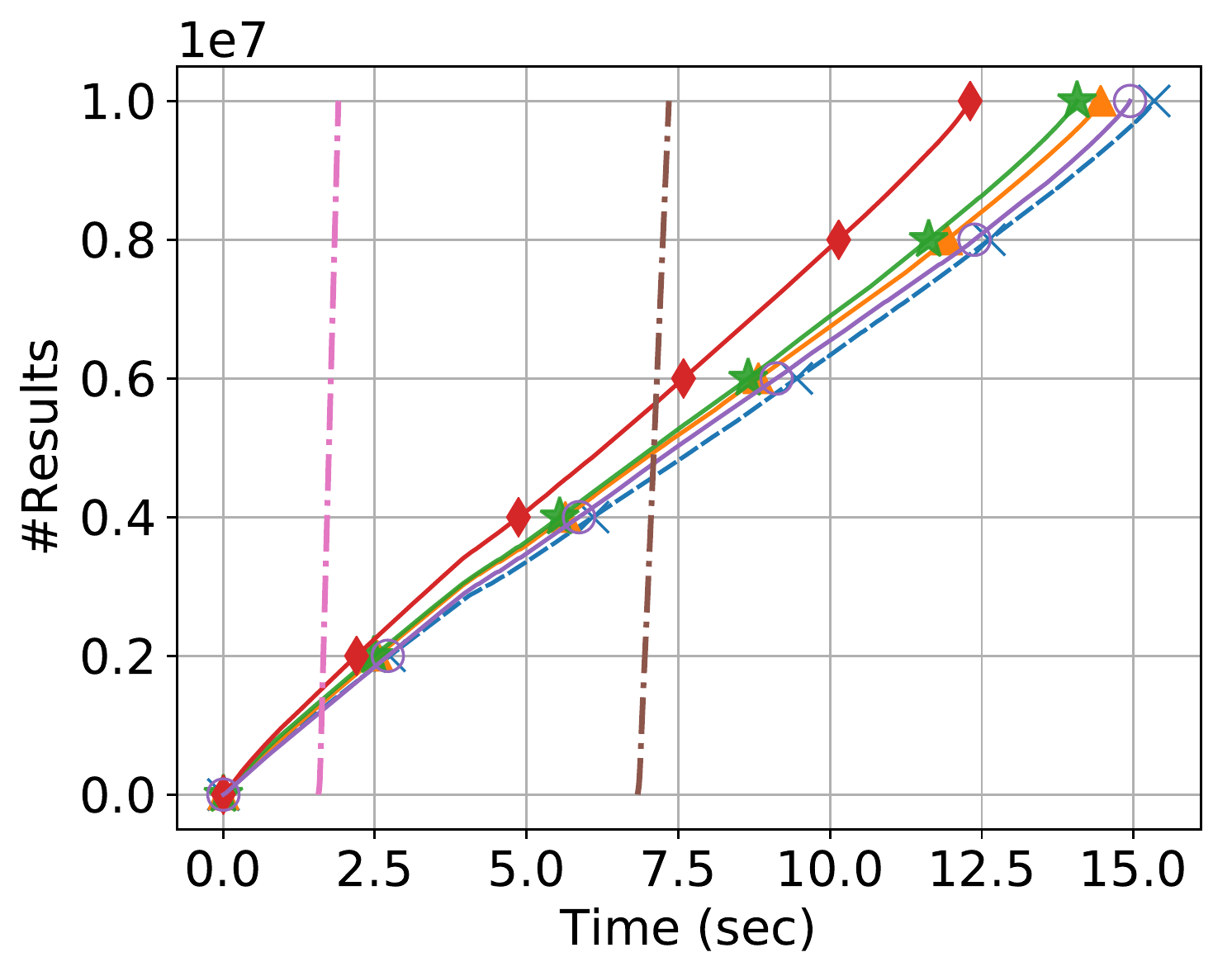}
        \caption{6-Star Synthetic\\($n\!=\!10^2$):\\All $\sim \!10^{7}$ results.}
		\label{exp:6star_syn_small}
    \end{subfigure}%
    \hfill
    \begin{subfigure}{0.24\linewidth}
        \centering
        \includegraphics[width=\linewidth]{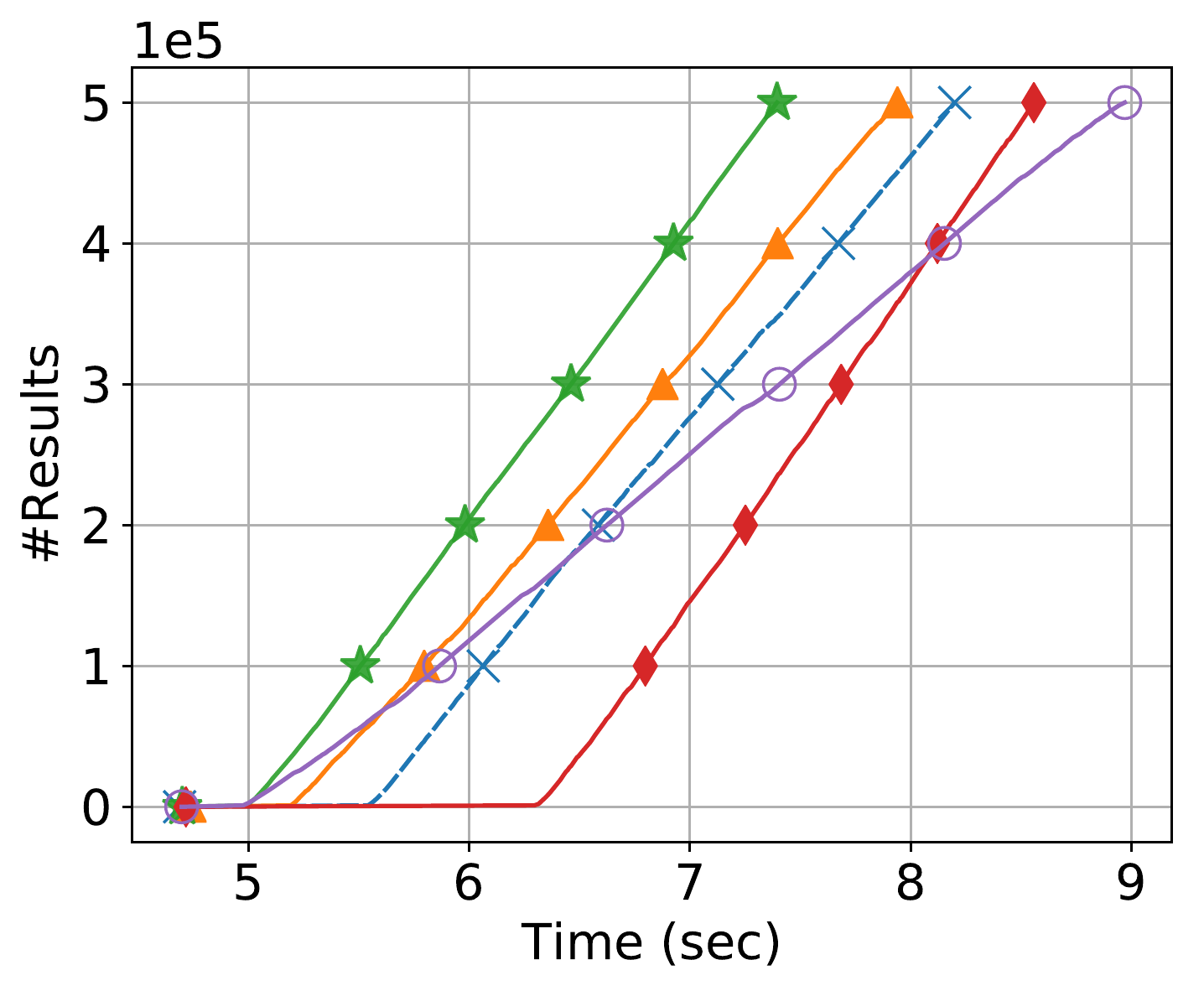}
        \caption{6-Star Synthetic\\($n\!=\!10^6$):\\Top $n/2$ of $\sim 10^{11}$ results.}
		\label{exp:6star_syn_large}
    \end{subfigure}%
    \hfill
    \begin{subfigure}{0.24\linewidth}
        \centering
        \includegraphics[width=\linewidth]{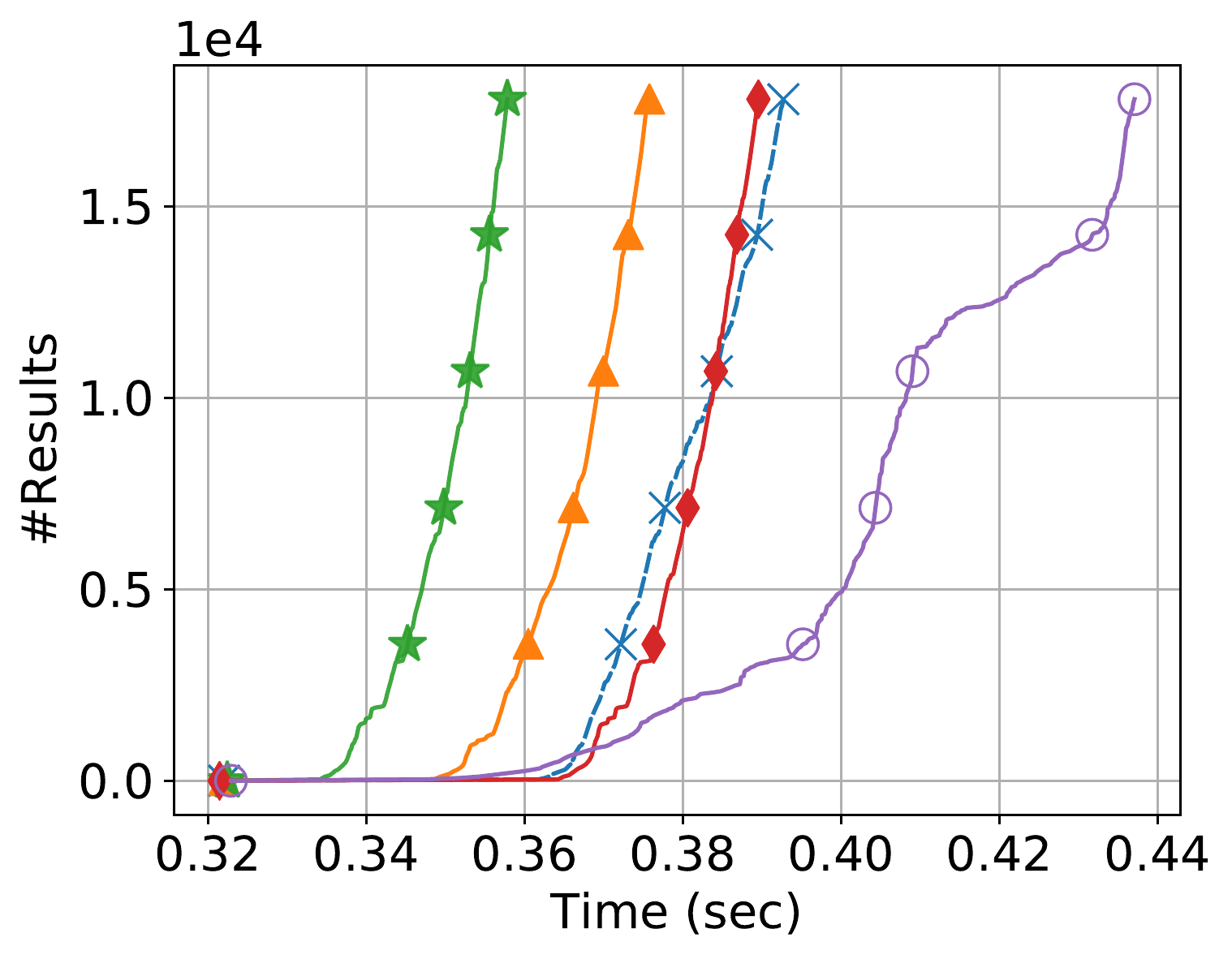}
        \caption{6-Star Bitcoin\\($n\!\sim\!3.6 \!\cdot\! 10^4$):\\Top $n/2$ of $\sim 3 \!\cdot\! 10^{14}$ results.}
		\label{exp:6star_bitcoin}
    \end{subfigure}%
    \hfill
    \begin{subfigure}{0.24\linewidth}
        \centering
        \includegraphics[width=\linewidth]{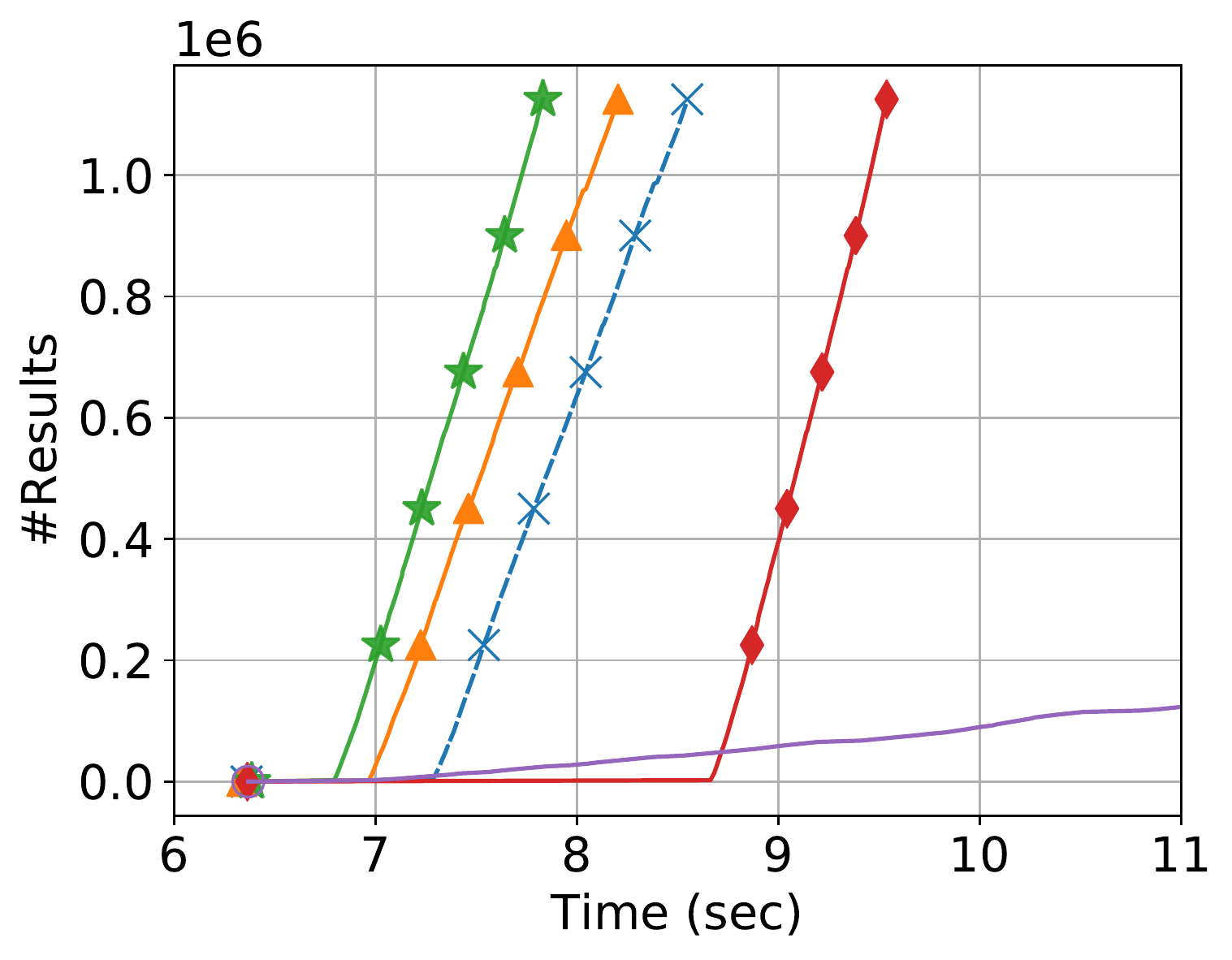}
        \caption{6-Star TwitterL.\\Top $n/2$ of $\sim 10^{22}$ results.}
		\label{exp:6star_twitter}
    \end{subfigure}
    \vspace{-1mm}
    
    \caption{Experiments on star queries of sizes $3$ and $6$.}
    \label{exp_extra_stars}
\end{figure*}

\begin{figure*}[t]
    \centering
    \begin{subfigure}{\linewidth}
        \centering
        \includegraphics[width=0.7\linewidth]{figs/experiments/legend.pdf}
    \end{subfigure}
    \vspace{-3mm}
    
    \begin{subfigure}{0.24\linewidth}
        \centering
        \includegraphics[width=\linewidth]{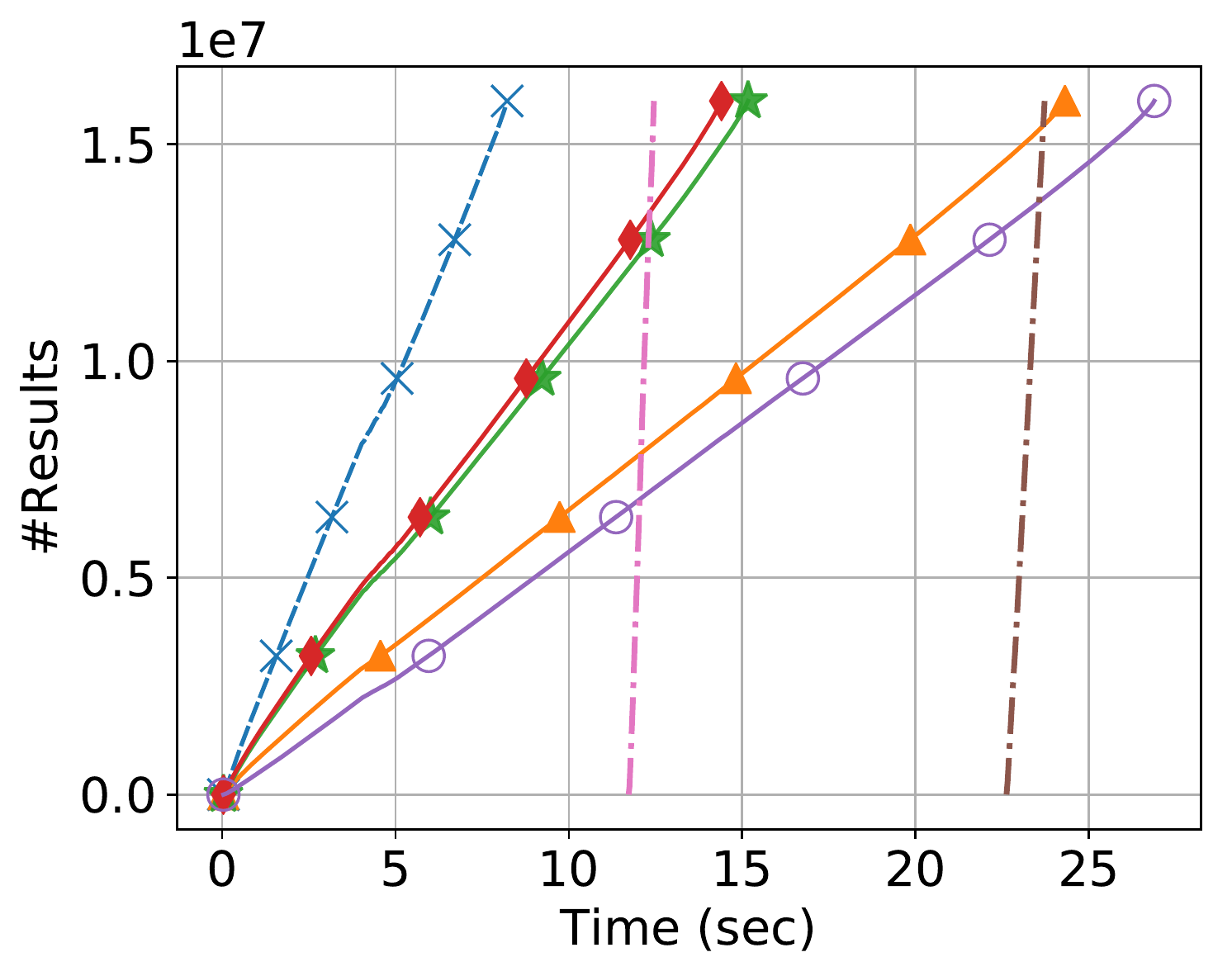}
        \caption{6-Cycle Synthetic\\($n\!=\! 4 \cdot 10^2$):\\All $\sim \!1.6 \!\cdot\! 10^{7}$ results.}
		\label{exp:6cycle_syn_small}
    \end{subfigure}%
    \hfill
    \begin{subfigure}{0.24\linewidth}
        \centering
        \includegraphics[width=\linewidth]{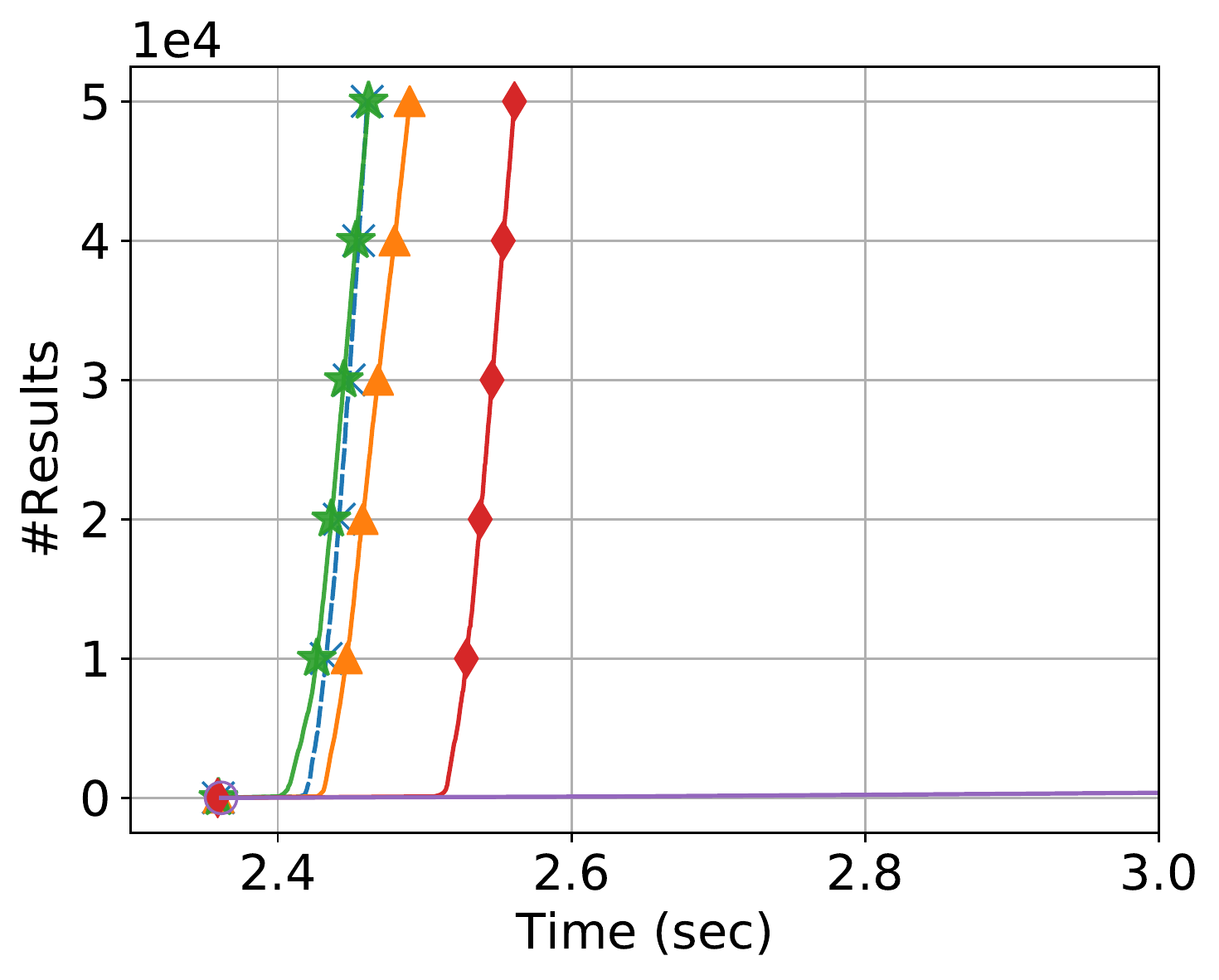}
        \caption{6-Cycle Synthetic\\($n\!=\!10^5$):\\Top $n/2$ of $\sim 3 \!\cdot\! 10^{14}$ results.}
		\label{exp:6cycle_syn_large}
    \end{subfigure}%
    \hfill
    \begin{subfigure}{0.24\linewidth}
        \centering
        \includegraphics[width=\linewidth]{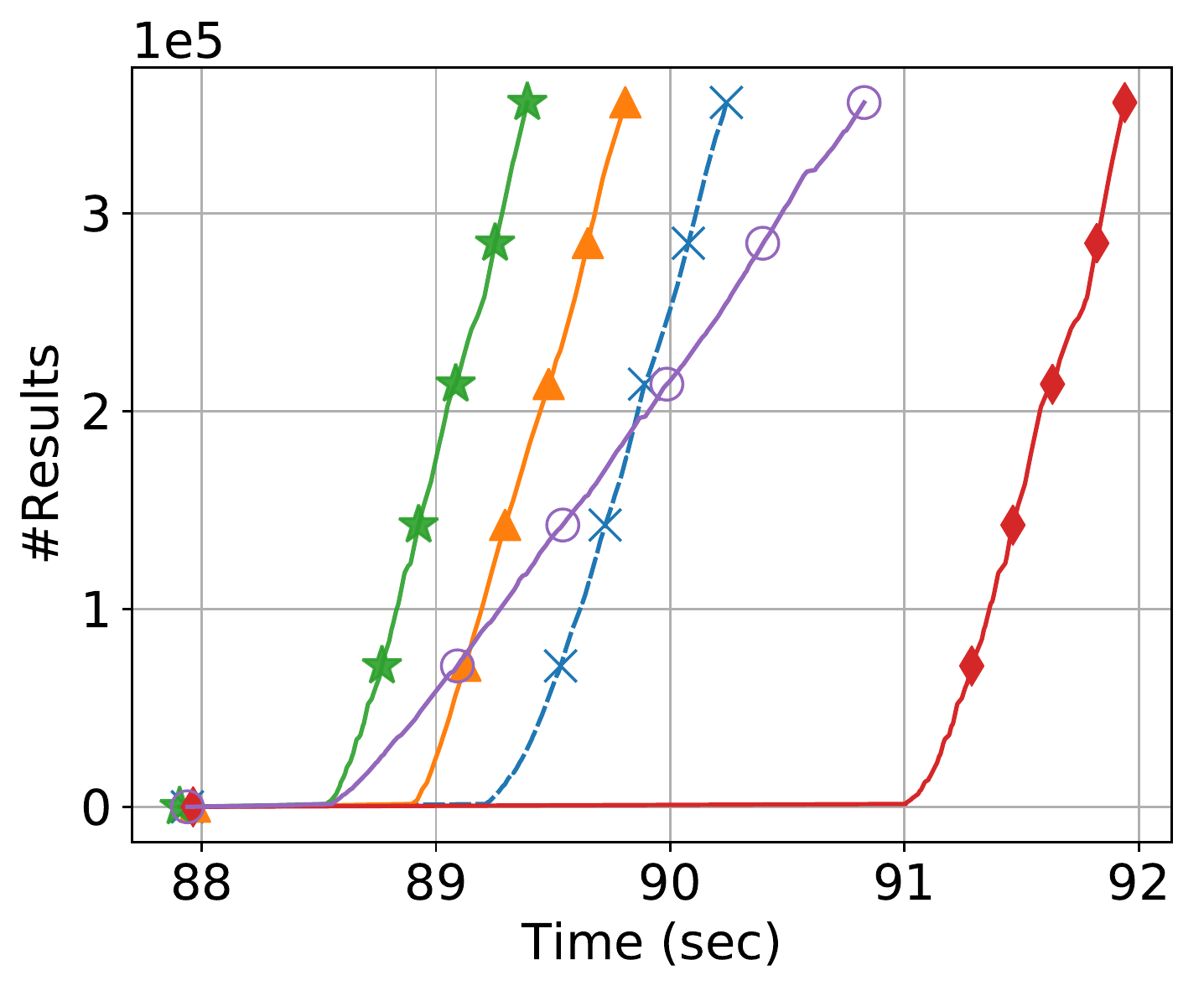}
        \caption{6-Cycle Bitcoin\\($n\!\sim\!3.6 \!\cdot\! 10^4$):\\Top $50n$ of $> 10^{9}$ results.}
		\label{exp:6cycle_bitcoin}
    \end{subfigure}%
    \hfill
    \begin{subfigure}{0.24\linewidth}
        \centering
        \includegraphics[width=\linewidth]{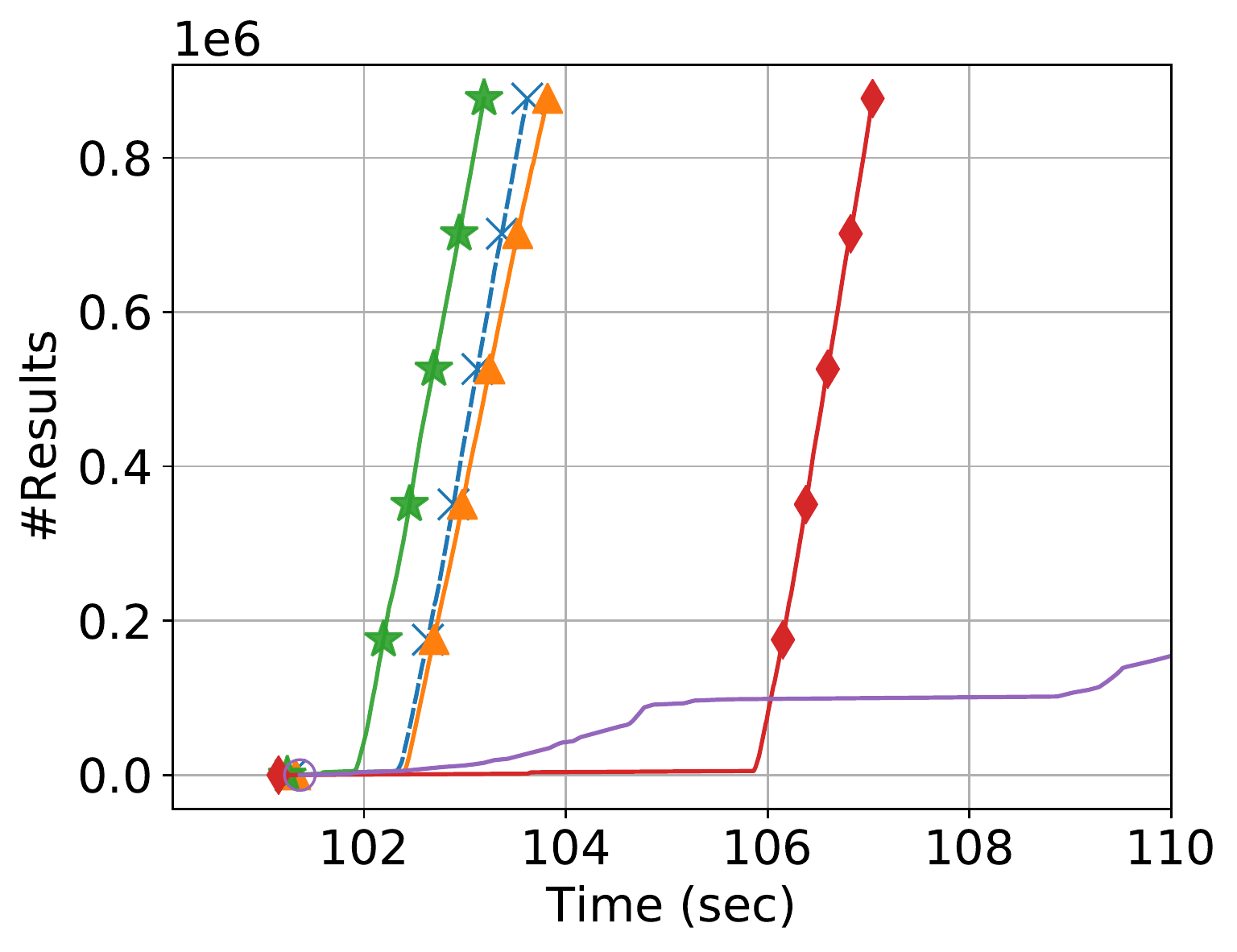}
        \caption{6-Cycle TwitterS\\($n\!\sim\!8.8 \!\cdot\! 10^4$):\\Top $50n$ of $> 10^{9}$ results.}
		\label{exp:6cycle_twitter}
    \end{subfigure}
    \vspace{-1mm}
    
    \caption{Experiments on cycle queries of size $6$.}
    \label{exp_extra_cycles}
\end{figure*}

\subsection{More results for different query sizes}

We performed the same experiments for different query sizes:
3-Path, 6-Path, 3-Star, 6-Star, and 6-Cycle. 
We do not consider the cycle of length 3 (i.e., the triangle query) because our simple cycle decomposition does not give any bound that would be better than the \NAIVE algorithm.
Our goal with these experiments is to observe how the conclusions we made in \Cref{sec:evaluation_4} are affected when the query size changes.

\Cref{exp_extra_paths,exp_extra_stars,exp_extra_cycles} depicts our results.
The main observation is that \RECURSIVE's \TTL benefits more from longer queries (\Cref{exp:6path_syn_small,exp:6cycle_syn_small}) than shorter ones (\Cref{exp:3path_syn_small}).
This makes sense because in a long path, there are more solutions that share the same suffixes and \RECURSIVE essentially reuses the ranking of those common suffixes to sort the entire solution set faster.
In general, the situation that we saw in \cref{exp} is repeated.  
\LAZY is again the winning algorithm for small $k$ across the board, while \MIN generally underperforms.
\EAGER only makes sense in cases where a large number of results is returned if \RECURSIVE cannot reuse computation as it does in the path case, e.g. in the extreme case of star queries (\Cref{exp:3star_syn_small,exp:6star_syn_small})

\resultbox{%
\introparagraph{Results}
\RECURSIVE's $\TTL$ advantage over \NAIVE is more evident in longer queries since there are more opportunities of reusing computation. 
\LAZY again dominates for the first results (small $k$) for all query sizes.}

\subsection{Comparison against PostgreSQL}
\label{sec:psql}

To validate our \NAIVE implementation, we compare it against
PostgreSQL 9.5.20. 
Following standard methodology~\cite{bakibayev12fdb}, we remove the system overhead
as much as possible and make sure that the input relations are cached in memory:
we turn off fsync, synchronous\_commit, full\_page\_writes, we set bgwriter\_delay to the maximum (10 sec), bgwriter\_lru\_maxpages to 0, checkpoint\_timeout to 1 hour and max\_wal\_size to a large value (1000 GB). 
We also give shared\_buffers and work\_mem 32 GB and set the isolation level to the lowest possible (READ UNCOMMITED).
Like before we run 200 instances and report the median result.
For each of those instances, we run PSQL 3 times and time only the last run to ensure that the input relations are cached.

Our results for the synthetic datasets are gathered in \Cref{tab:psql}.
Overall, we found our \NAIVE implementation to be $12\%$ to $54\%$ faster than PSQL.
Although the two implementations are not directly comparable since they are written
in different languages and PostgreSQL is a full-fledged database system,
this result shows that our \NAIVE implementation is competitive with existing
batch algorithms.

\begin{figure*}[h]
\centering
\footnotesize
\begin{tabular}{|c|c|c|c|c|c|c|c|c|}
\hline
 & \makecell{\textbf{3-Path}\\$n=10^5$\\$10^7$ Results} & \makecell{\textbf{4-Path}\\$n=10^4$\\$10^7$ Results} & \makecell{\textbf{6-Path}\\$n=10^2$\\$10^7$ Results} & \makecell{\textbf{3-Star}\\$n=10^5$\\$10^7$ Results} & \makecell{\textbf{4-Star}\\$n=10^4$\\$10^7$ Results} &
 \makecell{\textbf{6-Star}\\$n=10^2$\\$10^7$ Results} & \makecell{\textbf{4-Cycle}\\$n=5 \cdot 10^3$\\$1.25 \cdot 10^7$ Results} & 
 \makecell{\textbf{6-Cycle}\\$n=4 \cdot 10^2$\\$1.6 \cdot 10^7$ Results}\\
 \hline
 \NAIVE & 9.74 & 8.27 & 7.51 & 8.32 & 7.34 & 7.35 & 14.09 & 23.72\\
 \hline
 \PSQL & 12.18 & 13.39 & 16.45 & 11.84 & 13.10 & 16.04 & 30.36 & 26.86\\
 \hline
 $\%$ faster & 20\% & 38\% & 54\% & 30\% & 44\% & 54\% & 54\% & 12\%\\
 \hline
\end{tabular} 
\caption{Seconds to return the full result for \NAIVE and \PSQL on our synthetic data.
}
\label{tab:psql}
\vspace{-2mm}
\end{figure*}

\section{Extensions}

\subsection{Join queries with projections}

So far, we have only considered \emph{full} conjunctive queries, i.e.\ those that can be written in Datalog as $Q(\vec x) \datarule g_1(\vec x_1),\ldots, g_\stages(\vec x_\stages)$ 
where $\vec x = \bigcup_{i=1}^{\stages} \vec x_i$.
A non-full conjunctive query (also called a join query with projection) 
$Q(\vec y) \datarule g_1(\vec x_1),\ldots, g_\stages(\vec x_\stages)$ has $\vec y \subset \vec x$ and asks to return only the \emph{free} variables $\vec y$,
while the remaining variables $\vec x \setminus \vec y$ 
(also called existentially quantified
variables) are projected away.
As mentioned in \cref{sec:cq_def}, our approach covers in principle all conjunctive queries:
For non-full queries, we can perform the enumeration as if they were full and then project the output tuples on the free variables, discarding the duplicates.
However, this approach might not always be ideal.
In this section, we investigate the different possible semantics of ranked enumeration with projections and extend our approach to cover some of these cases more efficiently.

\introparagraph{Two principal ways to define ranked enumeration}
There are at least two reasonable semantics for ranked enumeration over joins queries with projections.
Consider the $2$-path query $Q(x_1) \datarule R_1(x_1, x_2), R_2(x_2, x_3)$ where we want to return only the first attribute $x_1$.
Recall that we assume input weights have been placed on the relation tuples.
What do we do if the same value $v_1$ of $x_1$ appears in two different results 
of the full query 
$(v_1, v_2, v_3)$ and $(v_1, v_2', v_3')$ 
with weights $w$ and $w'$, respectively?
We identify two different semantics:

\begin{enumerate}
\item \emph{\Allweights} 
semantics: The first option is to return $v_1$ twice with both weights $w, w'$ in the correct sequence.
The corresponding SQL query would be:
\begin{verbatim}
    SELECT   R1.A1, R1.W + R2.W as Weight
    FROM     R1, R2
    WHERE    R1.A2=R2.A2
    ORDER BY Weight ASC
    LIMIT    k
\end{verbatim}
In general, we return the results and the weights that the full conjunctive query would return
projected on the variables $\vec y$.\footnote{In the case that two output results have the same weight, we still return both of them.}
Thus, it is trivial to extend our approach to \allweights semantics, as it is essentially equivalent to the ranked enumeration of full CQs:
We enumerate the full CQ $Q(\vec x)$ as before, yet we apply a projection $\pi_{\vec y}(r)$ 
to the output tuples $r$ before returning them.
The guarantees that we get in this case are the same as in \cref{TH:MAIN}.

\item \emph{\Minweight} semantics: 
The second option is to return $v_1$ only once with the best (minimum) of the two weights.
In this case, the SQL query is:

\begin{verbatim}
    SELECT 	 X.A1, X.Weight
    FROM
        (SELECT   R1.A1, MIN(R1.W + R2.W) as Weight
        FROM     R1, R2
        WHERE    R1.A2=R2.A2
        GROUP BY R1.A1) X
    ORDER BY X.Weight
    LIMIT 	k
\end{verbatim}
If $r$ is a result of the query $Q(\vec y)$ denoted by $r \in Q(\vec y)$, we return the results of the query $Q(\vec y)$ 
ranked by weight 
$w(r) = \min_{r' \in Q(\vec y) \::\: \pi_{\vec y}(r') = r} \{ w(r') \}$.
In other words, each returned tuple $r$ has the minimum weight over all tuples $r'$ of the full query $Q(\vec x)$ 
that map to $r$ if projected on $\vec y$.
While it is still possible to apply the projection as a
post-filtering step and get a correct algorithm,
there is no guarantee on the delay.
If a lot of consecutive results project to the same variables $y$, 
then we might have to wait for the next result for a time that can be as high as $\O(|\mathrm{out}|)$ in the worst case.
We next discuss a non-trivial extension that can handle \minweight semantics with logarithmic delay in certain cases.

\end{enumerate}

\introparagraph{\Minweight semantics}
To efficiently handle 
\minweight semantics, we resort to the techniques that have been developed for \emph{unranked} constant-delay enumeration.
Bagan et al.~\cite{bagan07constenum} show that the acylic queries which are \emph{free-connex} 
admit constant delay enumeration after linear-time preprocessing.
Multiple characterizations of these queries exist \cite{Berkholz20tutorial}.
One particularly useful way to identify them is to check for the
acyclicity of the corresponding hypergraph that includes an additional hyperedge connecting the head variables $\vec y$ \cite{brault13thesis}.
For further reading on the topic, we refer the reader to the
paper that introduced the concept \cite{bagan07constenum} and recent surveys and tutorials \cite{Berkholz20tutorial,durand20tutorial,DBLP:journals/sigmod/Segoufin15}.
In the following, we proceed to modify some of the common
techniques for free-connex acyclic queries in order to
accommodate efficient ranked enumeration under \minweight semantics.

Intuitively, unranked constant-delay enumeration on free-connex acyclic queries \cite{Berkholz20tutorial}
works by constructing an appropriate join tree
that groups the free variables together.
The tree is first swept bottom-up with
semi-joins as in the
Yannakakis algorithm \cite{DBLP:conf/vldb/Yannakakis81}
and then pruned so that only the free variables remain.
The answers to the query can then be enumerated as if it were full (with no projections).
We present a modification of this approach for ranked enumeration under \minweight semantics with a logarithmic (instead of constant) delay.
Essentially, it involves replacing the semi-joins with our Dynamic Programming scheme.

\begin{figure*}[tb!]
    \centering    
    \begin{subfigure}{0.3\linewidth}
        \centering
        \includegraphics[width=0.45\linewidth]{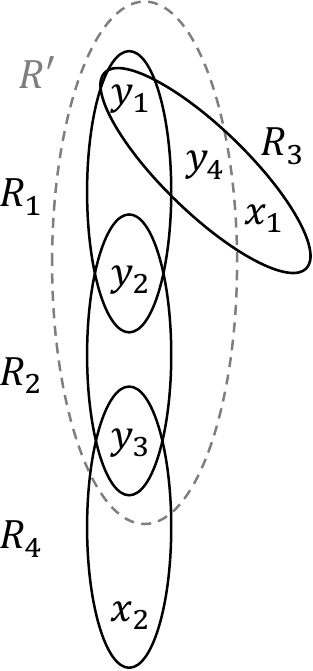}
        \caption{The hypergraph of the query and an additional atom $R'$ that is used to check the free-connex property.}
		\label{fig:fc_query}
    \end{subfigure}%
    \hfill
    \begin{subfigure}{0.4\linewidth}
        \centering
        \includegraphics[width=\linewidth]{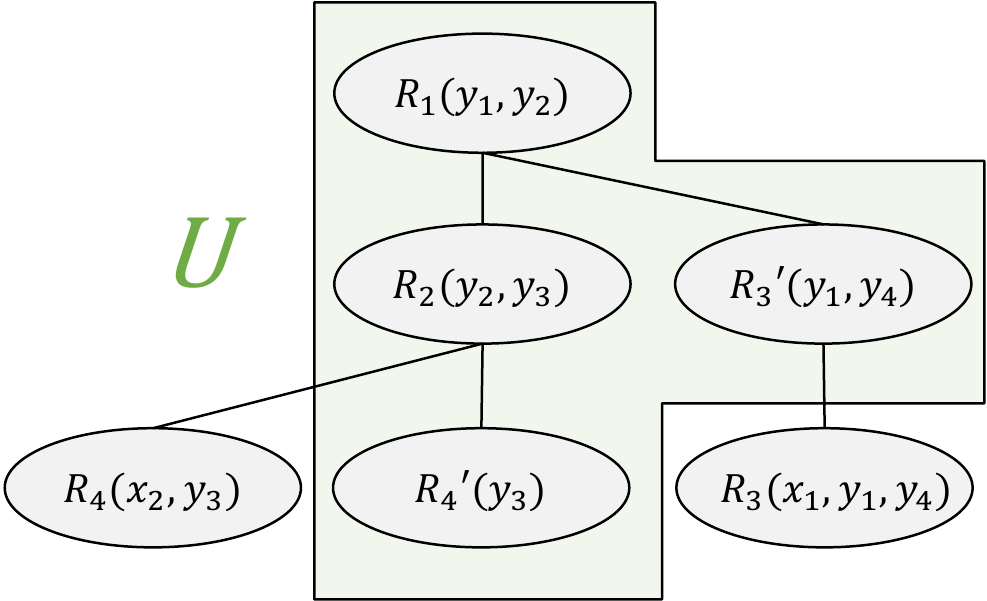}
        \caption{The join tree with a connected subset of nodes $U$ that contain precisely the free variables.}
		\label{fig:fc_join_tree}
    \end{subfigure}%
    \hfill
    \begin{subfigure}{0.25\linewidth}
        \centering
        \includegraphics[width=0.95\linewidth]{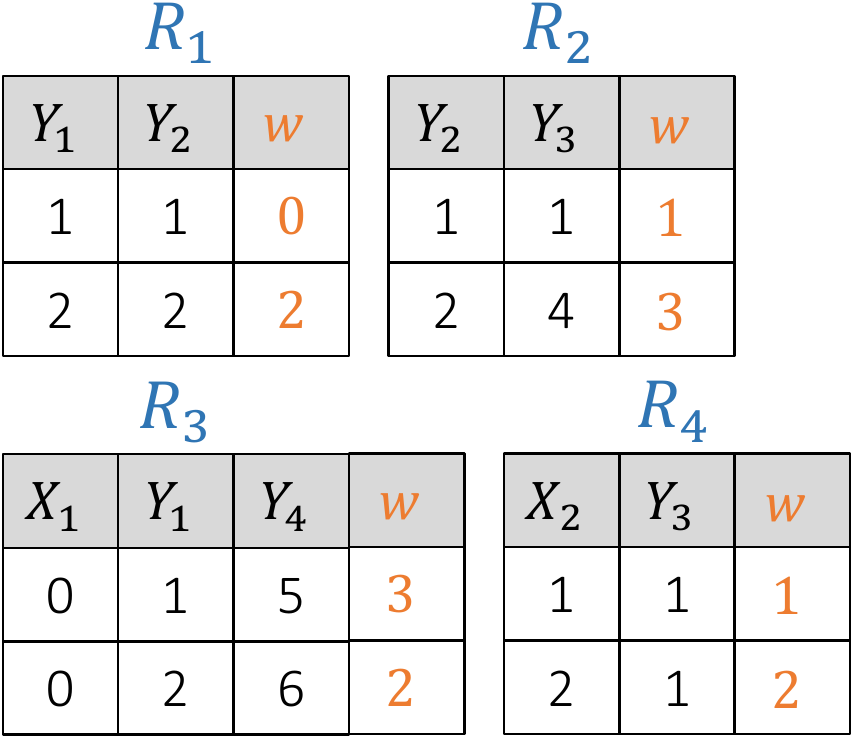}
        \caption{An example database instance.}
		\label{fig:fc_relations}
    \end{subfigure}%

    \begin{subfigure}{0.48\linewidth}
        \centering
        \includegraphics[width=\linewidth]{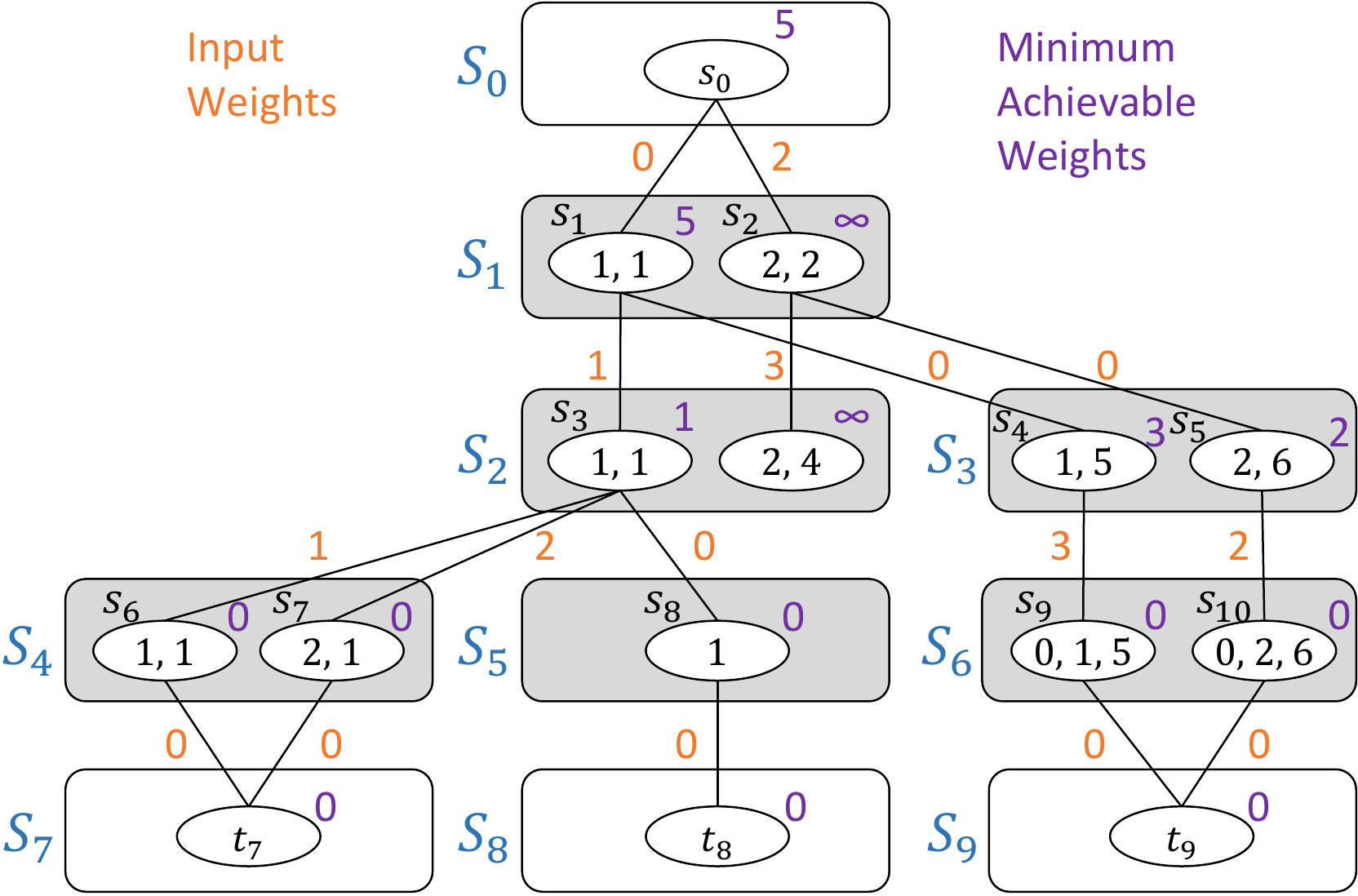}
        \caption{T-DP state-space $T$ that corresponds to the full query $Q(x_1, x_2, y_1, y_2, y_3, y_4)$ using the join tree of \cref{fig:fc_join_tree}.
        For a state $s$, an identifier is depicted on its top-left and $\solW_1(s)$ on its top-right.
        }
		\label{fig:fc_tdp1}
    \end{subfigure}%
    \hfill
    \begin{subfigure}{0.48\linewidth}
        \centering
        \includegraphics[width=\linewidth]{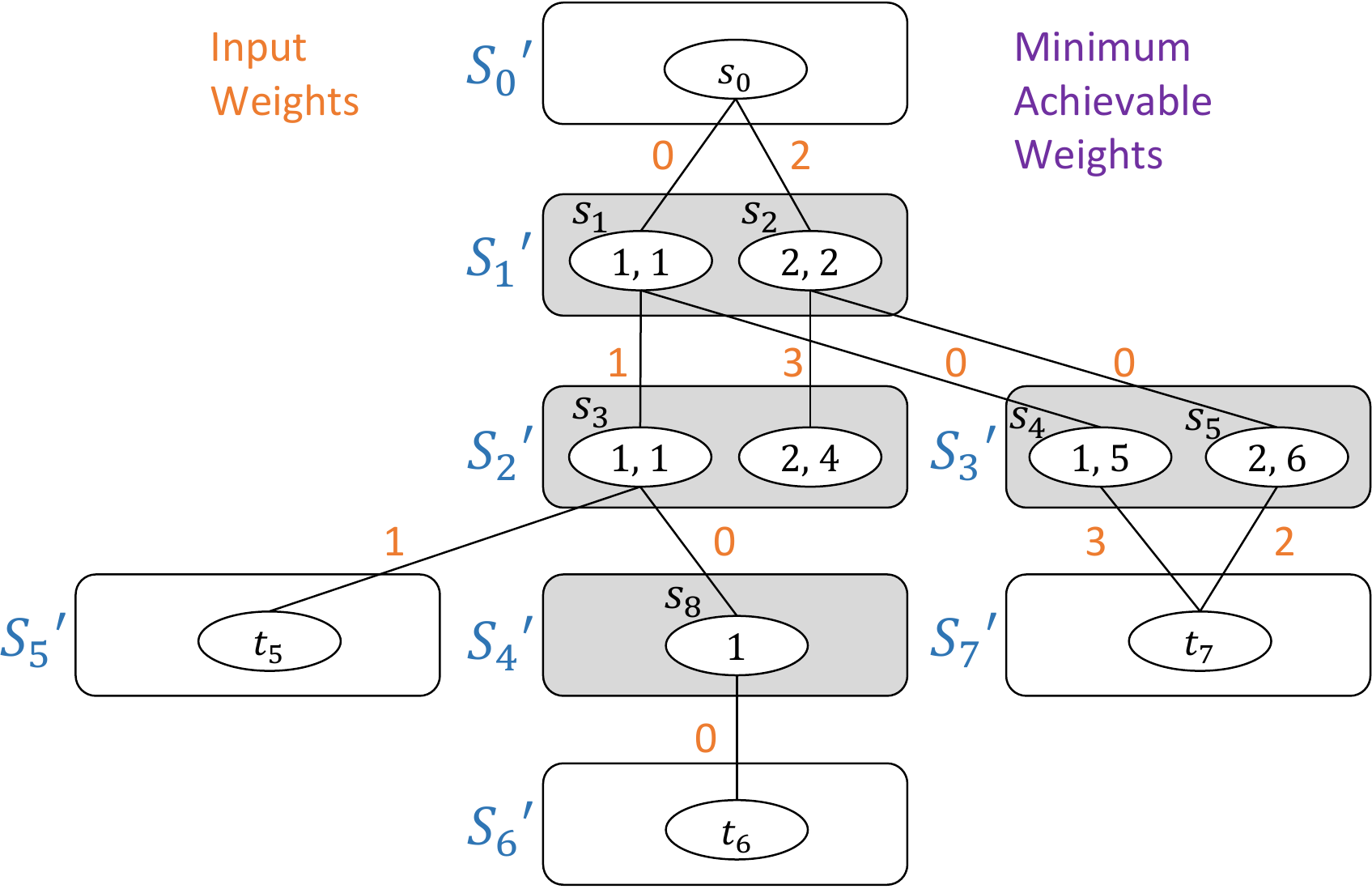}
        \caption{T-DP state-space $T'$ used for the ranked enumeration of $Q(y_1, y_2, y_3, y_4)$.
        Notice that stages not in $U$ have been removed and replaced by artificial terminal ones. The input weights have been modified accordingly.}
		\label{fig:fc_tdp2}
    \end{subfigure}%
    
    \caption{\Cref{ex:fc}: ranked enumeration under \minweight semantics for the acyclic and free-connex query
    $Q(y_1, y_2, y_3, y_4) \datarule R_1(y_1, y_2), R_2(y_2, y_3), R_3(x_1, y_1, y_4), R_4(x_2, y_3)$ 
    on an example database.}
    \label{fig:fc}
\end{figure*}

\begin{example}[Free-connex query]
\label{ex:fc}
Consider the free-connex acyclic query    
$Q(y_1, y_2, y_3, y_4) \datarule R_1(y_1, y_2), R_2(y_2, y_3), R_3(x_1, y_1, y_4), R_4(x_2, y_3)$.
We can check that it is indeed free-connex if we
add an additional relation $R'$ 
that encompasses all the free variables $\vec y$
(\cref{fig:fc_query})
and then verify that the modified query is acyclic
(e.g. by finding a join tree).
Using the algorithm of Brault-Baron~\cite{brault13thesis},
we can construct a join tree for $Q$ such that a
connected subset of nodes $U$ contain all the free variables
and no existentially quantified ones (\cref{fig:fc_join_tree}).
In order to achieve that, we have to introduce two additional atoms with two new relation symbols: 
$R_3' = \pi_{Y_1, Y_4}(R_3)$ and
$R_4' = \pi_{Y_3}(R_4)$.
Given this join tree and a database instance (\cref{fig:fc_relations}), 
we can construct the T-DP state space $T$ shown in \cref{fig:fc_tdp1}.
The non-artificial stages (depicted in gray) correspond to
the nodes of the join tree and are populated by states 
(depicted by small white circles)
that correspond to tuples from the relations.
If we were to run our ranked enumeration algorithms on $T$,
we would enumerate the answers to the full query 
$Q(x_1, x_2, y_1, y_2, y_3, y_4)$.
Instead, we only run the bottom-up phase that computes the values $\solW_1(s)$ for all states $s$, 
shown with purple color in the figure.
We proceed by removing the stages that do not belong to $U$ and replacing them with artifical terminal nodes, thereby getting a modified state-space $T'$ shown in \cref{fig:fc_tdp2}.
Observe that ranked enumeration on $T'$ will now enumerate the answers to the query $Q(y_1, y_2, y_3, y_4)$.
To get the correct \minweight semantics, we also have to modify the input weights on $T'$.
Consider state $s_3 \in S_2$ on $T$ that has two branches, one towards $S_4$ and one towards $S_5$.
For the first branch, we have to choose among two decisions $(s_3, s_6)$ and $(s_3, s_7)$.
The minimum is achieved with $(s_3, s_6)$ since 
$1 = w(s_3, s_6) + \solW_1(s_6) < w(s_3, s_7) + \solW_1(s_7) = 2$.
Therefore, when we remove stage $S_4$ in $T'$ and replace it with a terminal stage 
$S_5' = \{ t_5 \}$, we set the weight of the 
decision $(s_3, t_5)$ to be equal to that minimum, i.e.,
$w(s_3, t_5) = \min_{(s_3, s) \in \Dec_{24}} 
\big\{ 
\weight(s_3, s) \aggr \solW_1(s)
\big\}
= 1
$.
Notice that the minimum achievable weights per branch that we need have
already been computed from the bottom-up phase on $T$ (see \cref{eq:TDP_recursion}).
\end{example}

\begin{theorem}[Free-connex acyclic queries]
Ranked enumeration of the results of a free-connex acyclic query under \minweight semantics can be performed with $\TTF = \O(n)$ and $\Del(k) = \O(\log k)$ in data complexity.
\end{theorem}
\begin{proof}
Let $\vec y$ be the set of free variables of a  free-connex acyclic query $Q(\vec y)$.
Recall that in a join tree, each node represents one atom of the query --
for a node $t$, let $\varset(t)$ denote the set of variables of the corresponding atom.
Since the query is free-connex, we can compute 
a join tree with a connected subset of nodes $U$ that satisfy $\bigcup_{t \in U} \varset(t) = \vec y$ 
in $\O(|Q|)$ time using known techniques \cite{Berkholz20tutorial,brault13thesis}.
In order to achieve this, some additional atoms might be introduced in the query.
Set the input weights of all the tuples or T-DP states materialized from those atoms to $0$.
Also set the root of the tree to be some node $t \in U$.

Next, from the join tree, construct in a bottom-up fashion the corresponding T-DP state space graph as in \cref{sec:tdp} and denote it as $T$.
This takes $\O(n)$ time.
Every T-DP solution of $T$ 
is by construction an answer to the full query $Q(\vec x)$ (with no projections).
Given that (1) $U$ contains all the free variables $y$ needed for answering $Q(\vec y)$ and (2) bottom-up consistency has already been enforced, 
it suffices to perform ranked enumeration only to the subtree induced by $U$.
To do that, create a copy $T'$ of $T$ 
that only retains the stages
that belong to $U$
\footnote{By slightly abusing the notation, we say that a stage belongs to $U$ if its corresponding node in the join tree belongs to $U$}.
Complete $T'$ with an artificial starting stage as the root of the tree
and terminal stages as the leaves, exactly as in \cref{sec:tdp}.
We argue that 
there is a 1 to 1 correspondence between the T-DP solutions 
of $T'$ and the answers to $Q(\vec y)$.
To see this, first consider a T-DP solution $\sol$ of $T'$.
It must contain states that belong to $\SsetR$, 
(recall that these are the ones that were not removed from the bottom-up pass) 
hence they can reach the terminal nodes of the original T-DP state space graph $T$.
Thus, there is a way to extend $\sol$ to a solution to the original state-space $T$, 
which corresponds to an answer to the full query $Q(\vec x)$.
Thus, the values assigned to the variables $\vec y$ are an answer to $Q(\vec y)$.
Conversely, consider an answer to $Q(\vec y)$ which is an assignment of values to the the $y$ variables.
Since the subset $U$ of the join tree contains precisely the free variables $y$, 
we can find tuples in the materialized relations of the join tree and equivalently, states in $T'$ that form a T-DP solution using those values.

To get \minweight semantics, we have to make adjustments to the input weights of $T'$.
In particular, we introduce some additional terminal stages 
to $T'$ and set the weights of the decisions that reach them according to the weights of $T$ that have been removed from $T'$.
Let $S_r$ be a (non-artificial) stage in $T$ and $S_p = \parent(S_r)$ its parent such that $S_r \notin U$ and $S_p \in U$.
Also let $S_p'$ be the copy of $S_p$ in $T'$.
We add to $T'$ a stage $S_t' = \{ s_t' \}$ and decisions 
$(s_p', s_t')$ for $s_p' \in S_p'$ with $\solW_1(s_p) \neq \infty$.
Notice that this does not add or remove any T-DP solutions from $T'$.
The weight of the new decisions is set to be the minimum achievable weight that $s_p$ could reach in $T$ from the $S_r$ branch:
$w(s_p', s_t') = \min_{(s_p, s_r) \in \Dec_{pr}} =
\big\{ 
\weight(s_p, s_r) \aggr \solW_1(s_r)
\big\}$.
This can be done in time linear in $T$.
For a solution $\sol$ of $T$ and a solution $\sol'$ of $T'$, let $\sol' \subset \sol$ denote the fact that $\sol$ is an extension of $\sol'$, i.e., they agree on the subset $U$.
By our construction, it is easy to see that for a solution $\sol'$ of $T'$ we get $w(\sol') = \min_{\sol : \sol' \subset \sol} w(\sol)$.

The total time spent so far is linear in the size of the database.
After we modify the weights of the decisions as described above,
we perform a bottom-up pass on $T'$ once more
and finally, apply \HEAP on $T'$ to get ranked enumeration with $\O(\log k)$ delay.
\end{proof}

We proceed to strengthen the above result with lower bounds 
that also originate from the works on unranked enumeration.
First, we state some complexity-theoretic assumptions that are commonly used and on which we will rely on.
For more information on these conjectures, we refer the reader to
the extensive discussion by Berkholz et al. \cite{Berkholz20tutorial}.
\begin{itemize}
    \item \BMM is the hypothesis that two $n \times n$ Boolean matrices cannot be multiplied over the Boolean semiring in time $\O(n^2)$.
    \item \SBMM is the hypothesis that two Boolean matrices cannot be multiplied over the Boolean semiring in time $\O(m)$, where $m$ is the number of the non-zero entries in the input and the output.
    \item \TRIANGLE is the hypothesis that a triangle cannot be identified in a graph of $m$ edges within $\O(m)$ time.
    \item \HYPERCLQ is the hypothesis that a $(k+1, k)$-hyperclique cannot be identified in a $k$-uniform hypergraph within $\O(m)$, where $m$ is the number of hyperedges and $k \geq 3$. 
    Note that a $(k+1, k)$-hyperclique is a set of $k+1$ vertices such that every $k$-element subset is a hyperedge
    and in a $k$-uniform hypergraph, all hyperedges contain $k$ vertices.
\end{itemize}
Besides the efficient algorithm for free-connex acyclic queries,
Bagan et al.~\cite{bagan07constenum} provide a complementary negative result for queries without self-joins that rests on the \BMM hypothesis.
In particular, if a self-join-free acyclic query is not free-connex then we cannot enumerate its answers with constant delay after linear preprocessing.
This essentially creates a dichotomy for the class of self-join-free acyclic queries: 
the only ones that can be handled efficiently are those that are free-connex.
Later developments replace the original \BMM assumption with \SBMM, which is considered more likely to be true \cite{Berkholz20tutorial}
and further extend the dichotomy to all self-join-free conjunctive queries with some additional assumptions.
We restate this result below:

\begin{theorem}[\cite{bagan07constenum,brault13thesis}]
Assuming \SBMM, \TRIANGLE and \HYPERCLQ, 
unranked enumeration of the results of a self-join-free conjunctive query 
that is not acyclic free-connex 
cannot be done with $O(n)$ preprocessing 
and $\O(\log n)$ delay. 
\end{theorem}

The following is immediate, since ranked enumeration
is a harder problem than unranked enumeration 
and $\log k = \O(\log n)$ in data complexity:

\begin{corollary}[Ranked enumeration of conjunctive queries]
Assuming \SBMM, \TRIANGLE and \HYPERCLQ, 
ranked enumeration of the results of a self-join free conjunctive query under \minweight semantics
can be performed with $\TTF = \O(n)$
and $\Del(k) = \O(\log k)$ delay
in data complexity
if and only if the query is acyclic and free-connex.
\end{corollary}

\subsection{Minimum-cost homomorphism}
\label{sec:mincosthomo}

The connections between conjunctive query evaluation, constraint satisfaction, 
and the hypergraph homomorphism problem are well-known
\cite{DBLP:conf/stoc/ChandraM77,DBLP:journals/siamcomp/FederV98,KOLAITIS2000302}.
We now apply our framework to the minimum-cost homomorphism problem
and generalize it in the same spirit 
as we generalized a standard Dynamic Programming (find the top-1 solution) to an any-$k$ problem.
In other words, we want to perform ranked enumeration, finding the min cost homomorphism, then the 2nd lowest cost homomorphism, etc.
For that purpose we need to introduce a slight variation of the well-studied hypertree decompositions.

\introparagraph{Pinned hypertree decomposition}
A hypergraph $\H(N,E)$ 
is a pair $\H = (N, E)$ 
where $N$ is a set of elements called \emph{nodes}, 
and $E$ is a set of non-empty subsets of $N$ (i.e.\ $E \subseteq 2^N \setminus \emptyset$) called \emph{edges}.

\begin{definition}[TD~\cite{RobertsonS:1986}]
\label{def:TD}
A tree decomposition  
of a hypergraph $\H(N, E)$
is a pair $\langle T, \chi \rangle$ 
where $T = (V, F)$ is a tree, and $\chi$ is a labeling function 
assigning to each vertex $v \in V$ 
a set of vertices $\chi(v) \subseteq N$, 
such that the following three conditions are satisfied: 
	\begin{enumerate}[nolistsep]
	\item (node coverage) 
	for each node $b \in N$, there exists 
	$v \in V$ 
	such that $b \in \chi(v)$; 
	
	\item (edge coverage)
	for each hyperedge $h \in E$, 
	there exists $v \in V$ 
	such that $h \subseteq \chi(v)$; 
	and 
	
	\item (coherence)
	for each node $b \in N$, 
	the set 
	$\chi^{-1}(b) =\{ v \in V \mid b \in \chi(v) \}$ 
	induces a connected subtree of $T$.

	\end{enumerate}
\end{definition}

To distinguish between vertices of $\H$ and $T$, 
we will denote the former \emph{nodes}  $N$,
and the latter \emph{vertices} $V$.
Thus, we also call the set $\chi(v)$ for $v \in V$ the nodes of $v$.

\begin{definition}[HD~\cite{GottlobLS:2002}]
A (generalized) hypertree decomposition \HD{}
of a hypergraph $\H$ is a 
triple $\HD = \langle T, \chi, \lambda \rangle$, 
called a hypertree for $\H$, 
where $\langle T, \chi \rangle$ 
is a tree decomposition of $\H$, 
and $\lambda$ is a function labeling the vertices of $T$ by sets of hyperedges of $\H$ such that, 
\begin{enumerate}[nolistsep]
	\setcounter{enumi}{3}
	\item
	for each vertex $v$ of $T$, $\chi(v) \subseteq \bigcup_{h \in \lambda(v)} h$.\footnote{Notice we use continuous numbering for the conditions as they build upon each other and we will need 6 conditions in total for our formulation.}
\end{enumerate}
\end{definition}

In other words, the additional condition is that all nodes in the $\chi$ labeling of the TD are covered by hyperedges in the $\lambda$ labeling.

A rooted hypertree decomposition $\langle T, \chi, \lambda, r \rangle$ of $\H$ 
is obtained by additionally choosing a root $r \in V$,
which defines a child/parent relation between every pair of adjacent vertices, 
and ancestors/descendants in the usual way: 
In a rooted tree, the parent of a vertex is the vertex connected to it on the path to the root; 
every vertex except the root has a unique parent. 
A child of a vertex v is a vertex of which v is the parent. 
A descendant of any vertex $v$ is any vertex which is either the child of $v$ or is (recursively) 
the descendant of any of the children of $v$.
We write $p(v)$ for the parent of node $v$,
$C(v)$ for the set of children,
and $D(v)$ for the set of descendents. 
A node without children is called a leaf.

Define as \emph{reverse tree-order} the partial ordering on the vertices $V(T)$ with 
$u \leq_T v$ if and only if the unique path from $u$ to the the root passes through $v$~\cite{Diestel:2005}.
if $u <_T v$ we say that $u$ lies below $v$ in T.
We call 
\begin{align*}
	\lceil v \rceil := \{ u | u \leq_T v \}
\end{align*}
the \emph{down-closure} of $y$.
In other words,
$\lceil v \rceil = D(v) \cup \{ v \}$.

Let $\langle T, \chi, r \rangle$ be a rooted tree decomposition of $\H$. 
For a node $v \in V$, 
we denote $\chi(\lceil v \rceil) = \bigcup_u \chi(u)$, 
with $u \in \lceil v \rceil$. 
In other words, $\chi(\lceil  v \rceil)$ contains any node that is contained in either $v$ or any of its descendants.
We also define the subgraph $\H( \lceil v \rceil)$  as 
$\H( \lceil v \rceil) = \H[\chi( \lceil v \rceil)]$.

A vertex $v$ \emph{introduces} node $b$
if $b \in \chi(v) \setminus \bigcup_c \chi(c)$, 
where the union is taken over all children $c \in C(v)$ of $v$~\cite{BodlaenderBL:2013}.
In other words, a vertex introduces a node if that node is contained in the vertex but none of its children.
Analogously, a vertex $v$ \emph{forgets} (or ``projects away'') node $b$
if $b \in \bigcup_c \chi(c) \setminus \chi(v)$.
Since a node can only be present in a connected set of vertices (forming a subtree), 
each node can be introduced multiple times, but only forgotten once.

DP algorithms rely on the following two key properties, 
which follow easily from \cref{def:TD}:
($i$) first, $\H(\lceil r \rceil) = \H$;
($ii$) second, 
for every $v \in V$, 
the only vertices of $\H( \lceil  v \rceil )$ that 
(in $\H$) may be incident with edges that are not in $\H( \lceil v \rceil )$ are vertices in $\chi(v)$.

For our particular formulation of dynamic programming (DP) over hypertree decompositions,
we need to add one more labeling function to the hypertree decomposition. We explain the intuition first:
In a \HD{}, an edge $h \in E$ can be mapped to multiple vertices in $\HD$. 
However, when we add up the weights of a homomorphism, we need to make sure that weights are counted only once.
We thus use a ``pinning'' function that maps each $h$ to exactly one vertex in $T$ in which $h$ appears without projection 
(i.e., $\chi(v) \supseteq h$). 
We say that vertex $v$ ``pins'' edge $h$.\footnote{In the literature of tree decompositions of NP-hard problems, this problem is often solved by defining a ``nice'' tree decomposition that can be constructed from arbitrary tree decompositions with polynomial overhead, and subtracting weights at special vertices called ``joins.'' This construction requires $\oplus$ to be a group instead of a simpler monoid (because we need an inverse element). Because in our work, we do care about polynomial overhead, we prefer to use a definition that avoids requiring an inverse operation.}

\begin{definition}[Pinned HD]
A pinned hypertree decomposition 
of a hypergraph $\H$ is a 
quadruple $\langle T, \chi, \lambda, \rho \rangle$, 
where $\langle T, \chi, \lambda \rangle$ is a \HD{}
and $\rho$ is a function labeling the vertices of $T$ by sets of hyperedges of $\H$ such that, 
\begin{enumerate}[nolistsep]
	\setcounter{enumi}{4}
	\item
	for each $h \in E$,
	there exists \emph{exactly one} $v \in V$,
	such that $h \subseteq \rho(v)$ and $\chi(v) \supseteq h$; and
	
	\item
	for each $v \in V: \rho(v) \subseteq \lambda(v)$.
\end{enumerate}
\end{definition}

\introparagraph{Min Cost Homomorphism} 
Let $\H$ and $\G$ be hypergraphs, possibly with loops. 
A homomorphism from $\H$ to $\G$ is a function 
$\theta : V(\H) \rightarrow V(\G)$ such that for all 
$h \in E(\H)$, 
$\theta(h) \in E(\G)$.
Let 
$w : E(\G) \rightarrow \R$ be a cost function.
The cost of a homomorphism $\theta$ from $\H$ to $\G$ is then
\begin{align*}
	w(\theta) &= \bigoplus_{h \in E(\H)} w \big(\theta(h)\big)
\end{align*}	

The Minimum Cost Homomorphism problem (MCH) is then defined as

\begin{definition}[Min Cost Homomorphism (MCH)]
Given hypergraphs $\H$ and $\G$, 
with edge weights $w$, 
decide whether a homomorphism from $\H$ to $\G$ exists, 
and if so, compute one of minimum weight.
\begin{align}
	w^* = \min_{\theta} \big\{ w(\theta) \big\} \label{eq:mincost}
\end{align}
\end{definition}

Let $\langle T, \chi, \lambda, \rho, r \rangle$
be a rooted pinned hypertree decomposition of $\H$. 
For a node $v \in V (T)$, 
our DP computes values 
$\val(v, \theta)$ for every 
$\theta : \chi(v) \rightarrow N(\G)$, 
defined as follows:
\begin{align*}
	\val(v, \theta) = \min \big\{w(\mu) \mid \mu : 
	\chi( \lceil v \rceil) 
	\rightarrow 
	N(\G) \textrm{ s.t. } \mu|_{\chi(v)} = \theta \big\}.
\end{align*}

So $\val(v, \theta)$ is the minimum weight of a homomorphism $\mu$ from 
$\H(\lceil u \rceil)$ to $\G$ that coincides with $\theta$. 
Then since $\H(\lceil r \rceil) = \H$, 
the minimum weight of a homomorphism from $\H$ to $\G$ is computed by taking the minimum value of 
$\val(r, \theta)$ 
over all 
$\mu : \chi(r) \rightarrow N(\G)$.

The values 
$\val(v, \theta)$
can then be computed recursively in any sequence 
consistent with the reverse tree order $\leq_T$ as follows (\cref{alg:MCH-DP}):
\begin{enumerate}[nolistsep]

	\item
	If $v$ is a leaf node, then initialize the weights with all grounded weights $\val(v,\theta) = w(\theta)$ (\cref{line:leaf}).

	\item
	If $v$ is not a leaf node, then first let $S$ be the set of variables that appear in either $v$ or its children (\cref{line:defS}).
	Then for each homomorphism $\mu$ of the nodes in $S$ to nodes $N(\G)$, 
	determine the cost for the sum of: 
	($i$) cost inherited all children consistent with $\mu$, and
	($ii$) additional cost incurred at that node $v$ for all pinned edges $h \in \rho(v)$ (\cref{line:homo}).	
	\begin{align*}
		\val(v, \mu) = &\bigoplus_{c \in C(v)} \val(c, \mu|_{X(c)} ) \oplus 
					    \bigoplus_{h \in \rho(v)} 	w \big( \mu(h) \big)  
	\end{align*}
	Then determine the minimum over all $\mu$ consistent with the variables in $v$ (\cref{line:takemin}).

	\item
	Finally, determine the minimum weight of all homomorphisms consistent with the root (\cref{line:end}).
\end{enumerate}

The minimum weight homomorphism $\theta^*$ can then be reassembled in one pass forward from the root to the leaves 
in the standard way as explained in \cref{sec:nsdp}. (To simplify the exposition, 
we are following the example of Bertsekas~\cite{bertsekas05dp}: all dynamic programming algorithms construct the solution from the trace by which the optimum cost is found.)

\begin{theorem}\label{th:mch-dp}[Correctness of \FuncSty{MCH-DP}]
	If the operation $\oplus$ is commutative, associative, and nondecreasing, 
	then \cref{alg:MCH-DP} finds all optimal solutions of \cref{eq:mincost}.
\end{theorem}	

\begin{proof}
We prove that \FuncSty{MCH-DP} finds the optimal value $w^*$ of $w(\theta)$. 
At each vertex $v$ of the hypertree decomposition,
the computation eliminates (``forgets'') a set of nodes $T$ from its children:
$T \define \bigcup_c \chi(c) \setminus \chi(v)$.
This elimination replaces all homomorphisms $\mu$ involving $T$ and $\chi(v)$
with the restricted homomorphism $\theta$ 
where the image contains only variables in $\chi(v)$ (\cref{line:end}).
So it suffices to show that the elimination of 
$T$ does not change the minimum value. 

The proof succeeds by induction. 
Consider a vertex $v$ that forgets variables $T$ 
and for which none of its descendants forgets any variable.
Let $\rho(\lceil v \rceil)$ be the set of hyperedges $h \in E(\H)$ that were pinned by any of $\lceil v \rceil$:
$\rho(\lceil v \rceil) = \{h \in E(\H) \mid \exists u \in \lceil v \rceil: h \in \rho(u) \}$.
Similarly, let $\rho(D(v))$ be the set of hyperedges $h \in E(\H)$ that were pinned by any of the descendants of $v$, but not $v$.
Further, let $\theta = \mu|_{\chi(v)}$.

Then
\begin{align}
	&\min_{\theta} \big\{ \bigoplus_{h \in E(\H)} w \big(\theta(h)\big) \big\} = 
	\notag
	\\
	&\min_{\mu} 
		\big\{
			\val(v, \mu|_{\chi(v)}) 
		\oplus
			\bigoplus_{h \in E(\H) \setminus \rho[v]} w \big(\mu(h)\big) \label{eq:right}
	\big\}	
\end{align}

But the right-hand side \cref{eq:right} is now
\begin{align*}
	&\min_{\mu} 
	\big\{
		\min_{\theta} 
		\big\{
			\bigoplus_{h \in \rho[v]} w \big(\theta(h)\big) 		
		\big\}
		\oplus
		\bigoplus_{h \in E(\H) \setminus \rho[v]} w \big(\theta(h)\big) 
	\big\}	
\end{align*}

By the commutativity, associativity, and nondecreasing monotonicity of $\oplus$, 
this expression is equal to the left-hand side \cref{eq:right}.
By straight-forward induction, it follows that $\min_{\theta}\val(r, \theta)$ is the optimal value \cref{eq:mincost}.
\end{proof}

\setlength{\textfloatsep}{0.1cm}
\begin{algorithm}[t]
	\smaller
	\caption{DP Formulation for Minimum Cost Homomorphism Problem over Pinned Hypertree Decomposition}\label{alg:MCH-DP}
\SetKwInput{Algorithm}{Algorithm}
\SetKwFunction{AAA}{AAA}
\Algorithm{\FuncSty{MCH-DP}}

\KwIn{Hypergraphs $\H$ and $G$. 
Cost function $w : V(\G) \rightarrow \R$.
Rooted pinned \HD{} of $\H$: $\langle T, \chi, \lambda, \rho, r \rangle$ with reverse tree order $\leq_T$.}
\KwOut{$w^* = \min_{\theta} \big\{ w(\theta) \big\}$}
\BlankLine

\BlankLine
\For{\text{\normalfont each vertex $v \in V(T)$ in order $\leq_T$}}{
	\uIf{$v$ is a leaf node}{
		$\val(v,\theta) = w(\theta)$ \label{line:leaf}
	}
	\Else{
		Let $S = \chi(v) \cup \bigcup_c \chi(c)$ with $c \in C(v)$\; \label{line:defS}
		Define a homomorphism
		$\mu : S \rightarrow N(\G)$ s.t.
			$\val(v, \mu) = \bigoplus_{c \in C(v)} \val(c, \mu|_{X(c)} ) \oplus 
						    \bigoplus_{h \in \rho(v)} 	w \big( \mu(h) \big)$\; \label{line:homo}
		$\val(v, \theta) = 
			\min \big\{  \val(v, \mu) \mid \mu : S \rightarrow N(\G), \textrm{ s.t. } \mu|_{\chi(v)} = \theta  \big\}$ \label{line:takemin}
	}
}
\BlankLine

\Return{$\min_{\theta} \big\{ \val(r, \theta) \big\}$ \label{line:end}} 

\end{algorithm}

\section{Related Work}
\label{sec:relatedWork}

\introparagraph{Top-$k$}
Top-k queries received significant attention in the database community
\cite{Agrawal:2009:CJA:1546683.1547478,akbarinia11topk,DBLP:conf/vldb/BastMSTW06,bruno02,chaudhuri99,ilyas08survey,DBLP:journals/vldb/TheobaldBMSW08,tsaparas03topk}.
Much of that work relies on the value of $k$ given in advance in order to
\emph{prune the search space}. Besides, the cost model introduced by the seminal
Threshold Algorithm (TA)~\cite{fagin03} only accounts for the \emph{cost of fetching} input tuples
from external sources. Later work such as
J*~\cite{natsev01},
Rank-Join~\cite{ilyas04},
LARA-J* \cite{mamoulis07lara},
and a-FRPA \cite{finger09frpa}
generalizes TA to more complex join patterns,
yet also focuses on minimizing the number of accessed input tuples.
While some try to find a balance between the cost of accessing tuples and the
cost of detecting termination, previous work on top-k queries is
\emph{sub-optimal when accounting for all steps of
the computation}, including intermediate result size
(see~\Cref{sec:top_suboptimal}).
We also refer the reader to a recent tutorial \cite{tziavelis20tutorial} that explores
the relationship between top-k and the paradigms discussed in this paper.

\introparagraph{Optimality in Join Processing}
Acyclic Boolean queries can be evaluated optimally in $\O(n)$ data complexity by the Yannakakis algorithm \cite{DBLP:conf/vldb/Yannakakis81}.
The AGM bound~\cite{AGM}, a tight bound on the worst-case output size for full conjunctive queries,
motivated worst-case optimal algorithms
\cite{navarro19wco,
ngo2018worst,
Ngo:2014:SSB:2590989.2590991,
veldhuizen14leapfrog}
and was extended to more general scenarios, such as the presence of functional dependencies~\cite{gottlob12fds}
or degree constraints~\cite{abo2016degree,khamis17panda}. 
The upper bound for cyclic Boolean CQs was improved over the years with decomposition methods
that achieve ever smaller width-measures,
such as treewidth~\cite{RobertsonS:1986}, 
(generalized) hypertree width (\ghtw) \cite{GottlobLS:2002,
gottlob03width,
gottlob09generalized,
greco17greedy,
greco17consistency}, 
fractional hypertree width (\fhtw) \cite{grohe14fhtw},
and submodular width (\subw) \cite{Marx:2013:THP:2555516.2535926}.
Current understanding suggests that achieving 
the improvements of \subw over \fhtw 
requires decomposing a cyclic query into a union of acyclic queries.
Our method can leverage this prior work on \subw~\cite{khamis17panda,Marx:2013:THP:2555516.2535926}
to match the \subw bound of Boolean CQs for TTF.
We also show that it is possible to achieve better complexity for TTL 
than sorting the output of any of these batch computation algorithms.

\introparagraph{Unranked enumeration of query results}
Enumerating the answers to CQs \emph{with projections} in no particular order
can be achieved only for some classes of CQs with constant delay,
and much effort has focused on identifying those classes
\cite{bagan07constenum, 
  Berkholz:2017:ACQ:3034786.3034789,
  DBLP:conf/pods/CarmeliK19,
  DBLP:journals/sigmod/Segoufin15,
  segoufin_et_al:LIPIcs:2017:7060,
  }.
If the ranking function is defined over the Boolean semiring, our technique achieves
constant delay if we replace the priority queues with simple unsorted lists.
However, we consider only \emph{full} CQs, 
eschewing the difficulties introduced by projections and focusing instead on the challenges
of ranking.
A recent paper by Berkholz and Schweikard~\cite{berkholz19submodular}
also uses a union of tree decompositions based on \subw.
Our focus is on the issues arising from imposing a rank on the output tuples, 
which requires solutions for \emph{pushing sorting into such enumeration algorithms}.

\introparagraph{Factorization and Aggregation}
Factorized databases 
\cite{bakibayev12fdb,
olteanu16record,olteanu12ftrees,
DBLP:conf/sum/SchleichOK0N19}
exploit the distributivity of product over union to represent query results compactly
and generalize the results on bounded \fhtw to the non-Boolean case~\cite{olteanu15dtrees}. 
Our encoding as a DP graph leverages the same principles and is at least as efficient space-wise.
Finding the top-1 result is a case of aggregation that is supported by both factorized databases,
as well as the FAQ framework 
\cite{AboKhamis:2019:FAQ:3294052.3319694,abo16faq}
that captures a wide range of aggregation problems over semirings.
Factorized representations 
can also enumerate the query results with constant delay 
according to lexicographic orders of the variables~\cite{bakibayev13fordering},
which is a special case of the ranking that we support
(\Cref{sec:ranking}).
For that to work, \emph{the desired lexicographic order has to agree with the factorization order}; 
a different order requires a restructuring operation 
that could result in a quadratic blowup even for a simple binary join
(see~\Cref{sec:fdb_order} 
for the full example).
Related to this line of work are 
different representation schemes \cite{DBLP:conf/icdt/KaraO18} and
the exploration of the continuum between representation size and enumeration delay \cite{deep18compressed}.

\introparagraph{Ranked enumeration}
Both Chang et al. \cite{chang15enumeration} and
Yang et al. \cite{yang2018any} provide any-$k$ algorithms for
\emph{graph queries} instead of the more general CQs;
they describe the ideas behind \LAZY and \MIN respectively. 
Kimelfeld and Sagiv \cite{KimelfeldS2006} give an any-$k$ algorithm for acyclic queries with polynomial delay. 
Similar algorithms have appeared for the equivalent Constraint Satisfaction Problem (CSP) \cite{DBLP:journals/jcss/GottlobGS18,DBLP:conf/cp/GrecoS11}.
These algorithms fit into our family \RPDP, yet do not exploit common structure
between sub-problems hence have weaker asymptotic guarantees for delay than any of the
any-$k$ algorithms discussed here. After we introduced the general idea of ranked enumeration
over \emph{cyclic} CQs based on multiple tree decompositions~\cite{YangRLG18:anyKexploreDB},
an unpublished paper \cite{deep19} on arXiv proposed an algorithm for it.
Without realizing it, the authors reinvented the REA algorithm~\cite{jimenez99shortest}, which
corresponds to \RECURSIVE, for that specific context.
We are the first to \emph{guarantee optimal time-to-first result and optimal delay for both
acyclic and cyclic} queries. For instance, we return the top-ranked result
of a 4-cycle in $\O(n^{1.5})$, while 
Deep and Koutris \cite{deep19} require $\O(n^{2})$.
{Furthermore, our work (1) addresses the more general problem of ranked enumeration
for DP over a union of trees, 
(2) unifies several approaches that have appeared in the past,
from graph-pattern search to $k$-shortest path, and shows that neither dominates all others,
(3) provides a theoretical and
experimental evaluation of trade-offs including algorithms that perform best for
small $k$, 
and (4) is the first to prove that it is possible to achieve a time-to-last
that asymptotically improves over batch processing by exploiting the stage-wise
structure of the DP problem.}

\introparagraph{$k$-shortest paths}
The literature is rich in algorithms for finding the $k$-shortest paths in general graphs
\cite{azevedo93kshortest,
bellman60kbest, dreyfus69shortest, eppstein1998finding,  
hoffman59shortest, 
jimenez03shortest, jimenez99shortest, 
katoh82kshortest, 
lawler72,
martins01kshortest, martins03kshortest,
yen1971finding}.
Many of the subtleties of the variants arise from issues caused by cyclic
graphs whose structure is more general than the acyclic multi-stage graphs in our
DP problems.
Hoffman and Pavley \cite{hoffman59shortest}
introduces the concept of ``deviations'' as a sufficient condition for finding the
$k^{\textrm{th}}$ shortest path. Building on that idea, Dreyfus~\cite{dreyfus69shortest}
proposes an algorithm that can be seen as a modification to the procedure of
Bellman and Kalaba~\cite{bellman60kbest}. The \emph{Recursive Enumeration Algorithm}
(REA)~\cite{jimenez99shortest} uses the same set of equations as Dreyfus,
but applies them in a top-down recursive manner. 
Our \REDP builds upon REA.
To the best of our knowledge, prior work has ignored the fact that this algorithm
reuses computation in a way that can asymptotically outperform sorting in some cases.
In another line of research, Lawler \cite{lawler72} generalizes an earlier algorithm of 
Murty \cite{murty1968} and applies it to $k$-shortest paths. 
Aside from $k$-shortest paths, the Lawler procedure has been widely used for a variety of problems in the database community~\cite{golenberg11parallel}.
Along with the Hoffman-Pavley deviations, they are one of the main ingredients
of our \RPDP approach.
Eppstein's algorithm \cite{eppstein1998finding,jimenez03shortest} achieves the
best known asymptotical complexity, albeit with a complicated construction
\emph{whose practical performance is unknown}.
His ``basic'' version of the algorithm has the same complexity as \EAGER,
while our \HEAP algorithm matches the complexity of the ``advanced'' version 
for our problem setting where output tuples are materialized explicitly.

\subsection{Detailed comparison to other paradigms}

\subsubsection{WCO join algorithms}
\label{sec:NPRR_and_TTF}

We now show how the \NPRR algorithm \cite{ngo2018worst} fails to find the top ranked result in the same time bound as our approach. The key innovation of such worst-case optimal join algorithms is that they achieve the same complexity as the worst-case size of the output for every query. In the case of a 4-cycle query, \NPRR produces the full join result in $\bigO(n^2)$, a tight worst-case optimal bound. We next demonstrate with the help of the example database $I_1$ in \Cref{fig:NPRR} that it requires $\bigO(n^2)$ for the top-1 result as well, which cannot be easily fixed, whereas the techniques presented in this paper yield $\bigO(n^{1.5})$.

\smallsection{\NPRR execution on $I_1$} 
We follow the general treatment and formalism of \cite{ngo2018worst}.

(Step 1) WLOG, we use the total order of attributes $B \rightarrow C  \rightarrow A  \rightarrow D$, which implies choosing relation $W(A,D)$ 
in the first iteration: 
$f = \{A, D\}$, 
$\overline{f} = \{B, C\}$. 
The implied relations are: 
$E_1 = \{(B, C), (A, B), (C, D)\}$, and 
$E_2 = \{(A, D), (A, B), (C, D)\}$. 

(Step 2) The algorithm will compute an intermediate set of values $L$ for attributes in $\overline{f}$
with a recursive call on $E_1$.
In our example, $L$ will be a set of $(b, c)$ pairs that satisfy $R(a_i, b), S(b, c), T(c, d_j)$ for some $a_i$ and $d_j$.
Its size is $|L| = 2n$. 

(Step 3) For every $(b, c)$ pair in $L$, there are two ways to find an $(a, d)$ pair that forms a 4-cycle $(a, b, c, d)$:
\begin{itemize}[nolistsep]
    \item \emph{light pair:} if the number of $(a, d)$ pairs joining with $(b, c)$ in $R(a ,b), S(b, c), T(c, d)$ is smaller than $|W(d, a)|$, iterate through those $(a, d)$ pairs. For any such pair, if it is in $W(D, A)$ output $(a, b, c, d)$.
    \item \emph{heavy pair:} if $|W(d, a)|$ is smaller, then check for each $(d, a)$ in $W(D, A)$, whether $(a, b)$ is in $R(a, b)$ and $(c, d)$ is in $T(c, d)$. Output $(a, b, c, d)$ if both conditions are true.
\end{itemize}
All $(b, c)$ pairs in $L$ are light pairs in our $I_2$ example, since the join-generated pairs are of size $n$, while $|W| = 2n$.

\begin{figure}[t]
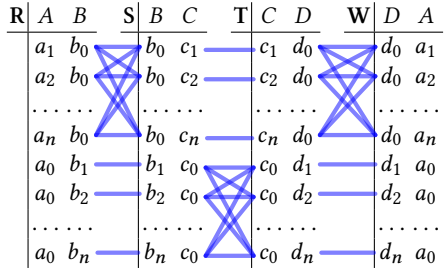

	\small
\centering
	\centering
	\setlength{\tabcolsep}{0.4mm}
		\hspace{0mm}
				\mbox{
				\begin{tabular}[t]{ >{$}c<{$} | >{$}c<{$}  >{$}c<{$} }
	 			\mathbf{R}	&  A 	& B \\
				\hline
				& a_1 & \markZwicky(R1){$b_0$} 	\\
				& a_2 & \markZwicky(R2){$b_0$} 	\\
				& \ldots & \ldots \\
				& a_n & \markZwicky(R3){$b_0$} 	\\
				& a_0 & \markZwicky(R4){$b_1$} 	\\
				& a_0 & \markZwicky(R5){$b_2$} 	\\
				& \ldots & \ldots \\
				& a_0 & \markZwicky(R6){$b_n$}
				\end{tabular}
		}
		\hspace{0mm}
		\mbox{
				\begin{tabular}[t]{ >{$}c<{$} | >{$}c<{$}  >{$}c<{$} }
	 			\mathbf{S}	&  B 	& C \\
				\hline
				& \markZwicky(S1){$b_0$} & \markZwicky(SS1){$c_1$}  	\\
				& \markZwicky(S2){$b_0$} & \markZwicky(SS2){$c_2$} 	\\
				& \ldots & \ldots \\
				& \markZwicky(S3){$b_0$} & \markZwicky(SS3){$c_n$} 	\\
				& \markZwicky(S4){$b_1$} & \markZwicky(SS4){$c_0$} 	\\
				& \markZwicky(S5){$b_2$} & \markZwicky(SS5){$c_0$}	\\
				& \ldots & \ldots \\
				& \markZwicky(S6){$b_n$} & \markZwicky(SS6){$c_0$}
				\end{tabular}
		}
		\hspace{0mm}
		\mbox{
				\begin{tabular}[t]{ >{$}c<{$} | >{$}c<{$}  >{$}c<{$} }
	 			\mathbf{T}	&  C 	& D \\
				\hline
				& \markZwicky(T1){$c_1$} & \markZwicky(TT1){$d_0$}	\\
                & \markZwicky(T2){$c_2$} & \markZwicky(TT2){$d_0$}	\\
				& \ldots & \ldots \\
				& \markZwicky(T3){$c_n$} & \markZwicky(TT3){$d_0$}	\\
				& \markZwicky(T4){$c_0$} & \markZwicky(TT4){$d_1$}	\\
				& \markZwicky(T5){$c_0$}& \markZwicky(TT5){$d_2$}	\\
				& \ldots & \ldots \\
				& \markZwicky(T6){$c_0$} & \markZwicky(TT6){$d_n$}
				\end{tabular}
		}
		\hspace{0mm}
		\mbox{
				\begin{tabular}[t]{ >{$}c<{$} | >{$}c<{$}  >{$}c<{$} }
	 			\mathbf{W}	&  D 	& A \\
				\hline
				& \markZwicky(W1){$d_0$} & a_1  	\\
				& \markZwicky(W2){$d_0$} & a_2  	\\
				& \ldots & \ldots \\
				& \markZwicky(W3){$d_0$} & a_n  	\\
				& \markZwicky(W4){$d_1$} & a_0  	\\
				& \markZwicky(W5){$d_2$} & a_0 	\\
				& \ldots & \ldots \\
				& \markZwicky(W6){$d_n$} & a_0
				\end{tabular}
		}
\caption{
Database $I_1$ showing sub-optimality of \NPRR for TTF.
}
\label{fig:NPRR}
\tikzZwicky[blue](R1.east)(S1.west)
\tikzZwicky[blue](R1.east)(S2.west)
\tikzZwicky[blue](R1.east)(S3.west)
\tikzZwicky[blue](R2.east)(S1.west)
\tikzZwicky[blue](R2.east)(S2.west)
\tikzZwicky[blue](R2.east)(S3.west)
\tikzZwicky[blue](R3.east)(S1.west)
\tikzZwicky[blue](R3.east)(S2.west)
\tikzZwicky[blue](R3.east)(S3.west)
\tikzZwicky[blue](R4.east)(S4.west)
\tikzZwicky[blue](R5.east)(S5.west)
\tikzZwicky[blue](R6.east)(S6.west)
\tikzZwicky[blue](SS1.east)(T1.west)
\tikzZwicky[blue](SS2.east)(T2.west)
\tikzZwicky[blue](SS3.east)(T3.west)
\tikzZwicky[blue](SS4.east)(T4.west)
\tikzZwicky[blue](SS4.east)(T5.west)
\tikzZwicky[blue](SS4.east)(T6.west)
\tikzZwicky[blue](SS5.east)(T4.west)
\tikzZwicky[blue](SS5.east)(T5.west)
\tikzZwicky[blue](SS5.east)(T6.west)
\tikzZwicky[blue](SS6.east)(T4.west)
\tikzZwicky[blue](SS6.east)(T5.west)
\tikzZwicky[blue](SS6.east)(T6.west)
\tikzZwicky[blue](TT1.east)(W1.west)
\tikzZwicky[blue](TT1.east)(W2.west)
\tikzZwicky[blue](TT1.east)(W3.west)
\tikzZwicky[blue](TT2.east)(W1.west)
\tikzZwicky[blue](TT2.east)(W2.west)
\tikzZwicky[blue](TT2.east)(W3.west)
\tikzZwicky[blue](TT3.east)(W1.west)
\tikzZwicky[blue](TT3.east)(W2.west)
\tikzZwicky[blue](TT3.east)(W3.west)
\tikzZwicky[blue](TT4.east)(W4.west)
\tikzZwicky[blue](TT5.east)(W5.west)
\tikzZwicky[blue](TT6.east)(W6.west)
\end{figure}

\smallsection{Ranked enumeration with \NPRR} 
A straightforward way to to turn this algorithm into a ranked enumeration algorithm is to compute all output tuples 
$(a, b, c, d)$ and then sort them, which incurs $\O(n^2 \log n)$ for TTF. Is it possible to do better than that? 
We will next show that any reasonable attempt to ``retrofit'' this algorithm fails to achieve $\O(n^{1.5})$ TTF for our example.
After Step 2 above, we have $2n$ pairs in $L$, the weight of which is known. Let us revisit Step 3 and break it further into the following parts:

(i) For every $(b, c)$ pair in $L$, 
we know its weight $w_S$ (the subscript refers to the relation of origin). 
We have already established that the number of $(a, d)$ pairs that can be connected from any
$(b, c)$ is 
$1 \cdot n = n < |W(D, A)| = 2n$. 
Since all (b, c) in L are all light pairs, the execution plan is always to compute $(a, d)$ pairs that satisfy $R(a,b), S(b, c), T(c, d)$.
In this step, we can compute their weights as $w_{\mathrm{light}} = w_R + w_S + w_T$. 
Therefore, for each $(b, c)$ pair, 
we have a pool of matching pairs $(a, d)$, each associated with a weight.

(ii) There are $2n$ $(b, c)$ pairs and each one has $n$ matching pairs of $(a, d)$ in its pool. 
To find the result with the minimum weight, we need to go through the pool for each such $(b, c)$ pair. 
For each combination, we have to verify that it is a result by checking $W(D, A)$ and also compute the total weight by adding $w_W$. 
Thus, $\O(n^2)$ in total.

\introparagraph{Experimental results}
To better illustrate our point, we run \NPRR against our algorithms ($\RECURSIVE$ and $\LAZY$) 
on a 4-cycle query $Q_{C4}$ and database $I_1$ (\cref{fig:NPRR}) for various sizes $n$. 
Notice that even though our decomposition method guarantees $\O(n^{1.5})$ for a 4-cycle query, it only needs $\O(n)$ on $I_1$, since every relation has only one heavy value.
We use the same experimental setup as in \Cref{sec:experiments} and plot the time-to-first (TTF) and time-to-last (TTL). 
For \NPRR we only plot the TTF, since TTL is very similar.
We also plot two lines that show the trend of a linear and a quadratic function.

\begin{figure}[t]
    \centering
    \includegraphics[width=0.5\linewidth]{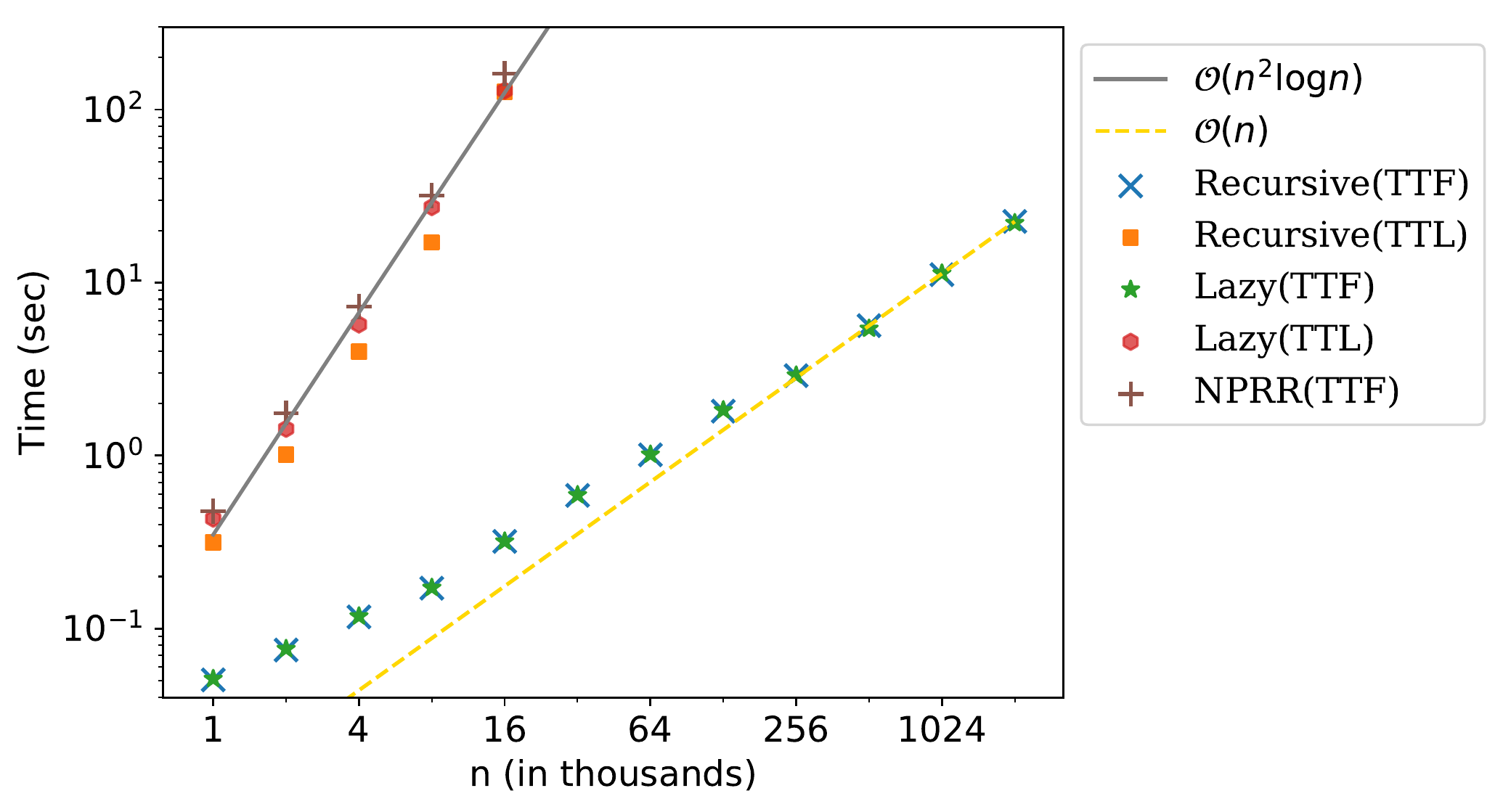}
    \caption{TTF of \NPRR vs our algorithms
	}
	\label{exp:nprr_n2}
\end{figure}

\Cref{exp:nprr_n2} shows the results. 
We can clearly see that \NPRR, as well as the TTL of our algorithms grow quadratically with $n$ and soon become infeasible for large $n$. 
On the contrary, despite an initial overhead for small $n$, the TTF of our algorithms closely follows the linear line and is viable even for $n$ in the order of millions of tuples:
For example, for 16k tuples, our algorithm returns the top-1 result in 300 msec, 
while \NPRR takes over 100 secs.

\begin{figure}[tb]
	\small
\centering
	\centering
	\setlength{\tabcolsep}{0.4mm}
		\hspace{0mm}
				\mbox{
				\begin{tabular}[t]{ >{$}c<{$} | >{$}c<{$}  >{$}c<{$} }
	 			\mathbf{R}	&  A 	& B \\
				\hline
				& $1$ & \markZwicky(R1){$1$} 	\\
				& $2$ & \markZwicky(R2){$1$} 	\\
				& \ldots & \ldots \\
				& $n$ & \markZwicky(R3){$1$}
				\end{tabular}
		}
		\hspace{0mm}
		\mbox{
				\begin{tabular}[t]{ >{$}c<{$} | >{$}c<{$}  >{$}c<{$} }
	 			\mathbf{S}	&  B 	& C \\
				\hline
				& \markZwicky(S1){$1$} & $1$ 	\\
				& \markZwicky(S2){$1$} & $2$	\\
				& \ldots & \ldots \\
				& \markZwicky(S3){$1$} & $n$
				\end{tabular}
		}
\caption{\Cref{ex:fdbs_lexicographic}:
Database showing sub-optimality of factorised databases for certain lexicographic orders.
}
\label{fig:FDBs}
\tikzZwicky[blue](R1.east)(S1.west)
\tikzZwicky[blue](R1.east)(S2.west)
\tikzZwicky[blue](R1.east)(S3.west)
\tikzZwicky[blue](R2.east)(S1.west)
\tikzZwicky[blue](R2.east)(S2.west)
\tikzZwicky[blue](R2.east)(S3.west)
\tikzZwicky[blue](R3.east)(S1.west)
\tikzZwicky[blue](R3.east)(S2.west)
\tikzZwicky[blue](R3.east)(S3.west)
\end{figure}

\subsubsection{Comparison to Factorised Databases}\label{sec:fdb_order}

Factorised databases (FDBs)~\cite{bakibayev12fdb,olteanu16record,olteanu12ftrees,olteanu15dtrees}
support constant-delay enumeration of query results according to a desired lexicographic order on the attributes~\cite{bakibayev13fordering}.
Lexicographic orders are a special case of the ranking function considered in this paper and our approach supports them (see \Cref{sec:ranking}), albeit with logarithmic delay.
Here we look closer at the differences between the two approaches for this special case of lexicographic orders and show that our approach can be asymptotically better in certain cases despite the logarithmic delay. Also note that the lexicographic ordering we describe in \Cref{sec:ranking} is on the relations instead of the attributes but combined with the method in \cref{sec:attribute_weights}, we also support lexicographic orders on the attributes.

First, we provide a very short description of the main idea behind factorised databases and we refer the reader to the original papers for an in-depth presentation. 
To achieve a succinct representation, factorised databases repeatedly apply the distributivity law in an order described by a tree structure whose nodes are the attributes \cite{olteanu12ftrees}. 
Intuitively, if $X$ is the attribute of a node of the tree and $\mathrm{anc}(X)$ are its ancestor attributes, then every value $x \in X$ is represented at most once for each combination of values of $\mathrm{anc}(X)$.
D-representations~\cite{olteanu15dtrees} provide further succinctness by making the dependencies of each attribute in the tree explicit. 
This means that some attributes in $\mathrm{anc}(X)$ might not actually determine what the possible $X$ values are.
Truly dependent ancestor attributes of a node are denoted as $\mathrm{key}(X)$.
Each value $x \in X$ is then represented at most once for each combination of values of $\mathrm{key}(X)$.

These factorised representations provide unranked constant-delay enumeration out-of-the-box.
Yet for specific lexicographic orders, there are two conditions that have to be met: 
($i$) the order-by attributes have to be ``at the top'' of the tree and 
($ii$) the tree order has to agree with the lexicographic order.
If the tree order is not in agreement 
(e.g., we want $A$ before $B$ but $A$ is a child of $B$ in the tree), 
then the whole representation has to be restructured. 
The restructuring operation takes an input representation and transforms it to an 
output representation consistent with the lexicographic order 
in time linear (ignoring log factors) in the input \emph{and output representation sizes}.
However, the output representation itself could be inefficient.
We next illustrate the simplest example where an ill-chosen lexicographic order results in a quadratic representation for a simple binary join.

\begin{example}[Lexicographic orders]\label{ex:fdbs_lexicographic}
Consider the path-2 query $Q_{P2}(A, B, C) \datarule R(A, B), S(B, C)$. 
As usual, $n$ is the maximum number of tuples in a relation.
Ideally, we would want to factorise it using a tree that has $B$ as the root and $A, C$ as its children.
That way, every $A$ and $C$ value in the query result would be represented independently for each $B$ value.
However, for the lexicographic order $A \rightarrow C \rightarrow B$ this factorisation is not in agreement since $B$ comes after $A$ and $C$.
The \emph{only possible} tree that satisfies the condition ($ii$) above is a path from $A$ to $C$ to $B$. 
Note that the tree with $A$ as the root and $B$, $C$ as the children is not possible because of the path condition in factorised databases: attributes that belong to the same relation ($B$ and $C$ here) are in general dependent and have to lie in the same root-to-leaf path.
In that tree, $\mathrm{key}(B) = \{ A, C \}$.
According to Lemma 7.20 in \cite{olteanu15dtrees}, there exist arbitrarily large databases such that the number of $B$ values in the representation is at least $n^{\rho^{*}(B \cup key(B))}$, where $\rho^{*}$ is the fractional edge cover, thus $\Omega(n^2)$.

\Cref{fig:FDBs} presents a concrete instance where this happens.
For this database, the single $B$-value $1$ will be represented once for each combination of $A, C$ values and there are $n^2$ of them. 
In contrast, our approach begins the enumeration after only linear time preprocessing.
Thus in this case, the preprocessing step of FDBs takes $\O(n^2)$ after which results can be enumerated in constant time.
In contrast, our approach has TTF in $\O(n)$ and TTL in $\O(n^2)$ with logarithmic delay.
\end{example}

\subsubsection{Comparison to top-$k$ join algorithms}
\label{sec:top_suboptimal}

Consider the database $I_2$ from \cref{fig:TA} with $\ell=3$ relations and $n=10$ tuples per relation.
The top output tuple is marked in blue; it consists of the lightest tuples from the first $\ell-1$ relations and the heaviest tuple from $R_\ell$.
J*~\cite{natsev01} and Rank-Join~\cite{ilyas04} access 
the tuples in the input relations by decreasing weight. 
Their cost model takes into account only the number of database accesses, 
hence they try to minimize the depth up to which the sorted relations have to be accessed in order to find the top-k results. 
In this case, 
\emph{both J* and Rank-Join will consider the $(n-1)^{\ell-1}$ combinations between $R$ and $S$
before getting to the 
the top-1 tuple $(r_0, s_0, t_0)$}. 
This happens because J* over-estimates their weight by using the large weight of $t_0$ to upper-bound them, 
while Rank-Join by default joins each newly encountered tuple with all the other ones seen so far.
In contrast, our approach achieves $O(n \cdot \ell)$ for the top ranked result.

\begin{figure}[t]
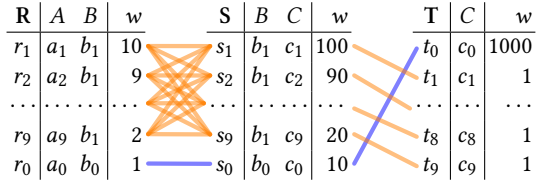

	\small
\centering
	\centering
	\setlength{\tabcolsep}{0.4mm}
		\hspace{0mm}
				\mbox{
				\begin{tabular}[t]{ >{$}c<{$} | >{$}c<{$}  >{$}c<{$} | >{$}r<{$} }
	 			\mathbf{R}	&  A 	& B 	& w\\
				\hline
				r_1		& a_1 	& b_1 	& \markZwicky(1-1){$10$}		\\			
				r_2		& a_2 	& b_1 	& \markZwicky(2-1){$9$}\\
				\cdots 	& \cdots & \cdots 	& \markZwicky(3-1){$\cdots$\vphantom{\large $0$}} \\
				r_9		& a_9 	& b_1 	& \markZwicky(4-1){$2$}	\\				
				r_{0}	& a_{0} & b_0 	& \markZwicky(0-1){$1$}								
				\end{tabular}			
		}
		\hspace{5mm}
		\mbox{
				\begin{tabular}[t]{ >{$}c<{$} | >{$}c<{$} >{$}c<{$} | >{$}r<{$}}
	 			\mathbf{S}	&  B 	& C 	& w\\
				\hline
				\markZwicky(1-2){\vphantom{\large $0$}$s_1$}	& b_1 & c_1 	& \markZwicky(1-3){$100$}	\\			
				\markZwicky(2-2){\vphantom{\large $0$}$s_2$}	& b_1 & c_2 	& \markZwicky(2-3){$90$}  	\\
				\markZwicky(3-2){\vphantom{\large $0$}$\cdots$} & \cdots  & \cdots 	& \markZwicky(3-3){$\cdots$} \\
				\markZwicky(4-2){\vphantom{\large $0$}$s_9$}	& b_1 & c_9 	& \markZwicky(4-3){$20$}	\\				
				\markZwicky(0-2){\vphantom{\large $0$}$s_0$}	& b_0 & c_0 	& \markZwicky(0-3){$10$}								
				\end{tabular}	
		}
		\hspace{5mm}
		\mbox{
				\begin{tabular}[t]{ >{$}c<{$} | >{$}c<{$} | >{$}r<{$}}
	 			\mathbf{T}	& C & w		\\
				\hline
				\markZwicky(0-4){$t_0$}	& c_0	& 1000\\	
				\markZwicky(1-4){$t_1$}	& c_1	& 1\\						
				\markZwicky(2-4){$\cdots$} & \cdots & \cdots \\				
				\markZwicky(3-4){$t_8$}	& c_8	& 1\\								
				\markZwicky(4-4){$t_9$}	& c_9	& 1
				\end{tabular}			
		}
\caption{
Database $I_2$ showing sub-optimality of J* and Rank-Join. (\Cref{sec:top_suboptimal})}
\label{fig:TA}
\tikzZwicky[blue](0-1.east)(0-2.west)
\tikzZwicky[blue](0-3.east)(0-4.west) %

\tikzZwicky[orange](1-1.east)(1-2.west)
\tikzZwicky[orange](1-1.east)(2-2.west)
\tikzZwicky[orange](1-1.east)(3-2.west)
\tikzZwicky[orange](1-1.east)(4-2.west)
\tikzZwicky[orange](2-1.east)(1-2.west)
\tikzZwicky[orange](2-1.east)(2-2.west)
\tikzZwicky[orange](2-1.east)(3-2.west)
\tikzZwicky[orange](2-1.east)(4-2.west)
\tikzZwicky[orange](3-1.east)(1-2.west)
\tikzZwicky[orange](3-1.east)(2-2.west)
\tikzZwicky[orange](3-1.east)(3-2.west)
\tikzZwicky[orange](3-1.east)(4-2.west)
\tikzZwicky[orange](4-1.east)(1-2.west)
\tikzZwicky[orange](4-1.east)(2-2.west)
\tikzZwicky[orange](4-1.east)(3-2.west)
\tikzZwicky[orange](4-1.east)(4-2.west)

\tikzZwicky[orange](1-3.east)(1-4.west)
\tikzZwicky[orange](2-3.east)(2-4.west)
\tikzZwicky[orange](3-3.east)(3-4.west)
\tikzZwicky[orange](4-3.east)(4-4.west)

\end{figure}

\section{Conclusions and Future Work}

We proposed a framework for ranked enumeration over a class of dynamic programming problems 
that generalizes seemingly different problems that to date had been studied in isolation. 
Uncovering those relationships enabled us to propose
the first algorithms with \emph{optimal time complexity for ranked enumeration}
of the results of both cyclic and acyclic full CQs.
In particular, our technique returns the top result in a time that meets the currently best known bounds for Boolean queries, 
and even beats the batch algorithm on some inputs when all results are produced.
It will be interesting to go beyond our worst-case analysis and study the \emph{average-case} behavior~\cite{SA0G18} of our algorithms.

\smallsection{Acknowledgements}

We are grateful to Hung Q. Ngo for reading drafts of this paper and providing valuable feedback. 
This work was supported in part by %
the National Institutes of Health (NIH) under award number R01 NS091421 and by
the National Science Foundation (NSF) under award number CAREER IIS-1762268.
The content is solely the responsibility of the authors and does not necessarily represent
the official views of %
NIH or NSF.

\balance

\printbibliography	%

\clearpage
\appendix

\section{Nomenclature}

\begin{table}[h]
\centering
\small
\begin{tabularx}{\linewidth}{@{\hspace{0pt}} >{$}l<{$}  @{\hspace{2mm}}  X @{}}
\hline
\textrm{Symbol}	& Definition 	\\
\hline
Q               & full conjunctive query $Q(\vec x) \datarule R_1(\vec x_1),\ldots, R_\stages(\vec x_\stages)$ \\
Q^B             & Boolean version of $Q$: $Q^B \datarule Q(\vec x) $ \\
\boldsymbol{\mathcal{Q}} & set of acyclic queries $\{ Q \}$ \\
\stages			& number of atoms in $Q$ or non-unary stages in DP\\
m				& number of variables in $Q$ \\
D               & input database \\
n				& number tuples in largest input relation in $D$ \\
w(r_i) \in W    & weight of input tuple $r_i$ \\
w(r) \in W      & weight of result tuple $r$ with witness $(r_1,\ldots, r_\stages)$: $w(r) = w(r_1) \otimes \cdots \otimes w(r_\stages)$ \\
\mathrm{out}    & set of output tuples for $Q$ on $D$ \\
s               & state in DP problem\\
s_0, t          & start and terminal state of DP problem \\
\Sset_\sgiter	& set of states in stage $\sgiter$ \\
\Sset           & set of all states: $\Sset = \bigcup_\sgiter \Sset_\sgiter$ \\
\SsetR_\sgiter	& set of states in stage $\sgiter$ after bottom-up phase\\
\SsetR           & set of all states after bottom-up phase: $\SsetR = \bigcup_\sgiter \SsetR_\sgiter$ \\
\Dec            & set of possible decisions: $\Dec \subseteq S \times S$ \\
\DecR           & set of possible decisions after bottom-up phase\\
\sol(s)			& solution from state $s$ \\
\sol			& complete solution: $\sol = \sol(s_0)$ \\
\sol_\prank(s) 	& $\prank$-th best solution from state $s$ \\
\solW_\prank(s) & cost or weight of $\prank$-th best solution from state $s$\\
\langle s_1 \ldots s_\sglim \rangle & solution prefix of length $\sglim \in \N_1^\stages$ \\
\langle s_\sglim \ldots t \rangle & solution suffix starting at state $s_\sglim$ \\
				& %
				\\
\weight(s, s')	& cost/weight of transitioning from state $s$ to $s'$ \\
\solW(s_r)		& cost/weight of a solution/path suffix 
				$\langle s_r s_{r+1} s_{r+2} \ldots s_\stages t \rangle$: $\solW(s_r) 
				= w(s_r,s_{r+1}) \otimes w(s_{r+1},s_{r+2}) \otimes\cdots\otimes w(s_\stages, t)$ \\
\Choices_j(s)   & set of suffixes starting at $s$ from which the $j$-th best suffix is selected \\
\Cand           & Set of prefixes containing the prefix of the next result to be returned\\
\Suc(s, s')     & successors of $s'$ at state $s$ that are considered by the Lawler-based partitioning approaches \RPDP \\
\subw           & submodular width of a query \\
\hline
\end{tabularx}
\end{table}

\section{Additional information on the queries used for the experiments}
\label{sec:sql_queries}

For completeness, we list the SQL code for the queries we used in our experiments.
Each relation $R_i, i \geq 1$ has two attributes $A_1$ and $A_2$ and also an additional attribute $W$ that contains the tuple weight.
For the real datasets that correspond to networks, all relations $R_i$ correspond to the same relation $\mathrm{EDGES}(\mathrm{fromNode}, \mathrm{toNode}, \mathrm{edgeWeight})$. 
In other words, $A_1$ corresponds to the source of a directed edge, $A_2$ corresponds to the destination and $W$ corresponds to the edge weight.

The 4-Path query is:
\begin{verbatim}
    SELECT   R1.A1, R2.A1, R3.A1, R4.A1, R4.A2
    FROM     R1, R2, R3, R4
    WHERE    R1.A2=R2.A1 AND R2.A2=R3.A1 AND 
             R3.A2=R4.A1
    ORDER BY R1.W + R2.W + R3.W + R4.W ASC
    LIMIT    k
\end{verbatim}

The 4-Star query is:
\begin{verbatim}
    SELECT   R1.A1, R2.A2, R3.A2, R4.A2, R4.A2
    FROM     R1, R2, R3, R4
    WHERE    R1.A1=R2.A1 AND R1.A1=R3.A1 AND 
             R1.A1=R4.A4
    ORDER BY R1.W + R2.W + R3.W + R4.W ASC
    LIMIT    k
\end{verbatim}

The 4-Cycle query is:
\begin{verbatim}
    SELECT   R1.A1, R2.A1, R3.A1, R4.A1
    FROM     R1, R2, R3, R4
    WHERE    R1.A2=R2.A1 AND R2.A2=R3.A1 AND 
             R3.A2=R4.A1 AND R4.A2=R1.A1
    ORDER BY R1.W + R2.W + R3.W + R4.W ASC
    LIMIT    k
\end{verbatim}

The shorter or longer queries are similar and thus omitted.

\end{document}